\documentclass[11pt]{article}
\usepackage{amsmath,amssymb,amsfonts,amsthm,epsfig}
\usepackage[usenames,dvipsnames]{xcolor}
\usepackage{bm,xspace}
\usepackage{tcolorbox}
\usepackage{cancel}
\usepackage{fullpage}
\usepackage{guy}
\usepackage{framed}
\usepackage{verbatim}
\usepackage{enumitem}
\usepackage{array}
\usepackage{multirow}
\usepackage{afterpage}
\usepackage{mathrsfs}
\usepackage{pifont} 
\usepackage{chngpage}
\usepackage[normalem]{ulem}
\usepackage{boxedminipage}
\usepackage{caption}
\usepackage{forest}
\usepackage{svg}
\usepackage{microtype}
\usepackage{setspace}

\usepackage{algorithm}
\usepackage{algorithmicx}
\usepackage{algpseudocode}

\usepackage{pgfplots}
\pgfplotsset{width=8cm,compat=newest}
\usepgfplotslibrary{fillbetween}

\def\colorful{1}

\ifnum\colorful=1

\fi
\ifnum\colorful=0

\fi

\makeatletter
\newcommand{\subalign}[1]{%
  \vcenter{%
    \Let@ \restore@math@cr \default@tag
    \baselineskip\fontdimen10 \scriptfont\tw@
    \advance\baselineskip\fontdimen12 \scriptfont\tw@
    \lineskip\thr@@\fontdimen8 \scriptfont\thr@@
    \lineskiplimit\lineskip
    \ialign{\hfil$\m@th\scriptstyle##$&$\m@th\scriptstyle{}##$\hfil\crcr
      #1\crcr
    }%
  }%
}
\makeatother

\newcommand{\Conv}{\mathrm{Conv}}
\newcommand{\Maj}{\textsc{Maj}}
\newcommand{\Ber}{\mathrm{Ber}}
\newcommand{\Unif}{\mathrm{Unif}}

\newcommand{\Learn}{\textsc{Learn}}
\newcommand{\FW}{\textsc{FrankWolfe}}
\newcommand{\error}{\mathrm{error}}

\DeclareMathOperator*{\argmax}{arg\,max}
\DeclareMathOperator*{\argmin}{arg\,min}

\newlist{enumprop}{enumerate}{1} %
\setlist[enumprop]{label=\arabic*.,ref=\theproposition.\arabic*}
\crefalias{enumpropi}{proposition}

\newcommand{\citet}{\cite}
\newcommand{\citep}{\cite}

\newcommand{\Bin}{\mathrm{Bin}}
\newcommand{\Rad}{\mathrm{Rad}}
\newcommand{\VC}{\mathrm{VCdim}}
\newcommand{\corr}{\mathrm{corrupted}}
\newcommand{\mcDc}{\mcD^{(\corr)}}
\newcommand{\bmcDc}{\bmcD^{(\corr)}}

\newcommand{\rand}{\mathrm{random}}
\newcommand{\wrong}{\mathrm{wrong}}
\newcommand{\correct}{\mathrm{correct}}

\newcommand{\pparagraph}[1]{\bigskip \noindent {\bf {#1}}}

\begin{document}
\hypersetup{pageanchor=false}

\title{
Robust Learning with Optimal Error %
}

\author{ 
Guy Blanc \vspace{6pt} \\ 
\hspace{-7pt} {\sl Stanford} 
}

\maketitle

\begin{abstract}
    We construct algorithms with optimal error for learning with adversarial noise. The overarching theme of this work is that the use of \textsl{randomized} hypotheses can substantially improve upon the best error rates achievable with deterministic hypotheses. 
    \begin{enumerate}
        \item[$\circ$] For $\eta$-rate malicious noise, we show the optimal error is $\frac{1}{2} \cdot \eta/(1-\eta)$, improving on the optimal error of deterministic hypotheses by a factor of $1/2$. This answers an open question of Cesa-Bianchi et al. (JACM 1999) who showed randomness can improve error by a factor of $6/7$. 
        \item[$\circ$] For $\eta$-rate nasty noise, we show the optimal error is $\frac{3}{2} \cdot \eta$ for distribution-independent learners and $\eta$ for fixed-distribution learners, both improving upon the optimal $2 \eta$ error of deterministic hypotheses. This closes a gap first noted by Bshouty et al. (Theoretical Computer Science 2002) when they introduced nasty noise and reiterated in the recent works of Klivans et al. (NeurIPS 2025) and Blanc et al. (SODA 2026). 
        \item[$\circ$] For $\eta$-rate agnostic noise and the closely related nasty classification noise model, we show the optimal error is $\eta$, improving upon the optimal $2\eta$ error of deterministic hypotheses. 
    \end{enumerate}
    All of our learners have sample complexity linear in the VC-dimension of the concept class and polynomial in the inverse excess error. All except for the fixed-distribution nasty noise learner are time efficient given access to an oracle for empirical risk minimization. 

    Our techniques reveal close connections between optimal robustness and algorithmic fairness: We show that optimal error agnostic learning is equivalent to multiaccuracy and that optimal malicious and distribution-independent nasty noise learners can be constructed via a black-box post-processing of a calibrated multiaccurate learner, though with suboptimal sample and oracle complexity. Our techniques for improving these complexities may be of independent interest. 
\end{abstract}

 \thispagestyle{empty}
 \newpage 

 \thispagestyle{empty}
 \setcounter{tocdepth}{2}
 
\begin{spacing}{0.92} %
  \tableofcontents
 \end{spacing}

 \thispagestyle{empty}
 \newpage

 \hypersetup{pageanchor=true}
 \setcounter{page}{1}

\section{Introduction}

Valiant's original model of PAC learning~\cite{Val84} assumes that every element of the received dataset is correctly labeled by the target function, but the real world is rarely so kind. For this reason, there has been an intense and sustained effort to design \textsl{robust} learning algorithms that succeed even given less idealized datasets ~\cite{MP96,Aue97,Bsh98,CBNDFSS99,Ser03,AW98,KKMS08,KLS09,KK09,F10,LS11,BS12,ABL17,DKS18,SZ21,BLMT22,DK23book,Shen23,HKLM24,BV25,BHMS26}. These robust algorithms are formalized in various models of adversarial noise in which the learner is expected to succeed even when an adversary corrupts a portion of the dataset. Such models include~nasty noise~\citep{BEK02}, malicious noise~\cite{Val85}, and agnostic noise \cite{Hau92,KSS94}. This work resolves a foundational question for each of these models of robust learning.
\begin{quote}
\begin{center}
    \textsl{What is the optimal error of robust learners?}
\end{center}
\end{quote}
\noindent\textbf{Randomized vs.~deterministic hypotheses.} It is well-known that for many models of adversarial noise, empirical risk minimization (ERM) achieves optimal error when benchmarking only against \textsl{deterministic} hypotheses. For example, consider $\eta$-rate nasty noise, in which the adversary is allowed to arbitrarily corrupt an $\eta$-fraction of the data. In the paper introducing nasty noise, Bshouty, Eiron, and Kushilevitz show that ERM produces a hypothesis with $\approx 2\eta$ error and that no algorithm outputting a deterministic hypothesis can achieve better error~\cite{BEK02}.%

To prove this lower bound,~\citet{BEK02} first observe that if $f_0$ and $f_1$ are at distance $2\eta$, the randomized function $f_{\mathrm{mix}} = \frac{f_0 + f_1}{2}$ is only at distance $\eta$ from each of $f_0$ or $f_1$. Therefore, an $\eta$-rate adversary can create a dataset that has labels sampled from $f_{\mathrm{mix}}$ regardless of whether the true target is $f_0$ or $f_1$. The key step is to then force all of the $2\eta$ distance between $f_0$ and $f_1$ to come from them classifying a single point, $x^{\star}$, differently. %
Then, any \textsl{deterministic} hypothesis must commit to either positively or negatively classifying $x^{\star}$. Regardless of its decision, it will have $2\eta$ error on whichever of $f_0$ or $f_1$ it disagrees with.

This lower bound is unsatisfactory. First, observe the incongruence between the adversary generating samples from a randomized function, $f_{\mathrm{mix}}$, while the learner is forced to output a deterministic hypothesis. More generally, learning with adversarial noise is, as the naming suggests, learning in an inherently stochastic environment. In stochastic environments, randomized hypotheses, which can equivalently be viewed as outputting probabilities, are far more informative than deterministic hypotheses. For example, in weather prediction, such probabilities are more useful than a deterministic positive or negative prediction. (see also the excellent introduction of Kearns and Schapire~\citep{KS94} for many more examples in which randomized hypotheses are preferable).

Second, observe that this lower bound against deterministic hypotheses involves a single ``heavy" point, $x^{\star}$, that appears multiple times in the dataset. In practical settings, there is little reason to not simply memorize the labels of heavy points.  Our work shows that even if one wishes to output a deterministic hypothesis, heavy points are the only obstruction to improved error. This is because randomized hypotheses can easily be converted to deterministic hypotheses with comparable error whenever there are no heavy points (via randomized rounding, see \Cref{claim:randomized-rounding}). Hence, our results also improve the error rates of deterministic hypotheses whenever there are no heavy points.

\section{Our results}
\subsection{Malicious noise}
\label{subsec:malicious-intro}
We first state our optimal algorithm for learning with \textsl{malicious noise}, a model of robust learning first studied by Valiant \cite{Val85} only one year after he defined the PAC model. In $\eta$-malicious noise, if the true target is $c^\star$ and data distribution is $\mcD$, each point is created as follows:
\begin{enumerate}
    \item With probability $1 - \eta$, a clean point is sampled $(\bx, c^\star(\bx))_{\bx \sim \mcD}$.
    \item With probability $\eta$, the adversary may choose \textit{any} $(x', y')$.
\end{enumerate}
Hence, approximately $\eta$-fraction of the dataset is chosen by the adversary. It is well-known that the optimal error of any learner forced to output a deterministic hypothesis with $\eta$-malicious noise is $\approx \eta/(1-\eta)$ and that this is achieved by ERM~\cite{KL93}. Cesa-Bianchi et al.~\cite{CBNDFSS99} followed this with an investigation of learners allowed to output random hypotheses and showed:
\begin{enumerate}
    \item No learner can achieve error better than $\frac{1}{2} \cdot \frac{\eta}{1-\eta}$.
    \item With a massive sample (larger than the domain size), it is possible to achieve error approaching $\frac{1}{2} \cdot \frac{\eta}{1-\eta}$. Such a large sample allows for a memorization-based strategy, though their result still requires a clever probabilistic construction in addition to memorization.
    \item With a sample size linear in the VC dimension of the concept class, it is possible to achieve error approaching $\frac{6}{7} \cdot \frac{\eta}{1-\eta}$. As they remark, their learner requires solving a ``challenging combinatorial problem" and does not appear to be efficiently implementable even given an ERM oracle.
\end{enumerate}
They explicitly ask whether it is possible to give a learner with efficient sample complexity and error $\approx \frac{1}{2} \cdot \frac{\eta}{1-\eta}$ and note ``it seems [this question is] not easy to answer." We resolve it.
\begin{theorem}[Optimal learning with malicious noise, see \Cref{thm:malicious-body} for the full version]
    \label{thm:malicious-intro}
    For any concept class $\mcC$ with VC dimension $d$ and $\eps > 0$, there is an algorithm that learns $\mcC$ with $\eta$-malicious noise and error at most $\frac{1}{2} \cdot \frac{\eta}{1-\eta} + \eps$ using $O(\frac{d}{\eps^2})$ samples. Furthermore, this algorithm is efficiently implementable using $O(1/\eps)$ calls to an ERM oracle for $\mcC$.
\end{theorem}
This learner improves on that of~\cite{CBNDFSS99} both because it achieves optimal error with a reasonable sample size and because it is efficiently implementable using an ERM oracle. The full version, \Cref{thm:malicious-body}, furthermore shows that the sample complexity can be improved to $\tilde{O}(\frac{d}{\eps})$ in exchange for an error of $(\frac{1}{2} + \alpha) \cdot \frac{\eta}{1-\eta} + \eps$ where $\alpha$ is an arbitrarily small constant. This mirrors how \cite{CBNDFSS99} state their result.

\subsection{Nasty noise}
\label{subsec:nasty-intro}
In the nasty noise model \cite{BEK02}, also called strong contamination \cite{DKKLMS18}, first a clean dataset of i.i.d. $(\bx, c^\star(\bx))_{\bx \sim \mcD}$ points is drawn, and then the adversary is allowed to arbitrarily corrupt any $\approx \eta$-fraction of this dataset (see \Cref{def:nasty-noise-basic} for details). %
As mentioned in the introduction, ERM gives a learner with $2\eta$ error which is optimal for learners outputting deterministic hypotheses \cite{BEK02}. The best known lower bound for learners outputting randomized hypotheses is $\eta$, also shown in \cite{BEK02}. Resolving this gap is an explicit open question of recent work~\cite{KSTV25}. We show that, surprisingly, neither the existing upper nor lower bound is tight.
\begin{theorem}[Optimal distribution-independent learning with nasty noise, see \Cref{thm:nasty-dist-free-upper-body,thm:nasty-lb-body} for the full versions]
    \label{thm:nasty-dist-free-upper-intro}
    The same algorithm from \Cref{thm:malicious-intro} also achieves $\frac{3}{2}\cdot \eta + \eps$ error when learning $\mcC$ with $\eta$-nasty noise. Furthermore, this error is optimal even when the algorithm is given an unbounded number of samples.
    
\end{theorem}

Our lower bound in \Cref{thm:nasty-dist-free-upper-intro} crucially relies on the distribution-independent setting, in which a single learner is required to work for any distribution $\mcD$ over unlabeled examples. Many works aim to design learners that only succeed in the \textsl{fixed-distribution} setting, in which the learner knows $\mcD$ and may specialize to it (see e.g.~\cite{Kha93,LMN93} and many other works). We show in the fixed-distribution setting improved error is possible.
\begin{theorem}[Optimal fixed-distribution learning with nasty noise, see \Cref{thm:nasty-fixed-dist-body} for the full version]
    \label{thm:nasty-fixed-dist-intro}
      For any concept class $\mcC$ with VC dimension $d$, distribution $\mcD$, and $\eps > 0$, there is an algorithm that learns $\mcC$ on distribution $\mcD$ with $\eta$-nasty noise and error at most $\eta + \eps$ using $O(\frac{d}{\eps^2})$ samples.
\end{theorem}
Our learner in \Cref{thm:nasty-fixed-dist-intro} does not need a full description of $\mcD$. It suffices to receive an i.i.d. (unlabeled and uncorrupted) sample from $\mcD$ containing $\tilde{O}(d/\eps^4)$ points (see \Cref{remark:access-to-dist}). Note this is slightly more than the number of labeled samples the learner needs. \Cref{thm:nasty-fixed-dist-intro} is quite related to the study of \textsl{semi-supervised learning} (see e.g. \cite{BDLP08,BB10}), in which a small portion of labeled samples are provided in addition to a larger quantity of unlabeled samples. Typically these unlabeled samples are used because they are cheaper to produce than labeled samples, but a comparison of \Cref{thm:nasty-fixed-dist-intro,thm:nasty-dist-free-upper-intro} shows they can also improve \textsl{optimal error} in robust learning provided the unlabeled samples are trustworthy.

It is an enticing open problem to achieve an error rate of $\eta$ with an algorithm that can be efficiently implemented only given an ERM oracle to $\mcC$ (note that the learner of \Cref{thm:nasty-dist-free-upper-intro} achieves $1.5\eta$ error, so beating the optimal $2\eta$ error of deterministic hypotheses is possible). One challenge is that, in all the other settings, the canonical baseline we improve upon is simply running ERM on the received dataset. However, to achieve the improved error of \Cref{thm:nasty-fixed-dist-intro}, it is necessary to construct a $\mcD$-dependent algorithm, and it's not clear what a suitable $\mcD$-dependent baseline is.

\pparagraph{The relationship between malicious and nasty noise.} Recent work \cite{BHMS26} closely examined the relationship between malicious noise and nasty noise. They showed any distribution-independent learner for $\eta$-malicious noise with error $\tau$ can be upgraded to one for $\eta$-nasty noise with error $\approx \tau(1-\eta) + \eta$. Using their result, we can derive \Cref{thm:nasty-dist-free-upper-intro} black-box from our malicious noise learner, \Cref{thm:malicious-intro}, though their reduction would lead to a worse sample complexity and more complicated algorithm. Instead, we show the same learner from \Cref{thm:malicious-intro} suffices for \Cref{thm:nasty-dist-free-upper-intro}. %

\citet{BHMS26} also asked if their error overhead could be improved. Prior to our work, it was feasible that using random hypotheses, a learner with error $\tau$ for $\eta$-malicious noise could be converted to one with error $\approx \tau(1-\eta) + \eta/2$ for $\eta$-nasty noise. Our upper bound in \Cref{thm:malicious-intro} combined with our lower bound in \Cref{thm:nasty-dist-free-upper-intro} show the answer is no: their error overhead is optimal.

\subsection{Agnostic noise and nasty classification noise}
Nasty classification noise \cite{BEK02} is a restricted variant of nasty noise where the adversary may only corrupt labels; i.e. change the point $(x,y)$ to $(x, -y)$. This is closely related to agnostic noise \cite{Hau92,KSS94} in which the adversary chooses some (possibly randomized) function $\boldf$ which is promised to be $\eta$-close to the target $c^{\star}$, and the learner is given i.i.d. examples $(\bx, \boldf(\bx))_{\bx \sim \mcD}$. Nasty classification noise can only be harder than agnostic noise so we state \Cref{thm:agnostic-intro} using this harder model (though the two models were recently shown to be equivalent up to polynomial factors in the sample and time complexity \cite{BV25}). In both cases, the optimal error of learners outputting deterministic hypotheses is $2\eta$ \cite{BEK02}.
\begin{theorem}[Optimal learning with nasty classification noise, see \Cref{thm:agnostic-body} for the formal version]
    \label{thm:agnostic-intro}
    For any concept class $\mcC$ with VC dimension $d$ and $\eps > 0$, there is an algorithm that learns $\mcC$ with $\eta$-nasty classification noise and error at most $\eta+ \eps$ using $\tilde{O}(\frac{d}{\eps^4})$ samples. Furthermore, this algorithm is efficiently implementable using $O(1/\eps^2)$ calls to an ERM oracle for $\mcC$.
\end{theorem}

\begin{remark}[Two definitions of error for agnostic noise]
    \label{remark:agnostic-two-defs}
     There are two ways to define the error of a hypothesis $h$ with agnostic noise: One is to measure $\error(h, c^{\star})$ where $c^{\star} \in \mcC$ is the original uncorrupted function. With this measure, randomness is necessary for optimal error and \Cref{thm:agnostic-intro} is the first optimal error agnostic learner. A second common approach is to measure $\error(h, \boldf)$ where $\boldf$ is the \textsl{corrupted} function. In this second case, ERM already achieves the optimal error of $\approx \eta$.
\end{remark}

\subsection{Broader observations and connections}
\noindent \textbf{A close connection to algorithmic fairness.} Our techniques illuminate a connection between definitions from the algorithmic fairness literature and optimal error robust learners. We show that \textsl{multiaccuracy} \cite{KGZ19} is equivalent to a slight weakening of the learner given in \Cref{thm:agnostic-intro}, and any calibrated multiaccurate hypothesis \cite{GHKRW23,CGKR25} can be post-processed into one satisfying the requirements of \Cref{thm:malicious-intro,thm:nasty-dist-free-upper-intro}, albeit with suboptimal sample and oracle complexity. We elaborate on these connections in \Cref{sec:fairness}.

\pparagraph{The necessity of improper learners.} 
If we restrict to deterministic hypotheses, ERM is optimal for all of nasty, malicious, and agnostic noise. This implies that optimal error can be achieved by a \textsl{proper} learner: One that always outputs a hypothesis in the concept class $\mcC$. The analogous notion of properness for learners outputting a randomized hypothesis is to require that hypothesis to be a mixture of functions contained within $\mcC$. However, we show that such learners cannot achieve optimal error.

\begin{restatable}[Improper learners are necessary]{claim}{improper}
    \label{claim:improper}
    None of the error guarantees of \Cref{thm:malicious-intro,thm:nasty-dist-free-upper-intro,thm:nasty-fixed-dist-intro,thm:agnostic-intro} are achievable by algorithms that are forced to output hypotheses that are mixtures of functions within $\mcC$, even when the VC dimension of $\mcC$ is $1$, the excess error $\eps$ is $\Omega(1)$, and the algorithm receives an unbounded number of samples.
\end{restatable}

Our learners output hypotheses that are mixtures of functions within $\Maj_k(\mcC)$, i.e. can be formed by taking the majority of $k$ functions within $\mcC$. The learner of \Cref{thm:malicious-intro,thm:nasty-dist-free-upper-intro} uses a constant $k$ (we show $k = 7$ suffices), whereas those of \Cref{thm:nasty-fixed-dist-intro,thm:agnostic-intro} set $k = O(1/\eps^2)$.

\pparagraph{The necessity of distinct learners.} A further implication of the optimality of ERM is that, when restricting to deterministic hypotheses, there is one algorithm that is optimal for all of nasty, malicious, and agnostic noise.
In contrast, not all of our learners are the same. We show this is necessary.
\begin{restatable}[Distinct learners are necessary]{claim}{distinct}
    \label{claim:distinct}
    There is no one algorithm achieving the error guarantee of both \Cref{thm:malicious-intro} and \Cref{thm:agnostic-intro} even when the VC dimension of $\mcC$ is $1$, the excess error $\eps$ is $\Omega(1)$, and the algorithm is given an unbounded number of samples.
\end{restatable}
In spite of \Cref{claim:distinct}, it is possible to design a single learner that improves upon the optimal deterministic learner for all three noise models. In particular, the learner of \Cref{thm:malicious-intro} achieves optimal error for malicious noise and $\frac{3\eta}{2}$ error nasty noise (by \Cref{thm:nasty-dist-free-upper-intro}) which translates to $\frac{3\eta}{2}$ error in the easier agnostic/nasty classification noise setting. All of these improve upon what is possible with deterministic classifiers.

\section{Technical overview}
Our learners are designed through a unified framework. We first convert the desired accuracy guarantee to a linear loss function. Then, using the minimax lemma, show there exists a single hypothesis that has good loss for all possible target functions. Finally, for those algorithms efficiently implementable with an ERM oracle, we use a general optimization procedure \Cref{thm:general-optimization} to efficiently search for such a hypothesis.

In this section, we will describe two representative proofs. First, we describe our malicious noise learner (\Cref{thm:malicious-intro}) which illustrates these three steps cleanly. Second, we will describe our fixed-distribution nasty noise learner (\Cref{thm:nasty-fixed-dist-intro}) as that requires a particularly tricky analysis.

\subsection{Learning with malicious noise}
\label{subsec:mal-overview}
\noindent\textbf{Step 1: Define a linear loss function.}
It will typically be convenient to work with a real-valued predictor, $\bar{h}:X \to [-1,1]$. At the end, we output its randomization $\bh \coloneqq \Rad(\bar{h})$,
\begin{equation*}
    \bh(x) = \begin{cases}
    1&\text{with probability }\frac{1 + \bar{h}(x)}{2}\\
    -1&\text{otherwise.}
    \end{cases}
\end{equation*}
Given a dataset $S$, we first define a loss function $\ell_S(c,\bar{h})$ such that if $\ell_S(c, \bar{h})$ is small, then $\bh$ will have the desired accuracy in the case where the target function $c^{\star}$ is $c$.

In the case of malicious noise, we receive a dataset $\bS = ((\bx_1, \by_1),...,(\bx_n, \by_n))$. If the noise rate were $\eta$, then there is an uncorrupted index set $I \subseteq [n]$ of size roughly $(1-\eta)n$, where $\bS_I$ are i.i.d. samples of the form $(\bx, c^{\star}(\bx))_{\bx \sim \mcD}$. A standard generalization argument gives that
\begin{equation}
    \label{eq:mal-generalization}
    \Prx_{\bx \sim \mcD, \bh}[\bh(\bx) \neq c^{\star}(\bx)] \approx \Prx_{\bi \sim \Unif(I), \bh}[\bh(\bx_{\bi}) \neq \by_{\bi}].
\end{equation}
Therefore, our goal is to learn a hypothesis $\bh$ satisfying
\begin{equation}
    \label{eq:mal-desired-accuracy}
    \Prx_{\bi \sim \Unif(I), \bh}[\bh(\bx_{\bi}) \neq \by_{\bi}] \leq \frac{\eta}{2(1-\eta)}.
\end{equation}
The challenge in establishing the above inequality is that we do not know what the uncorrupted indices $I$ are. Still, we show there is a loss that is only small if $\bh$ simultaneously has small error on \textsl{all} $I$ that are consistent with some $c \in \mcC$.
\begin{claim}[Malicious error can be linearized]
    \label{lem:mal-error-loss-overview}
    For any dataset $S \in (X \times \bits)^n$, let $\bar{h}:X \to [-1,1]$ be any function satisfying, for all $c \in \mcC$,
    \begin{equation}
        \label{eq:mal-loss-intro}
        \ell_S(c, \bar{h}) \leq 0 \quad\quad\text{where } \ell_S(c, \bar{h}) \coloneqq\Ex_{\bx,\by \sim \Unif(S)}\bracket*{c(\bx)\by \cdot \paren[\big]{2 - \bar{h}(\bx)\by} -\bar{h}(\bx)\by}.
    \end{equation}
    Then, the randomized function $\bh = \Rad(\bar{h})$ satisfies \Cref{eq:mal-desired-accuracy} simultaneously for all $\eta$ and $I$ of size $(1-\eta)n$ for which there exists some $c \in \mcC$ that is fully consistent with $S_I$.
\end{claim}

\pparagraph{Step 2: Use the minimax lemma to establish the existence of a low complexity hypothesis with small loss.}
Next, to ensure that the generalization step of \Cref{eq:mal-generalization} holds with a small sample size, we want to establish the existence of ``low complexity" $\bar{h}$ with small loss. In the case of malicious noise, we will always choose a $\bar{h}$ that is a convex combination of functions in $\Maj_7(\mcC)$, each formed by taking the majority of at most $7$ functions within $\mcC$. This class has the same VC dimension as $\mcC$ up to constants (see \Cref{fact:maj-VC}), and so $\approx \VC(\mcC)$ samples suffice. To establish the existence of such an $\bar{h}$, we apply the minimax lemma:
\begin{equation*}
    \min_{\bar{h} \in \Conv(\Maj_k(\mcC))} \set[\Big]{\max_{c \in \mcC} \set*{\ell_S(c,\bar{h})}}= \max_{\text{distribution }\mu \text{ on } \mcC}\set[\Big]{ \min_{\bar{h} \in \Conv(\Maj_k(\mcC))} \set[\Big]{\Ex_{\bc \sim \mu}\bracket*{\ell_S(\bc, \bar{h})}}}.
\end{equation*}
Hence, our task becomes to show that for any distribution $\mu$ over concepts, there is a way to choose $\bar{h}$ so that the expectation of $\ell_S(\bc, \bar{h})$ over $\bc \sim \mu$ is nonpositive. On a technical level, this makes our job substantially easier since we can choose $\bar{h}$ as a function of $\mu$. Indeed, the following becomes quite easy to show.
\begin{claim}[A good hypothesis exists for any distribution of concepts]
    \label{lem:intro-g-for-dist}
    Let $g:[-1,1] \to [-1,1]$ be any function satisfying, for all $t \in [-1,1]$,
    \begin{equation}
        \label{eq:g-constraints-mal-intro}
        \frac{2t}{1+t} \leq g(t) \leq \frac{2t}{1-t}.
    \end{equation}
    Then, for any distribution $\mu$ over functions and dataset $S$
    \begin{equation*}
        \Ex_{\bc \sim \mu}\bracket*{\ell_S(\bc, \bar{h})} \leq 0 \quad\quad\text{where}\quad \bar{h}(x) = g\paren*{\Ex_{\bc \sim \mu}[\bc(x)]}.
    \end{equation*}
\end{claim}
Finally, we need to pick $g$ satisfying the constraints of \Cref{eq:g-constraints-mal-intro} so that, for any distribution $\mu$ over $\mcC$, the function $g\paren*{\Ex_{\bc \sim \mu}[\bc(x)]}$ is in $\Conv(\Maj_k(\mcC))$. We show this is possible for $k = 7$ (see also \Cref{fig:g-constraints-mal-intro} for a visualization of these constraints and our specific choice of $g$).

At this point, we have a sample efficient but time inefficient learner: Brute force search for any $\bar{h} \in \Conv(\Maj_k(\mcC))$ satisfying $\max_{c \in \mcC} \set*{\ell_S(c,\bar{h})} \leq 0$ and output $\bh = \Rad(\bar{h})$.

\pparagraph{Step 3: Using an ERM oracle to make the learner efficient.}
We first observe that for any candidate hypothesis $\bar{h}$, an ERM oracle allows us to easily estimate $\max_{c \in \mcC} \ell_S(c, \bar{h})$. This is because the $c$ with maximal loss is that which minimizes a \textsl{weighted} empirical risk minimization problem,
\begin{equation*}
    \argmax_{c \in \mcC}  \ell_S(c, \bar{h}) = \argmin_{c \in \mcC} \sum_{i \in [n]} w_i \cdot \Ind[c(x_i) \neq y_i] \quad\quad\text{where }w_i = \overline{h}(x_i)y_i - 2.
\end{equation*}
It is quite straightforward to convert a weighted ERM problem to a standard unweighted problem \Cref{fact:weighted-ERM}. Therefore, we have an efficient procedure to check whether a candidate hypothesis $\bar{h}$ already has small loss, and if it does not, return a particular $c \in \mcC$ for which it has large loss.

Given this, we show an extremely simple iterative algorithm (see  \Cref{fig:MalLearner}) successfully finds a hypothesis with small loss. This algorithm maintains a distribution $\mu$ over $\mcC$ and hypothesis $\bar{h} = g(\mu)$ (where $g$ is as in \Cref{lem:intro-g-for-dist}). At each step, if there is some $c$ for which $\ell_S(c, \bar{h})$ is large, it updates $\mu$ in the direction of $c$.

\begin{figure}[H]
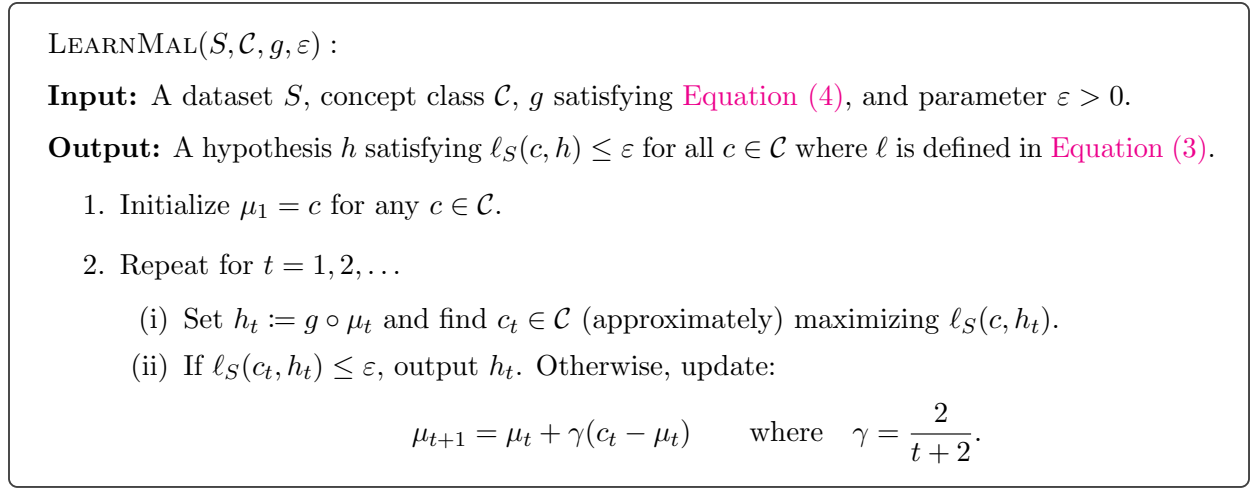
 
  \captionsetup{width=.9\linewidth}
    
    \begin{tcolorbox}[colback = white,arc=1mm, boxrule=0.25mm]
    \vspace{2pt} 
    \textsc{LearnMal}$(S, \mcC, g, \eps):$\vspace{6pt}
    
    \textbf{Input:} A dataset $S$, concept class $\mcC$, $g$ satisfying \Cref{eq:g-constraints-mal-intro}, and parameter $\eps>0$.\vspace{6pt}

    \textbf{Output:} A hypothesis $h$ satisfying  $\ell_S(c, h) \leq \eps$ for all $c \in \mcC$ where $\ell$ is defined in \Cref{eq:mal-loss-intro}.\vspace{-4pt}

    \begin{enumerate}
            \item Initialize $\mu_1 = c$ for any $c \in \mcC$.
            \item Repeat for $t = 1,2,\ldots$ \vspace{-2pt}
            \begin{enumerate}
                \item[(i)] Set $h_t \coloneqq g \circ \mu_t$ and find $c_t \in \mcC$ (approximately) maximizing $\ell_S(c, h_t)$.
                \item[(ii)] If $\ell_S(c_t, h_t) \leq \eps$, output $h_t$. Otherwise, update: \vspace{-4pt}
                \begin{equation*}
                    \mu_{t+1} = \mu_t + \gamma(c_t - \mu_t) \quad\quad\text{where}\quad\gamma = \frac{2}{t+2}.
                \end{equation*}
            \end{enumerate}
    \end{enumerate}
    \end{tcolorbox}
\caption{Our malicious noise learner.}
\label{fig:MalLearner}
\end{figure}

A priori, it is not even clear that there should exist a single choice of $\mu^*$ for which $\ell(c, g \circ \mu^*)$ is small for \textsl{all} $c \in \mcC$: The minimax lemma in the preceding step only establishes there exists, roughly speaking, a distribution over $\bmu$ so that $\overline{h} \coloneqq \Ex_{\bmu}[g \circ \bmu]$ that has small loss. Yet, the learner in \Cref{fig:MalLearner} successfully converges to such a single $\mu^{\star}$.
\begin{theorem}[Special case of \Cref{thm:general-optimization}]
    \label{thm:mal-opt-intro}
    The algorithm in \Cref{fig:MalLearner} converges in $O(1/\eps)$ iterations.
\end{theorem}

\Cref{thm:mal-opt-intro} is a special case of a more general theorem, which we also use to prove \Cref{thm:agnostic-intro} and may be of independent interest. We consider loss functions $\ell_S(c,h)$ which can be written as
\begin{equation*}
    \ell_S(c,h) = \Ex_{\bx, \by \sim \Unif(S)}[f(c(x), h(x), y)]
\end{equation*}
for some multilinear function $f$. We show the same simple algorithm converges provided $f$ meets two mild constraints: The first is that we need an analogue of \Cref{lem:intro-g-for-dist} to hold (roughly speaking, that the minimax lemma can establish the existence of a good hypothesis). Second, we need the loss function ``prefer" (i.e. be smaller) when $c$ and $h$ agree than when they disagree.

\subsection{Fixed-distribution nasty noise, overview of \Cref{thm:nasty-fixed-dist-intro}}
\label{subsec:overview-fixed-nasty}
Our most technically challenging theorem to prove is that it is possible to learn with $\eta$-error given a sample corrupted by the nasty adversary, provided the learner knows the distribution $\mcD$ over unlabeled examples. In this setting, first a clean sample is generated $\bS^{\circ} \iid (\bx, c^{\star}(\bx))_{\bx \sim \mcD}$ and then an adversary corrupts approximately $\eta$-fraction of the points to produce $\bS$, which is given to the learner.

The key step is to appropriately define a quantity $\eta(c, S, \mcD)$ which estimates how many corruptions the adversary would need to make to corrupt a ``typical" clean sample $\bS^{\circ}$ into the corrupted sample $\bS = S$. When defining this estimate, there are two criteria to balance.
\begin{enumerate}
    \item The estimates should be large enough so that, given a corrupted dataset $\bS$ and base distribution $\mcD$, the learner can find a single hypothesis $\bh$ for which
    \begin{equation}
        \label{eq:large-enough-est}
        \Prx_{\bx \sim \mcD, \bh}[\bh(\bx) \neq c(\bx)] \leq \eta(c, \bS, \mcD)\quad\quad\text{for all }c \in \mcC.
    \end{equation}
    \item The estimates cannot be (much) larger than the number of corruptions the adversary made. In particular, for the true target $c^{\star}$, it should be the case that with high probability over the randomness of the sample and the adversary's decisions
    \begin{equation*}
        \eta(c^{\star}, \bS, \mcD) \leq \eta + \eps.
    \end{equation*}
\end{enumerate}
Our method for constructing these estimates will combine two natural methods for certifying that the adversary needed to make a large number of corruptions to produce a given sample. First, we observe that if $S$ disagrees with $c$ on $\eta$-fraction of the points, then the adversary would have needed to corrupt $\eta$-fraction of any $\bS^{\circ}$ consistent with $c$ to form $S$. This means we can take
\begin{equation*}
    \eta^{(\mathrm{labels})}(c, S, \mcD) = \Prx_{(\bx, \by) \sim \Unif(S)}[c(\bx) \neq \by].
\end{equation*}
It turns out the above estimate will not be large enough to ensure \Cref{eq:large-enough-est} because it does not depend on the distribution $\mcD$. One way to take advantage of the learner's knowledge $\mcD$ is to use it to detect \textsl{distribution shift} in the unlabeled portion of $S$. In particular, for any ``low complexity" detection function $f$, it was likely that the adversary had to change roughly
\begin{equation*}
    \eta^{(\mathrm{shift)}}(c,S,\mcD) = \Prx_{(\bx, \by) \sim \Unif(S)}[f(\bx) = 1] - \Prx_{\bx \sim \mcD}[f(\bx) = 1]
\end{equation*}
fraction of the points within $\bS^{\circ}$ to produce $S$. This method, unlike the first, does take advantage of $\mcD$, but will still not be large enough to ensure \Cref{eq:large-enough-est} holds because it does not depend on $c$. Our final approach will combine these two approaches. We fix some detection class $\mcF$ and define
\begin{align*}
    \eta(c, S, \mcD) =  \sup_{f \in \mcF}&\set*{\eta(c, f, S, \mcD)} \quad\quad\mathrm{where}\\
    &\quad\eta(c, f, S, \mcD) \coloneqq \Prx_{(\bx, \by) \sim \Unif(S)}[c(\bx) \neq \by \text{ }\mathrm{or}\text{ } f(\bx) = 1] - \Prx_{\bx \sim \mcD}[f(\bx) = 1].
\end{align*}
The proof of \Cref{thm:nasty-fixed-dist-intro} then follows the same two steps we previously outlined. We first show in \Cref{claim:lb-corruption-formal} that, provided the class $\mcF$ is low complexity, these estimates will not be too large (i.e. $ \eta(c^{\star}, \bS, \mcD) \leq \eta + \eps$). The more challenging piece is to instantiate $\mcF$ so that the estimates are large enough, meaning \Cref{eq:large-enough-est} holds. For this second step, we once again use the minimax lemma to show it suffices, for any distribution $\mu$ over concepts, to choose a hypothesis $h_{\mu}$ and detection function $f_{c,\mu}$ for which
\begin{equation*}
    \Ex_{\bc \sim \mu}\bracket*{\Prx_{\bx \sim \mcD}[h_{\mu}(\bx) \neq \bc(\bx)]} \leq \Ex_{\bc \sim \mu}\bracket*{\eta(\bc, f_{\bc,\mu},\bS, \mcD)}.
\end{equation*}
This is analogous to \Cref{lem:intro-g-for-dist} but more complicated. Now rather than only choosing the hypothesis $h_{\mu}$, we must also simultaneously choose detection functions $f_{c, \mu}$ for all $c$ in the support of $\mu$. We elaborate on how to choose these functions in \Cref{subsec:proof-of-nasty-fixed-avg}. For now, we just mention that there is a sense in which the definitions of $f_{c, \mu}$ and $h_{\mu}$ are forced once we set up the above equation properly. This implies that our method for lower bounding corruptions is exactly expressive enough to prove the desired error bound.

\section{Preliminaries and notation}

\noindent\textbf{Basic notation.} We use $[n]$ as shorthand for $\set{1,2,\ldots, n}$. For $S = (z_1, \ldots, z_n)$ and indexing set $I = (i_1, \ldots, i_m) \subseteq [n]$, we use $S_I$ as shorthand for $(S_{i_1}, \ldots, S_{i_m})$. We frequently use the function $\sign:\R \to \bits$ defined as
\begin{equation*}
    \sign(x) \coloneqq \begin{cases}
        1&\text{if $x \geq 0$}\\
        -1&\text{otherwise.}
    \end{cases}
\end{equation*}

\pparagraph{Probability.} Throughout this paper, we use {\bf bold font} to denote random variables. To denote that $\bx_1, \ldots,\bx_n$ independent and identically distributed according to $\mcD$, we will interchangeably use either $\bx_1, \ldots, \bx_n \iid \mcD$ or $\bx \sim \mcD^n$. For a distribution $\mcD$, we will use $\mcD(x)$ as shorthand for $\Pr_{\bx \sim \mcD}[\bx=x]$. For a (multi)set $S$, we use $\Unif(S)$ to denote the uniform distribution over $S$. 

We will frequently invoke three basic probability distributions. The \textsl{Rademacher} distribution with mean $\mu \in [-1,1]$, denoted $\Rad(\mu)$ takes on $+1$ with probability $\frac{1 + \mu}{2}$ and $-1$ with probability $\frac{1 - \mu}{2}$. The $\textsl{Bernoulli}$ distribution with parameter $p \in [0,1]$, denoted $\Ber(p)$, takes on $1$ with probability $p$ and $0$ otherwise. The \textsl{Binomial} distribution with parameters $n \in \N$, $p \in [0,1]$, denoted $\Bin(n,p)$, is the sum of $n$ i.i.d. $\Ber(p)$ random variables.

We will also use the classic Chernoff bound.
\begin{fact}[Chernoff bound]
    \label{fact:Chernoff}
    Let $\bz_1, \ldots, \bz_n$ be independent random variables each bounded on $[0,1]$. For any parameter $\Delta$ and $\bZ \coloneqq \bz_1 + \cdots +\bz_n$,
    \begin{equation*}
        \Prx[\bZ \geq (1 + \Delta)\mu]\leq \exp\paren*{\frac{-\Delta^2 \mu}{2 + \Delta}} \quad\quad\text{where } \mu \coloneqq \Ex[\bZ].
    \end{equation*}
\end{fact}
We'll often apply it in the following regime.
\begin{corollary}
    \label{cor:Chernoff} For any $p,\alpha \in [0,1], \eps \in (0,1]$, taking
    \begin{equation*}
        n \coloneqq O\paren*{\frac{\log (1/\delta)}{\eps(\alpha + \eps)}}
    \end{equation*}
    suffices to ensure that
    \begin{equation*}
        \Prx_{\bZ \sim \Bin(n,p)}[\bZ \geq n\cdot ((1 + \alpha)p + \eps)] \leq \delta.
    \end{equation*}
\end{corollary}
We will only need to define total variation distance for distributions with finite domain, since it is only used in our lower bound construction (which is finite).
\begin{definition}[Total variation distance]
    For any two distributions $\mcD_1, \mcD_2$ over the same finite domain $\Omega$,
    \begin{equation*}
        \dtv(\mcD_1, \mcD_2) \coloneqq \sum_{x \in \Omega} \max\set*{0, \mcD_1(x) - \mcD_2(x)} \quad\quad\text{where } \mcD(x) \coloneqq \Prx_{\bx \sim \mcD}[\bx=x].
    \end{equation*}
\end{definition}

\pparagraph{Learning models.} We work with standard binary classification and always use $\bits$ labels (our results are invariant to any renaming of these labels, e.g. to $\zo$). We begin by recalling the definition of standard (noise-free) PAC learning.%
\begin{definition} [Noise-free PAC learning] \label{def:PAC}
An \emph{$(\eps,\delta)$-distribution-independent PAC learning algorithm} $A$ for a \emph{concept class} ${\mcC}$ (a class of functions from $X$ to $\bits$) \emph{that uses $n$ examples} is a learning algorithm with the following property:
For any (unknown) \emph{target concept} $c^{\star} \in {\mcC}$ and (unknown) distribution $\mcD$ over $X$, if $A$ receives as input a dataset $\bS$ consisting of $n$ i.i.d. labeled examples $(\bx, c^{\star}(\bx))_{\bx \sim \mcD}$, then with probability at least $1-\delta$ the hypothesis $\bh$ that $A$ returns has \emph{error rate} at most $\eps$, i.e.~
\[
\error_{\mcD}(\bh, c^\star) \leq \eps, \quad \text{where} \quad
\error_{\mcD}(\bh, c^\star):= \Prx_{\bx \sim \mcD,\bh}[\bh(\bx) \neq c^{\star}(\bx)].
\]
If a learning algorithm $A$ only satisfies the above guarantee for some particular distribution ${\mcD}$, then we say that $A$ is a \emph{fixed-distribution $(\eps,\delta)$-PAC learning algorithm for distribution ${\mcD}$.}
\end{definition}

We will default to the more challenging distribution-independent PAC model unless otherwise mentioned. Only \Cref{thm:nasty-fixed-dist-intro} specializes to the fixed-distribution model.

The many definitions of adversarial noise are all modifications of the above definition in which there is an adversary with some control of the dataset $\bS$ the learner receives.

\begin{definition}[Malicious noise, \cite{Val85}]
    \label{def:malicious-basic}
    When learning with $\eta$-malicious noise, the sample $\bS \coloneqq ((\bx_1, \by_1), \ldots, (\bx_n, \by_n))$ is generated in a sequential fashion. For each $i \in [n]$,
    \begin{enumerate}
        \item With probability $1-\eta$, a clean sample $\bx_i \sim \mcD$, $\by_i = c^{\star}(\bx_i)$ is sampled independently of all prior points.
        \item With probability $\eta$, the adversary may set $(\bx_i, \by_i)$ to \textsl{any} point $(x', y')$.
    \end{enumerate}
    The learner then receives $\bS$ and is not told which were adversarially chosen.
\end{definition}
Nasty noise is similar to malicious noise. The key distinction is that the adversary may choose which points it wishes to corrupt.
\begin{definition}[Nasty noise, \cite{BEK02}]
    \label{def:nasty-noise-basic}
    When learning with $\eta$-nasty noise, the input to the learning algorithm is generated via the following process.
    \begin{enumerate}
        \item A clean sample $\bS^{\circ} = \set{(\bx_1,\by_1),\ldots, (\bx_n, \by_n)}$ is generated where each $\bx_i \iid \mcD$ and each $\by_i = c^{\star}(\bx_i)$.
        \item After observing this sample, the adversary chooses some subset $\bJ \subseteq [n]$ of indices to corrupt with the requirement that the marginal distribution of $|\bJ|$ is $\Bin(n, \eta)$.
        \item The adversary chooses a corrupted dataset $\bS$ by arbitrarily changing $\bS^{\circ}_j$ for all $j \in \bJ$ (i.e. the only requirement is that $\bS_i = \bS^{\circ}_i$ for any $i \notin \bJ$). The learner then receives $\bS$.
    \end{enumerate}
\end{definition}

Nasty classification noise is similar to nasty noise except the adversary may only corrupt labels.
\begin{definition}[Nasty classification noise \cite{BEK02}]
    \label{def:nasty-classification-noise}
    Learning with $\eta$-nasty classification noise is identical to $\eta$-nasty noise (\Cref{def:nasty-noise-basic}) except, for $\bS = ((\bx_1',\by_1'), \ldots, (\bx_n', \by_n'))$ it must be the case that $\bx_i' = \bx_i$ for all indices $i$, even those falling within $\bJ$.
\end{definition}
Agnostic noise is similar to nasty classification noise in that the adversary can only corrupt labels, but the adversary must commit in advance to how it will corrupt those labels (note that we still measure accurate w.r.t. $c^{\star}$, not $\boldf$, see also \Cref{remark:agnostic-two-defs}).
\begin{definition}[Agnostic noise, \cite{Hau92,KSS94}]
    \label{def:agnostic-noise}
    In $\eta$-agnostic noise, the adversary chooses a (possibly randomized) function $\boldf$ satisfying that $\error_{\mcD}(\boldf, c^{\star}) \leq \eta$. The learner then receives i.i.d. examples $(\bx,\by)$ where $\bx \sim \mcD$ and $\by =\boldf(\bx)$.
\end{definition}
Any learner for nasty classification noise is automatically a learner for agnostic noise with the exact same parameters. A partial converse was recently shown, with a small increase in the failure probability $\delta$ as well as a polynomially worse sample size \cite{BV25}.

Many of our learners will be efficient given access to an ERM oracle. For simplicity, we write the below as maximizing the agreement on a dataset. Note this is equivalent to minimizing disagreement on a dataset with all the labels flipped.
\begin{definition}[ERM oracle]
    An \textsl{ERM oracle} for concept class $\mcC$ takes as input a dataset $S = \set{(x_1,y_1),\ldots, (x_n,y_n)}$ and returns any $c \in \mcC$ maximizing $\sum_{i \in [n]} c(x_i)  y_i$. In the case of ties, the ERM oracle may return an arbitrary maximizer.
\end{definition}
We will use a folklore conversion of standard ERM oracles to weighted ERM oracles. 
\begin{fact}[Weighted ERM using unweighted ERM]
    \label{fact:weighted-ERM}
    For any $\eps, \delta > 0$, there is an algorithm that takes as input a data set $S = \set{(x_1,y_1),\ldots, (x_n,y_n)}$ and weights $w_1, \ldots, w_n \in [-b,b]$, runs in time $\poly(b,1/\eps, \log(1/\delta))$, uses one call to an ERM oracle for $\mcC$, and with probability at least $1-\delta$, outputs some $c \in \mcC$ satisfying,
    \begin{equation*}
        \frac{1}{n} \sum_{i \in [n]} w_i c(x_i) y_i \geq \max_{c^{\star} \in \mcC}\set*{ \frac{1}{n} \sum_{i \in [n]} w_i c(x_i) y_i  } - \eps.
    \end{equation*}
\end{fact}
The algorithm in \Cref{fact:weighted-ERM} generates a synthetic dataset consisting of $m \coloneqq \poly(b, 1/\eps, \log(1/\delta))$ independent samples from $S$ where the point $(x_i, \sign(w_i)\cdot y_i)$ is sampled with probability proportional to $|w_i|$. It then runs unweighted ERM on this synthetic dataset.

\pparagraph{VC dimension and Rademacher complexity.}
We will work with the standard notion of \textsl{VC dimension}.
\begin{definition}[VC dimension]
    For any concept class $\mcC$ of functions mapping $X$ to $\bits$, $\VC(\mcC)$ is defined to be the cardinality of the largest set $S$ on which $\mcC$ takes on all $2^{|S|}$ labelings.
\end{definition}
VC dimension is well known to bound generalization error. 
To get the tightest sample complexities, we will use a multiplicative form of the generalization bound.
\begin{fact}[Generalization error from VC dimension, multiplicative form, \cite{LLS01}]%
    \label{fact:VC-multiplicative}
    For any target function $f:X \to \bits$, distribution $\mcD$, concept class $\mcC$, and parameters $\eps,\delta \in (0,1)$ and $\alpha \in [0,1]$, set
    \begin{equation*}
        n = O\paren*{\frac{\VC(\mcC)\cdot \log(2 +\alpha/\eps) + \log(1/\delta)}{\eps (\alpha+\eps)}}.
    \end{equation*}
    Then, with probability at least $1-\delta$ over $\bS \sim \mcD^n$, all $c \in \mcC$ satisfy
    \begin{equation*}
        \error_{\mcD}(c,f) \leq (1+\alpha)\cdot \Pr_{\bx \sim \Unif(\bS)}[f(\bx) \neq c(\bx)] + \eps.
    \end{equation*}
\end{fact}
The above generalization bound can be easily extended to any function that is a mixture over concepts in $\mcC$ by taking an expectation of both sides.
\begin{corollary}[Generalization bound for randomized hypothesis]
    \label{cor:VC-random}
    Under the conditions of \Cref{fact:VC-multiplicative}, we also have that, with probability $1-\delta$ over $\bS$, for any randomized function $\bh$ that is a mixture of functions within $\mcC$,
    \begin{equation*}
           \error_{\mcD}(\bh,f)\leq (1 + \alpha) \cdot \Prx_{\bx \sim \Unif(\bS),\bh}[f(\bx) \neq \bh(\bx)] + \eps.
    \end{equation*}
\end{corollary}

We will also use two straightforward manipulations of VC dimension. The first is about majority expansion:
\begin{fact}[The VC dimension of majority expansion, \cite{FS97}]
    \label{fact:maj-VC}
    For any concept class $\mcC$ with VC dimension $d$, the VC dimension of $\Maj_k(\mcC)$ is at most $O(d \cdot k\log k)$, where
    \begin{equation*}
        \Maj_k(\mcC) = \set{x \mapsto \sign(c_1(x) + \cdots + c_k(x)) \mid c_1,\ldots, c_k \in \mcC}.
    \end{equation*}
\end{fact}
The second manipulation concerns the VC dimension of symmetric differences.
\begin{fact}[The VC dimension of symmetric difference, \cite{Vid13}]
    \label{fact:sym-VC}
    For any concept class $\mcC$, let $\Delta(\mcC)$ denote the class containing every function of the form $-c_1 \cdot c_2$ for $c_1,c_2 \in \mcC$. Then $\VC(\Delta(\mcC)) \leq 10\cdot \VC(\mcC)$.
\end{fact}

We will also work with Rademacher complexity.
\begin{definition}[Rademacher complexity]
    Let $\mcF$ be a class of functions mapping $X$ to $[-1,1]$, $\mcD$ be a distribution over $X$, and $n$ be a sample size. The Rademacher complexity is
    \begin{equation*}
        \Rad_{n, \mcD}(\mcF) \coloneqq \Ex_{\substack{\bx_1,\ldots, \bx_n \iid \mcD \\ \bsigma_1, \ldots, \bsigma_n \iid \Unif(\bits)}}\bracket*{\sup_{f \in \mcF} \abs*{\frac{2}{n} \sum_{i \in [n]} \bsigma_i f(\bx_i)}}.
    \end{equation*}
\end{definition}
We use a few standard properties of Rademacher complexity, all of which can be found in or are easy consequences of the results in \cite{BM02}. The first is that it controls generalization error.
\begin{fact}[Generalization error from Rademacher complexity]
    \label{fact:rad-gen}
    For any $\mcF, \mcD$, with probability at least $1-\delta$ over $\bS \sim \mcD^n$, every $f \in \mcF$ satisfies
    \begin{equation*}
        \abs*{\Ex_{\bx \sim \mcD}[f(\bx)] - \Ex_{\bx \sim \Unif(\bS)}[f(\bx)]} \leq O\paren*{\Rad_{n, \mcD}(\mcF) + \sqrt{\frac{\log(1/\delta)}{n}}}.
    \end{equation*}
\end{fact}
Rademacher complexity can be bounded using VC dimension and satisfies many composition properties. The following is an easy consequence of these.
\begin{fact}[Rademacher bound from VC dimension]
    \label{fact:Rad-compose}
    For any concept class $\mcC$ and $O(1)$-Lipschitz function $\phi:[-1,1] \to [-1,1]$,
    \begin{equation*}
        \Rad_{n,\mcD}(\phi \circ \Conv(\mcC)) \leq O\paren*{\sqrt{\frac{\VC(\mcC)+1}{n}}}
        \qquad\text{for all $n$ and all distributions }\mcD \text{ over }X.
    \end{equation*}
\end{fact}

\pparagraph{The minimax theorem.} We'll use von Neumann's minimax theorem stated in a form particularly convenient for us.
\begin{fact}[The minimax theorem \cite{V28}]
    \label{fact:minimax} For any finite concept classes $\mcH$ and $\mcC$ and $\ell:\mcH \times \mcC \to \R$,
    \begin{equation*}
        \inf_{\text{distribution }\mu_h \text{ over } \mcH} \set*{\max_{c \in \mcC} \set*{\Ex_{\bh \sim \mu_h}[\ell(\bh, c)]}} = \sup_{\text{distribution }\mu_c\text{ over }\mcC} \set*{ \min_{h \in \mcH} \set*{\Ex_{\bc \sim \mu_c}[\ell(h, \bc)]}}.
    \end{equation*}
\end{fact}
We do not assume that our concept classes are finite, but they do always have finite VC dimension. Finite VC dimension means we can take a finite covering over the concept classes with arbitrarily low error, giving the following.
\begin{corollary}[Minimax lemma on finite VC classes]
    \label{cor:minimax}
    For any concept classes $\mcH$ and $\mcC$ with finite VC dimension and a loss function $\ell(h,c)$ that depends on only
    \begin{enumerate}
        \item The evaluations of $h$ and $c$ on a finite number of points.
        \item Smoothly (with finite Lipschitz constant) on $\Prx_{\bx \sim \mcD}[h(\bx) \neq c(\bx)]$ for some distribution $\mcD$.
    \end{enumerate}
    Then,
    \begin{equation*}
        \inf_{\text{distribution }\mu_h \text{ over } \mcH} \set*{\sup_{c \in \mcC} \set*{\Ex_{\bh \sim \mu_h}[\ell(\bh, c)]}} = \sup_{\text{distribution }\mu_c\text{ over }\mcC} \set*{ \inf_{h \in \mcH} \set*{\Ex_{\bc \sim \mu_c}[\ell(h, \bc)]}}.
    \end{equation*}
\end{corollary}

\section{A simple algorithm for optimizing losses}
\label{sec:general-opt}
In this section, we give a simple algorithm for finding a hypothesis with low loss induced by some function $f:\bits^3 \to \R$. This procedure will be used in the proof of all our ERM efficient learners.

Throughout this section, we will work with the multilinear extension of $f$, denoted $\tilde{f}:[-1,1]^3 \to \R$,
\begin{equation*}
    \tilde{f}(c,h,y) \coloneqq \Ex_{\bc \sim \Rad(c), \bh \sim \Rad(h), \by \sim \Rad(y)}[f(\bc,\bh,\by)].
\end{equation*}
The goal is to find a (real-valued) hypothesis $h:X \to [-1,1]$ with low loss for every $c \in \mcC$ on a dataset $S$ where the loss is
\begin{equation}
    \label{eq:def-loss-general}
    \ell_S(c,h) = \Ex_{(\bx, \by) \sim \Unif(S)}\bracket*{\tilde{f}(c(\bx), h(\bx), \by)}.
\end{equation}

We will have two requirements on the function $f$. The first requirement is that the agreement coefficients below must be nonpositive for both choices of $y \in \bits$. This corresponds to the loss being larger when $c$ and $h$ disagree.
\begin{definition}[Agreement coefficients]
    \label{def:agreement-coef}
    For any $f:\bits^3 \to \R$, the agreement coefficients of $f$ are defined as
    \begin{equation*}
        a_y = f(1,1,y) + f(-1,-1,y) - f(1,-1,y) - f(-1,1,y).
    \end{equation*}
\end{definition}

Our algorithm will maintain $\mu \in \Conv(\mcC)$. For a fixed function $g$, it will output $h \coloneqq g \circ \mu$. Our second requirement is that there exists a function $g$ for which this strategy always has nonpositive loss in expectation over $\bc \sim \mu$ (this corresponds to \Cref{lem:intro-g-for-dist} for the special case of our malicious noise loss).
\begin{definition}[$g$-feasible]
    \label{def:feasible-func}
    We say $f:\bits^3\to \R$ is $g$-feasible if, for any $\mu \in [-1,1]$ and $y \in \bits$, its multilinear extension $\tilde{f}$ satisfies
    \begin{equation*}
        \tilde{f}(\mu, g(\mu), y) \leq 0.
    \end{equation*}
    Equivalently, for any dataset $S$ and function $\mu:X \to [-1,1]$, $\ell_S(\mu, g \circ \mu) \leq 0$.
\end{definition}

We're now ready to state the main theorem of this section.

\begin{figure}[htb] 
  \captionsetup{width=.9\linewidth}
    
    \begin{tcolorbox}[colback = white,arc=1mm, boxrule=0.25mm]
    \vspace{2pt} 
    \Learn$(\ell, g, \eps):$\vspace{6pt}
    
    \textbf{Input:} A loss $\ell$, concept class $\mcC$, function $g$, and parameter $\eps$. \vspace{6pt}

    \textbf{Output:} A hypothesis $h$ satisfying $\ell(c, h) \leq \eps$ for all $c \in \mcC$.\vspace{-4pt}

    \begin{enumerate}
            \item Initialize $\mu_1 = c$ for arbitrary $c \in \mcC$.
            \item Repeat for $t = 1,2,\ldots$ \vspace{-2pt}
            \begin{enumerate}
                \item[(i)] Set $h_t \coloneqq g \circ \mu_t$ and find any $c_t \in \mcC$ satisfying $\ell(c_t, h_t) \geq \max_{c^\star \in \mcC}\ell(c^\star, h_t) - \eps/20$.
                \item[(ii)] If $\ell(c_t, h_t) \leq 19/20 \cdot \eps$, output $h_t$. Otherwise, update: \vspace{-4pt}
                \begin{equation*}
                    \mu_{t+1} = (1-\gamma)\mu_t + \gamma c_t\quad\quad\text{where}\quad\gamma = \frac{2}{t+2}.
                \end{equation*}
            \end{enumerate}
    \end{enumerate}
    \end{tcolorbox}
\caption{Our algorithm meeting the requirements of \Cref{thm:general-optimization}}
\label{fig:GeneralLearner}
\end{figure}

\begin{theorem}[The algorithm in \Cref{fig:GeneralLearner} converges quickly]
    \label{thm:general-optimization}
    For any $f:\bits^3 \to \R$ that has nonpositive agreement coefficients and is $g$-feasible for nondecreasing $g$ with Lipschitz constant $L$ and dataset $S$, let $\ell_S$ be the loss defined in \Cref{eq:def-loss-general}. For any concept class $\mcC$ and sufficiently small parameter $\eps$, the algorithm $\Learn(\ell_S,g,\eps)$ described in \Cref{fig:GeneralLearner} converges in $O(\linf{f}L/\eps)$-iterations and outputs a hypothesis $h$ satisfying $\max_{c \in \mcC}\ell_S(c, h) \leq \eps$.
\end{theorem}

\pparagraph{Overview of the analysis.}
We prove \Cref{thm:general-optimization} by reduction to well-known methods in convex optimization. The algorithm in \Cref{fig:GeneralLearner} maintains $\mu \in \Conv(\mcC)$ and will ultimately output $h = g \circ \mu$. The key step is to design a progress measure $P(\mu)$ such that the $c \in \mcC$ maximizing $\ell(c, g \circ \mu)$ is exactly that which minimizes $\la \nabla P(\mu), c \ra$. This means every iteration of $\Learn$ decreases this progress measure by a large amount, and this can only happen so many times before the algorithm converges. 

In a bit more detail, we ensure that the updates of $\Learn$ correspond to the classic Frank-Wolfe \cite{FW56} convex optimization algorithm  of this progress measure $P$, and so can appeal to existing convergence results. In particular, we'll use Jaggi's bound on the duality gap given below.
\begin{theorem}[Theorem 2 of \cite{J13}]
    \label{thm:Jaggi}
    Let $P$ be a convex objective over convex and compact $\Omega$ where $\nabla P$ is $L$-Lipschitz and $\Omega$ has diameter $d$. Then for curvature constant $C = Ld^2$, any $\delta > 0$, and timescale $T$, there is some $t \leq T$ for which $\nu_t$ will satisfy
    \begin{equation*}
        \max_{\omega^{\star} \in \Omega}\la \nabla P(\nu_t),\nu_t - \omega^{\star}\ra \leq \frac{27 C (1 + \delta)}{4(T+2)}
    \end{equation*}
    where $\nu_1,\ldots, \nu_T$ are the iterates of $\FW(P, \Omega, C, \delta)$ described in \Cref{fig:FW}.
\end{theorem}

\begin{figure}[htb] 
  \captionsetup{width=.9\linewidth}
    
    \begin{tcolorbox}[colback = white,arc=1mm, boxrule=0.25mm]
    \vspace{2pt} 
    $\FW(P,\Omega,C, \delta):$\vspace{6pt}
    
    \textbf{Input:} A convex objective $P$ over convex and compact $\Omega$, curvature constant $C$, and approximation parameter $\delta$ . 
    
    \begin{enumerate}
            \item Initialize $\nu_1 \in \Omega$ arbitrarily.
            \item Repeat for $t = 1,2,\ldots$ \vspace{-2pt}
            \begin{enumerate}
                \item[(i)] Set $\gamma = \frac{2}{2 + t}$.
                \item[(ii)] Find any $c_t \in \Omega$ satisfying $\la \nabla P(\nu_t), c_t \ra \leq \min_{\omega^{\star} \in \Omega} \la \nabla P(\nu_t), \omega^\star \ra  + \frac{\delta \gamma C}{2} $.
                \item[(iii)] Update $\nu_{t+1} = (1-\gamma)\nu_t + \gamma c_t$ 
            \end{enumerate}
    \end{enumerate}
    \end{tcolorbox}
\caption{Algorithm 2 in \cite{J13}. }
\label{fig:FW}
\end{figure} 

\begin{remark}[Use of AI tools]
    \label{remark:AI-FW}
    We originally had a more ad-hoc analysis of the convergence of the algorithm in \Cref{fig:GeneralLearner}. Appropriate querying of ChatGPT-5 (for one specific loss function) suggested a connection to the Frank-Wolfe algorithm, which ultimately allowed us to simplify the analysis.
\end{remark}

\subsection{Establishing a convex domain and progress measure}
Throughout this section, we fix some dataset $S = ((x_1,y_1),\ldots, (x_n, y_n))$ and a concept class $\mcC$ of functions from $X \to \bits$. For any function $h:X \to [-1,1]$, we use $h_S$ as shorthand for
\begin{equation*}
    h_S \coloneqq (h(x_1), \ldots, h(x_n)).
\end{equation*}
We define $\Omega$ to be the convex closure of $\mcC$
\begin{equation}
    \label{eq:def-domain}
    \Omega \coloneqq \Conv(c_S : c \in \mcC).
\end{equation}
We first observe that $\Omega$ satisfies the following requirements.
\begin{proposition}
    \label{prop:domain-properties}
    $\Omega$ is convex, compact, and has diameter at most $2\sqrt{n}$.
\end{proposition}
\begin{proof}
    Convexity follows immediately from the definition and the diameter bound from $\Omega$ being a subset of $[-1,1]^n$.
    
    For compactness, there exists a finite sequence $c^{(1)},\ldots, c^{(m)}$ which take on all labelings that $\mcC$ achieves on $S$, since there are at most $2^n$ unique labelings.  Then, an equivalent definition of $\Omega$ is the mapping $F(\Delta)$ where
    \begin{equation*}
        F(\lambda_1,\ldots, \lambda_m) = \sum_{i \in [m]} \lambda_i c^{(i)}_S \quad\text{and}\quad \Delta = \set*{ \lambda \in \mathbb{R}^m :
\lambda_i \ge 0 \ \forall i,\ \sum_{i=1}^m \lambda_i = 1 }.
    \end{equation*}
    $F$ is continuous and $\Delta$ is compact. Therefore, $\Omega = F(\Delta)$ is compact.
\end{proof}

We next give a decomposition of the loss function that will be useful in designing our objective $P$.

\begin{proposition}[Decomposing the loss function]
    \label{prop:loss-decompose}
    The loss $\ell_S(c, h)$ defined in \Cref{eq:def-loss-general} can be decomposed as
    \begin{equation*}
        \ell_S(c, h)
        = \frac{1}{4} \cdot
        \Ex_{(\bx,\by) \sim \Unif(S)}\bracket*{c(\bx)\paren*{a_{\by}h(\bx) + b_{\by}} + Q_{\by}(h(\bx))}
    \end{equation*}
    where
    \begin{align*}
        a_y &\coloneqq f(1,1,y)+f(-1,-1,y)-f(1,-1,y)-f(-1,1,y), \\
        b_y &\coloneqq f(1,1,y)+f(1,-1,y)-f(-1,1,y)-f(-1,-1,y), \\
        Q_y(t) &\coloneqq (1+t)\paren*{f(1,1,y)+f(-1,1,y)} \\
        &\qquad\quad + (1-t)\paren*{f(1,-1,y)+f(-1,-1,y)}.
    \end{align*}
\end{proposition}
\begin{proof}
A straightforward check of the $8$ cases confirms, for any $c(x),h(x),y \in \bits$, that
\begin{equation*}
    f(c(x),h(x),y) =\frac{1}{4} \cdot \paren*{ c(x)\cdot \paren*{a_{y}h(x) + b_{y}} + Q_{y}(h(x))}.
\end{equation*}
Any two multilinear functions that agree on all end-points must be identical (see e.g. Theorem 1.1 of the textbook \cite{ODBook}). Hence, the above implies that for any $y \in \bits$ and $c(x),h(x) \in [-1,1]$ we also have $ \tilde{f}(c(x),h(x),y) = \frac{1}{4} \cdot \paren*{c(x)\cdot \paren*{a_{y}h(x) + b_{y}} + Q_{y}(h(x))}$. The desired decomposition then follows by taking an expectation of $(\bx,\by) \sim \Unif(S)$.
\end{proof}

Using this decomposition, we are ready to define our progress measure. The key is that, in \Cref{eq:gradient-to-loss}, the term $Q(\mu_S)$ is independent of $c$. Therefore, the choice of $c$ which minimizes $\la \nabla P(\mu_S), c_S \ra$ is exactly that which maximizes $\ell(c, g \circ \mu)$.

\begin{claim}[Defining our progress measure]
    \label{claim:general-progress-measure}
    For any $g:[-1,1] \to [-1,1]$, let $G(s) = \int_0^s g(t)dt$ be its antiderivative and define, for any $\nu \in [-1,1]^n$,
    \begin{equation}
        \label{eq:def-progress-measure}
         P(\nu) \coloneqq -\frac{1}{4n}\sum_{i \in [n]} \paren*{a_{y_i} \cdot G(\nu_i) + b_{y_i} \cdot \nu_i}.
    \end{equation}
    Then, for any functions $c, \mu:X \to [-1,1]$,
    \begin{equation}
        \label{eq:gradient-to-loss}
        \la \nabla P(\mu_S), c_S \ra = -\ell_S(c, g \circ \mu) + Q(\mu_S)
    \end{equation}
    where $\ell_S$ is as defined in \Cref{eq:def-loss-general}, $a,b,Q_y$ are defined in \Cref{prop:loss-decompose}, and
    \begin{equation*}
        Q(\nu)\coloneqq \frac{1}{4n}\sum_{i \in [n]}Q_{y_i}(g(\nu_i)).
    \end{equation*}
\end{claim}

\begin{proof}
Apply \Cref{prop:loss-decompose} with $h=g \circ \mu$:
\begin{equation*}
    \ell(c,g \circ \mu)
    = \frac{1}{4}\Ex_{(\bx,\by)\sim\Unif(S)}\bracket*{
    c(\bx)\paren*{a_{\by}\cdot g \circ \mu(\bx)+b_{\by}} + Q_{\by}(g \circ \mu(\bx))}.
\end{equation*}
Rearranging,
\begin{equation*}
    \ell_S(c,g \circ \mu) - Q(\mu_S)
    = \frac{1}{4n}\sum_{i \in [n]} c(x_i)\paren*{a_{y_i}\cdot g \circ \mu(x_i)+b_{y_i}}.
\end{equation*}
By definition of $P$ and $G'=g$, we have that
\begin{equation}
    \label{eq:loss-partial}
    \frac{\partial P(\nu)}{\partial \nu_i} = -\frac{1}{4n} \cdot (a_{y_i}  \cdot g(\nu_i) + b_{y_i})
\end{equation}
Therefore,
\begin{align*}
    \la \nabla P(\mu_S), c_S\ra
    &= -\frac{1}{4n}\sum_{i \in [n]} c(x_i)\paren*{a_{y_i}\cdot g \circ \mu(x_i)+b_{y_i}} \\
    &= -\ell_S(c,g \circ \mu) + Q(\mu_S). \qedhere
\end{align*}
\end{proof}

We next show the following
\begin{proposition}[Convexity of $P$]
    \label{prop:P-convex}
    If $g$ is nondecreasing and $\set{a_y}_{y \in \bits}$ as defined in \Cref{prop:loss-decompose} are nonpositive, then $P$ defined in \Cref{claim:general-progress-measure} is convex.
\end{proposition}
\begin{proof}
Recall that,
\begin{equation*}
    P(\nu)\coloneqq-\frac{1}{4n}\sum_{i\in[n]}\paren*{a_{y_i}\cdot  G(\nu_i)+b_{y_i}\nu_i}.
\end{equation*}
Since linear functions are convex, and the sum of convex functions is convex, it suffices for $t \mapsto -a_{y_i}G(t)$ to be convex for all $i \in [n]$. Whenever $\set{a_y}_{y \in \bits}$  are nonpositive, it then suffices for $G$ to be convex, which follows from $g$ being nondecreasing.
\end{proof}
\begin{proposition}[Lipschitzness of $\nabla P$]
    \label{prop:gradP-lipschitz}
    Suppose $g$ is $L$-Lipschitz on $[-1,1]$. Then $P$ from \Cref{claim:general-progress-measure} satisfies
    \begin{equation*}
        \ltwo{\nabla P(\nu)-\nabla P(\nu')}
        \le \frac{\linf{f}\cdot L\cdot\ltwo{\nu-\nu'}}{n}
        \qquad\text{for all }\nu,\nu'\in[-1,1]^n.
    \end{equation*}
\end{proposition}
\begin{proof}
We bound,
\begin{align*}
    \ltwo{\nabla P(\nu)-\nabla P(\nu')}^2 &= \sum_{i \in [n]} \paren*{ \frac{\partial P(\nu)}{\partial \nu_i} - \frac{\partial P(\nu')}{\partial \nu'_i}}^2 \\
    &= \sum_{i \in [n]} \paren*{ -\frac{a_{y_i}}{4n} \cdot (g(\nu_i) - g(\nu'_i))}^2\tag{\Cref{eq:loss-partial}}\\
    &\leq \sum_{i \in [n]}\frac{a_{y_i}^2}{16n^2} \cdot (L (\nu_i - \nu'_i))^2 \tag{$g$ is $L$-lipschitz.}
\end{align*}
Next, since $a_y$ is the sum of $4$ evaluations of $f$ (see \Cref{prop:loss-decompose}), we have that $\abs{a_y} \leq 4\linf{f}$. Therefore,
\begin{equation*}
    \ltwo{\nabla P(\nu)-\nabla P(\nu')}^2 \leq \sum_{i \in [n]}\frac{16 \linf{f}^2}{16n^2} \cdot (L (\nu_i - \nu'_i))^2 = \paren*{\frac{\linf{f}\cdot L\cdot\ltwo{\nu-\nu'}}{n}}^2.
\end{equation*}
Taking the square root of both sides give the desired result.
\end{proof}

\subsection{Applying Jaggi's bound}
The key observation is that the duality gap can be used to bound the loss.
\begin{claim}[Duality gap bounds loss]
    \label{claim:duality-to-loss}
    For any $\mu:X \to [-1,1]$,
    \begin{equation*}
        \max_{c \in \mcC} \set*{\ell_S(c,g \circ \mu) }\leq  \max_{\omega^{\star} \in \Omega}\set*{\la \nabla P(\mu_S),\mu_S - \omega^{\star}\ra}.
    \end{equation*}
\end{claim}
\begin{proof}
For any $c \in \mcC$,
\begin{align*}
    \la \nabla P(\mu_S),\mu_S - c_{S}\ra &= \la \nabla P(\mu_S),\mu_S \ra - \la \nabla P(\mu_S), c_{S}\ra \tag{Linearity}\\
    &= \paren[\big]{ -\ell_S(\mu, g \circ \mu) + Q(\mu_S)} - \paren[\big]{  -\ell_S(c, g \circ \mu) + Q(\mu_S)} \tag{\Cref{eq:gradient-to-loss}}\\
    &=\ell_S(c, g \circ \mu) - \ell_S(\mu, g \circ \mu) \\
    &\geq \ell_S(c, g \circ \mu)  \tag{$\ell_S(\mu, g \circ \mu) \leq 0$ by \Cref{def:feasible-func}.}
\end{align*}
The desired bound is then implied by $c_S \in \Omega$ for all $c \in \mcC$.
\end{proof}
Using this, we are able to apply Jaggi's bound to prove \Cref{thm:general-optimization}.
\begin{proof}[Proof of \Cref{thm:general-optimization}]
    We map the notation of $\Learn$ in \Cref{fig:GeneralLearner} to those of $\FW$ in \Cref{fig:FW} by setting
    \begin{equation*}
        \nu_t \coloneqq (\mu_t)_S.
    \end{equation*}
    Using the existing definitions of $\Omega$ and $P$  given in \Cref{eq:def-domain,eq:def-progress-measure} respectively, in order to ensure that $\Learn$ faithfully follows $\FW$, we need to check that the gradient condition is satisfied at each step. For approximation parameter $\delta = 1$ and curvature constant $C = 4L\linf{f}$ (which follows from the diameter bound in \Cref{prop:domain-properties} and Lipschitz bound in \Cref{prop:gradP-lipschitz}), we need the choice of $c_t$ in $\Learn$ to satisfy
    \begin{equation*}
        \la \nabla P((\mu_t)_S), c_t \ra \leq \min_{\omega^{\star} \in \Omega} \la \nabla P((\mu_t)_S), \omega^\star \ra  + 2 \gamma L \linf{f}.
    \end{equation*}
    Since $\Omega = \Conv(c_S : c \in \mcC)$, the above minimum is obtained at $\omega^{\star} = c_S$ for some $c \in \mcC$. Then, rewriting the gradient dot product using \Cref{eq:gradient-to-loss} and removing the shared $Q((\mu_t)_S)$ from both sides, the above is equivalent to,
    \begin{equation*}
        -\ell_S(c_t, g \circ \mu_t) \leq \min_{c \in \mcC}\set*{-\ell_S(c, g \circ \mu_t)} + 2 \gamma L \linf{f}.
    \end{equation*}
    Recalling that $\gamma = \frac{2}{2 + t}$ and that we always choose a $c_t$ within $\eps/20$ of the maximum loss in \Cref{fig:GeneralLearner}, this condition holds for all $t \in [T]$ where
    \begin{equation}
        \label{eq:T-limit}
        T = \floor*{\frac{80 L \linf{f}}{\eps} - 2}.
    \end{equation}
    Next, by \Cref{thm:Jaggi} (still with $\delta = 1$ and $C =  4L\linf{f}$), there is some $t \in [T]$ for which,
    \begin{equation*}
        \max_{\omega^{\star} \in \Omega}\la \nabla P((\mu_t)_S),(\mu_t)_S- \omega^{\star}\ra \leq \frac{54 L \linf{f}}{T+2} \leq \frac{19 \eps}{20},
    \end{equation*}
    where the second inequality uses that for small enough $\eps$, the $\floor{\cdot}$ in \Cref{eq:T-limit} has negligible effect and so the bound approaches $\frac{54 \eps}{80}$ which is less than $\frac{19 \eps}{20}$. This implies, by \Cref{claim:duality-to-loss}, that $\Learn$ will successfully halt in at most $T$ steps.
\end{proof}

\section{Learning with malicious noise: Proof of \Cref{thm:malicious-intro}}
In this section, we prove the following.
\begin{theorem}[Optimal learning with malicious noise, formal version of \Cref{thm:malicious-intro}]
    \label{thm:malicious-body}
    For any concept class $\mcC$ with VC dimension $d$ and $\eps, \delta \in (0,1), \alpha \in [0,1)$, there is an algorithm that $(\eps + (\lfrac{1}{2} + \alpha)\cdot \eta/(1-\eta),\delta)$ learns $\mcC$ with $\eta$-malicious noise using
    \begin{equation*}
        n\coloneqq O\paren*{\frac{d \cdot \log(2 + \alpha/\eps)+ \log(1/\delta)}{\eps\cdot (\alpha + \eps)}}
    \end{equation*}
    samples. Furthermore, this algorithm is efficiently implementable using $O(1/\eps)$ calls to an ERM oracle for $\mcC$.
\end{theorem}
We encourage the reader to consider two regimes. \Cref{thm:malicious-body} shows that error $0.51 \cdot \lfrac{\eta}{1-\eta} + \eps$ can be accomplished with $\tilde{O}\paren*{\frac{d + \log(1/\delta)}{\eps}}$ samples and error $0.5 \cdot \lfrac{\eta}{1-\eta} + \eps$ with $O\paren*{\frac{d + \log(1/\delta)}{\eps^2}}$ samples.

Recall from \Cref{subsec:mal-overview}, the first step is to define a loss function. 
\begin{claim}[Malicious error can be linearized, generalization of \Cref{lem:mal-error-loss-overview}]
    \label{lem:mal-error-loss-body}
    For any dataset $S \in (X \times \bits)^n$, let $\bar{h}:X \to [-1,1]$ be any function satisfying, for all $c \in \mcC$,
    \begin{equation}
        \label{eq:mal-loss-body}
        \ell_S(c, \bar{h}) \leq \eps \quad\quad\text{where } \ell_S(c, \bar{h}) \coloneqq\Ex_{\bx,\by \sim \Unif(S)}[c(\bx)\by \cdot (2 - \bar{h}(\bx)\by) -\bar{h}(\bx)\by].
    \end{equation}
    Then, the randomized function $\bh = \Rad(\bar{h})$ satisfies
    \begin{equation*}
        \Prx_{\bi \sim \Unif(I), \bh}[\bh(x_{\bi}) \neq y_{\bi}] \leq \frac{\eta + \eps/2}{2(1-\eta)}.
    \end{equation*}
    simultaneously for all $\eta$ and $I$ of size $(1-\eta)n$ for which there exists some $c \in \mcC$ that is fully consistent with $S_I$.
\end{claim}

The second step is to give a method for constructing a good hypothesis $\bar{h}$ given any distribution over concepts.
\begin{claim}[A good hypothesis exists for any distribution of concepts, restatement of \Cref{lem:intro-g-for-dist}]
    \label{lem:body-g-for-dist}
    Let $g:[-1,1] \to [-1,1]$ be any function satisfying, for all $t \in [-1,1]$,
    \begin{equation}
        \label{eq:g-constraints-mal-body}
        \frac{2t}{1+t} \leq g(t) \leq \frac{2t}{1-t}.
    \end{equation}
    Then, for any distribution $\mu$ over functions and dataset $S$
    \begin{equation*}
        \Ex_{\bc \sim \mu}\bracket*{\ell_S(\bc, \bar{h})} \leq 0 \quad\quad\text{where}\quad \bar{h}(x) = g\paren*{\Ex_{\bc \sim \mu}[\bc(x)]}.
    \end{equation*}
\end{claim}
We will instantiate the function $g$ with a choice that ensures the hypothesis is always in $\Conv(\Maj_7(\mcC))$. We choose a majority over $k=7$ instances because for a function $g$ to satisfy \Cref{lem:body-g-for-dist} it must be the case that $g'(0) = 2$, and a smaller $k$ would give $g'(0) < 2$. Note the below claim most directly gives a hypothesis within $\Conv(\Maj_7(\mcC) \cup \mcC)$. However, we observe that
every $c \in \mcC$ is simply equal to $\Maj(c,c,\ldots,c)$ and so is also in $\Maj_k(\mcC)$, so $\Maj_7(\mcC) \cup \mcC = \Maj_7(\mcC)$.
\begin{claim}[Instantiating the function $g$]
    \label{claim:maj-suffices}
    For any $k$, define $M_k(t)$ as
    \begin{equation*}
        M_k(t) \coloneqq \Ex_{\bz_1,\ldots, \bz_k \iid \Rad(t)}[\sign(\bz_1 + \cdots + \bz_k)].
    \end{equation*}
    Then,
    \begin{equation*}
        g(t) \coloneqq \frac{16}{19}M_7(t) + \frac{3}{19}t
    \end{equation*}
    satisfies \Cref{eq:g-constraints-mal-body} and is $2$-Lipschitz.
\end{claim}
\begin{figure}[htb]
    \centering
    \begin{tikzpicture}
        \begin{axis}[
            width=0.6\textwidth,
            height=0.4\textwidth,
            xmin=-1, xmax=1,
            ymin=-1, ymax=1,
            xtick={-1,1},
            ytick={-1,1},
            xlabel={$t$},
            ylabel={$g(t)$},
            axis lines=middle,
            domain=-0.999:0.999,
            samples=401,
            clip mode=individual,
            legend cell align=left,
            legend style={at={(0.02,0.98)},anchor=north west,fill=white,draw=none,font=\small},
        ]
            \addplot[name path=lower,draw=none] {max(-1,2*x/(1+x))};
            \addplot[name path=upper,draw=none] {min(1,2*x/(1-x))};
            \addplot[blue!20,draw=none] fill between[of=lower and upper];

            \addplot[blue!80!black,thick,dotted] {max(-1,2*x/(1+x))};
            \addplot[blue!80!black,thick,dotted] {min(1,2*x/(1-x))};

            \addplot[black,very thick] {(38*x - 35*x^3 + 21*x^5 - 5*x^7)/19};
            \addlegendentry{Our choice of $g$}
        \end{axis}
    \end{tikzpicture}
    \caption{The shaded region depicts the constraints of \Cref{eq:g-constraints-mal-body}, which our choice of $g$ lies within.}
    \label{fig:g-constraints-mal-intro}
\end{figure}
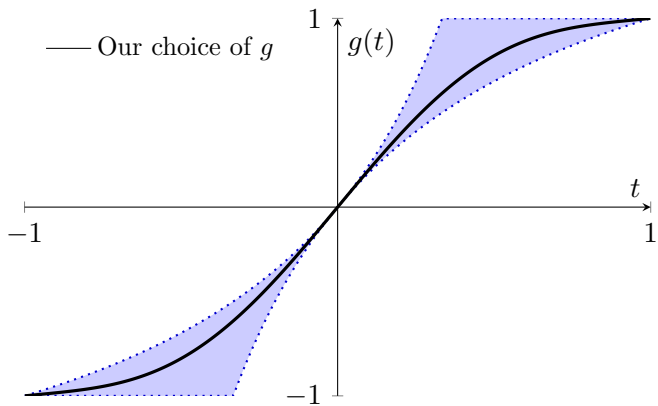

\noindent\textbf{Structure of this section.} The proof of \Cref{claim:maj-suffices} is straightforward but tedious. We defer it to \Cref{appendix:g-choice}. We prove \Cref{lem:mal-error-loss-body} in \Cref{subsec:proof-of-mal-error-loss}, \Cref{lem:body-g-for-dist} in \Cref{subsec:proof-of-g-for-dist}. For now, we show how they imply our malicious noise learner.

\begin{proof}[Proof of \Cref{thm:malicious-body}]
    Let $g$ be the function defined in \Cref{claim:maj-suffices}. We will use \Cref{thm:general-optimization} to find some function $\mu \in \Conv(\mcC)$ for which $\overline{h} = g \circ \mu$ has small loss, satisfying $\ell_S(c, \overline{h}) \leq \eps$ for all $c \in \mcC$. We first verify that the assumptions of that optimization algorithm are satisfied. Defining
    \begin{equation*}
        f(c,h,y) \coloneqq cy \cdot (2 - hy) - hy,
    \end{equation*}
    we have that $\ell_S(c,h)$ is the loss defined in \Cref{eq:def-loss-general}. A direct computation verifies that $f$ has nonpositive agreement coefficients and is bounded. Furthermore, it is $g$-feasible by \Cref{lem:body-g-for-dist,claim:maj-suffices}. Therefore, the algorithm in \Cref{fig:GeneralLearner} converges in $O(1/\eps)$ iterations to $\mu \in \Conv(\mcC)$ for which $\ell_S(c, g \circ \mu) \leq \eps$. Each iteration needs to approximately find the $c \in \mcC$ maximizing $\Ex_{\bx,\by \sim \Unif(S)}[c(\bx)\by \cdot (2 - \bar{h}(\bx)\by)]$ which is simply an approximate weighted ERM call, and can be done with a single unweighted ERM call by \Cref{fact:weighted-ERM}.

    Given this $\mu \in \Conv(\mcC)$ for which $\bar{h} = g \circ \mu$ has small loss, we output $\bh = \Rad(\bar{h})$. Note that $\bh$ is a mixture of hypotheses in $\Maj_7(\mcC)$. This is because, since $\mu \in \Conv(\mcC)$, there exists some distribution $\mcD_c$ over $\mcC$ for which $\mu(x) = \Ex_{\bc \sim \mcD_c}[\bc(x)]$. Then, if we sample $\bc_1, \ldots, \bc_7 \iid \mcD_c$ and set
    \begin{equation*}
        \bh \coloneqq \begin{cases}
            \Maj_7(\bc_1,\bc_2,\ldots, \bc_7) &\text{with probability $\frac{16}{19}$}\\
            \Maj_7(\bc_1,\bc_1,\ldots,\bc_1)&\text{otherwise,}
        \end{cases}
    \end{equation*}
   it is straightforward to verify that $\bh = \Rad(\bar{h})$
    
    What remains is to analyze the accuracy of $\bh$. Let $\bJ \subseteq [n]$ be the indices the adversary may corrupt and $\bI \coloneqq [n] \setminus \bJ$ be the clean indices. By our choice of $n$ and \Cref{cor:Chernoff}, we have with probability at least $1-\delta$ that $\boldeta \coloneqq \abs*{\bJ}/n$ is at most $(1 + \alpha) \eta + \eps$. 

    For the remainder of the proof, we can assume without loss of generality that $\eta \leq 2/3$ as otherwise the error guarantee is vacuous. Similarly, we can assume that $\eps,\alpha \leq 1/100$. Otherwise, we could divide them each by $100$ which would only affect the constant hidden in the $O(\cdot)$ of \Cref{thm:malicious-body}. These assumptions allow us to conclude that $(1 + \alpha)\eta + \eps \leq 3/4$. This is useful since it means that $\bS_{\bI}$ will consist of at least $n/4$ points (under our previous condition that $\boldeta \leq(1 + \alpha) \eta + \eps)$. Since all of these $\geq n/4$ points are i.i.d. clean draws, \Cref{cor:VC-random} along with $\VC(\Maj_7(\mcC)) = O(d)$ from \Cref{fact:maj-VC} gives, with probability at least $1-\delta$,
     \begin{align*}
        \error_{\mcD}(\bh, c^{\star}) &\leq (1 + \alpha) \cdot \Prx_{(\bx,\by) \sim \Unif(\bS_{\bI})}[\bh(\bx) \neq \by] + \eps\\
        &\leq (1 + \alpha)\cdot \frac{\boldeta + \eps}{2(1-\boldeta)} + \eps \tag{\Cref{lem:mal-error-loss-body}} \\
        &\leq (1 + \alpha)\cdot \frac{(1 + \alpha)\eta  +\eps+ \eps}{2(1-(\eta(1 + \alpha)+\eps))} + \eps \tag{$\boldeta \leq (1 +\alpha)\eta + \eps$.} \\
    \end{align*}
    Furthermore, since $\eta(1 + \alpha)+\eps \leq 3/4\leq 1 - \Omega(1)$, the above can be upper bounded by 
    \begin{equation*}
        \error_{\mcD}(\bh, c^{\star}) \leq (1 + O(\alpha + \eps))\frac{\eta}{2(1-\eta)} + O(\eps) \leq (1 + O(\alpha))\frac{\eta}{2(1-\eta)} + O(\eps).
    \end{equation*}
     Union bounding over the two failure probabilities, we have that with probability at least $1 - 2\delta$, the desired result holds up to constants. The constants can be adjusted to get exactly the desired error bound.
\end{proof}

\subsection{Proof of \Cref{lem:mal-error-loss-body}}
\label{subsec:proof-of-mal-error-loss}
Fix any $I$ of size $(1-\eta)n$ for which there is some $c \in \mcC$ fully consistent with $S_I$. Then,
\begin{align*}
    \ell_S(c,\overline{h}) = (1-\eta) &\cdot \Ex_{(\bx,\by) \sim \Unif(S_I)}[c(\bx)\by \cdot (2 - \bar{h}(\bx)\by) -\bar{h}(\bx)\by] \\
   + \eta &\cdot \Ex_{(\bx,\by) \sim \Unif(S_{[n] \setminus I})}[c(\bx)\by \cdot (2 - \bar{h}(\bx)\by) -\bar{h}(\bx)\by].
\end{align*}
Since $c(x) = y$ for all $(x,y) \in S_I$ and the expression $c(x)y \cdot (2 - \bar{h}(x)y) - \bar{h}(x)y$ is always at least $-2$,
\begin{equation*}
     \ell_S(c,\overline{h}) \geq (1-\eta) \cdot \Ex_{(\bx,\by) \sim \Unif(S_I)}\bracket*{2 - 2\bar{h}(\bx) \by} - 2 \eta.
\end{equation*}
This can be rearranged to give that,
\begin{equation*}
    \Ex_{(\bx,\by) \sim \Unif(S_I)}\bracket*{1 - \bar{h}(\bx) \by} \leq \frac{\ell_S(c,\overline{h})/2 + \eta}{1-\eta}.
\end{equation*}
The desired result then follows from $\Pr_{\bh}[\bh(x) \neq y] = \frac{1 - \overline{h}(x)y}{2}$.
\qed

\subsection{Proof of \Cref{lem:body-g-for-dist}}
\label{subsec:proof-of-g-for-dist}
We will show for any fixed choice of $x,y$ that
\begin{equation*}
    \Ex_{\bc \sim \mu}[\bc(x)y \cdot (2 - \overline{h}(x)y) - \overline{h}(x)y] \leq 0,
\end{equation*}
which implies the desired result by taking an expectation over $x,y$. Using $\overline{\mu}$ to denote $\Ex_{\bc \sim \mu}[\bc(x)]$, we have,
\begin{equation*}
    \Ex_{\bc \sim \mu}[\bc(x)y \cdot (2 - \overline{h}(x)y) - \overline{h}(x)y] = \overline{\mu} y \cdot (2  - g(\overline{\mu})y) - g(\overline{\mu})y = 2\overline{\mu}y - g(\overline{\mu})(\overline{\mu} + y).
\end{equation*}
With our constraints on $g$, the above is at most $0$ for both choices of $y \in \bits$.
\qed

\section{Distribution-independent learning with nasty noise: Proof of \Cref{thm:nasty-dist-free-upper-intro}}
In this section, we prove both the upper and lower bound of \Cref{thm:nasty-dist-free-upper-intro}. We begin with the upper bound.
\begin{theorem}[Optimal distribution-independent learning with nasty noise, upper bound of \Cref{thm:nasty-dist-free-upper-intro}]
    \label{thm:nasty-dist-free-upper-body}
    For any concept class $\mcC$ with VC dimension $d$ and $\eps, \delta \in (0,1), \alpha \in [0,1)$, the learner from \Cref{thm:malicious-body} is an $((3/2 + \alpha)\eta + \eps,\delta)$ learner for $\mcC$ with $\eta$-nasty noise using
    \begin{equation*}
        n\coloneqq O\paren*{\frac{d \cdot \log(2 + \alpha/\eps)+ \log(1/\delta)}{\eps\cdot (\alpha + \eps)}}
    \end{equation*}
    samples. Furthermore, this algorithm is efficiently implementable using $O(1/\eps)$ calls to an ERM oracle for $\mcC$.
\end{theorem}
\Cref{thm:nasty-dist-free-upper-body} implies that error $1.51\eta + \eps$ can be achieved using $\tilde{O}(d/\eps)$ samples and $1.5\eta + \eps$ using $O(d/\eps^2)$ samples. The proof mostly builds on that of \Cref{thm:malicious-body}.

\begin{proof}
    Recall from the definition of nasty noise (\Cref{def:nasty-noise-basic}) that first a clean sample $\bS^{\circ} \iid (\bx, c^{\star}(\bx))_{\bx \sim \mcD}$ was generated. Then, the adversary produced the dataset the learner receives, $\bS$, by corrupting indices of $\bS^{\circ}$ within $\bJ \subseteq [n]$.

    Let $\boldeta \coloneqq |\bJ|/n$ be the fraction of the samples the adversary corrupted, and $\bI \coloneqq [n] \setminus \bJ$ be the clean indices.  As in the proof of \Cref{thm:malicious-body}, we have with probability at least $1 - \delta$, that $ \boldeta \leq (1 + \alpha)\eta + \eps$ (in both malicious and nasty noise, the number of corrupted points is distributed as $\Bin(n,\eta)$). Then, as we showed in the proof of \Cref{thm:malicious-body}, the hypothesis $\bh$ that the algorithm outputs satisfies, 
     \begin{equation}
        \label{eq:nasty-good-error}
           \Prx_{(\bx,\by) \sim \Unif(\bS_{\bI})}[\bh(\bx) \neq \by]  \leq \frac{\boldeta + \eps}{2(1-\boldeta)}.
     \end{equation}
     Since the clean sample $\bS^{\circ}$ and corrupted sample $\bS$ agree on points within $\bI$, we can therefore upper bound the error on the clean sample as
     \begin{align*}
          \Prx_{(\bx,\by) \sim \Unif(\bS^{\circ})}[\bh(\bx) \neq \by] &= (1 - \boldeta) \cdot  \Prx_{(\bx,\by) \sim \Unif(\bS_{\bI})}[\bh(\bx) \neq \by] + \boldeta \cdot \Prx_{(\bx,\by) \sim \Unif(\bS^{\circ}_{\bJ})}[\bh(\bx) \neq \by] \\
          &\leq (1 - \boldeta) \cdot\frac{\boldeta + \eps}{2(1-\boldeta)} + \boldeta = \frac{3}{2}\boldeta + \eps/2,
     \end{align*}
     where the inequality uses \Cref{eq:nasty-good-error} and that the average error on $\bS^{\circ}_{\bJ}$ is naively upper bounded by $1$. Then, since $\bh$ is supported on $\Maj_7(\mcC)$ which has VC dimension $O(d)$ (see \Cref{fact:maj-VC}), \Cref{cor:VC-random} gives that, with probability $1 - \delta$
     \begin{equation*}
         \error_{\mcD}(\bh, c^{\star}) \leq (1 + \alpha) \cdot \Prx_{(\bx,\by) \sim \Unif(\bS^{\circ})}[\bh(\bx) \neq \by] + \eps.
     \end{equation*}
    Combining these bounds gives that, with probability at least $1 - 2\delta$, the error is at most $(3/2 + O(\alpha)) \eta + O(\eps)$. Adjusting the constants gives the desired result.
\end{proof}

\subsection{Lower bound for distribution-independent learning with nasty noise}
\label{sec:nasty-lb}
We prove this lower bound with an adversary known to be weaker than nasty noise, sometimes called ``general, non-adaptive, contamination" \cite{DK23book}, ``TV-contamination" \cite{DKPP22}, or ``TV-noise" \cite{BHMS26}.
\begin{definition}[TV-noise]
    \label{def:TV-noise}
    In PAC learning with $\eta$-TV noise, given a base distribution $\mcD$ and concept $c$, let $\mcD_c$ be the distribution over labeled examples $(\bx, c(\bx))_{\bx \sim \mcD}$. The adversary chooses distribution over labeled examples $\mcDc$ satisfying
    \begin{equation*}
        \dtv(\mcD_c, \mcDc)  \leq \eta,
    \end{equation*}
    and the learner receives i.i.d. samples from $\mcDc$.
\end{definition}
A standard coupling argument (e.g. Claim~2.5 of \cite{DKKLMS19}) shows that handling nasty noise is at least as hard as handling TV-noise. This is because, for any $\mcDc$ satisfying $\dtv(\mcD_c, \mcDc) \leq \eta$, there is a strategy for the $\eta$-nasty adversary that gives the learner i.i.d. samples from $\mcDc$. Hence, the below lower bound extends automatically to show that no learner in the nasty noise model can have expected error better than $\frac{3\eta}{2} - O(1/d)$. This also means no learner can have $\eps$-error with probability $1-\delta$ when $\eps + \delta \leq \frac{3\eta}{2}-O(1/d)$.
\begin{theorem}[Average case lower bound for distribution-independent learning with TV noise, lower bound of \Cref{thm:nasty-dist-free-upper-intro}]
    \label{thm:nasty-lb-body}
    For any noise rate and concept class $\mcC$ with VC dimension $d \geq 4$ there is a distribution over $\bmcD_{\bc}, \bmcDc$ that meet the requirements\footnote{Meaning, $\bmcD_{\bc}$ will always have labels fully consistent with some $\bc \in \mcC$ and $\dtv(\bmcD_{\bc},\bmcDc) \leq \eta$.} of \Cref{def:TV-noise} with probability $1$ satisfying, for any algorithm $A$ mapping a distribution over labeled examples to a hypothesis,
    \begin{equation*}
        \Ex_{\bh \leftarrow A(\bmcDc)}\bracket*{\Prx_{(\bx,\by) \sim \bmcD_{\bc}}[\bh(\bx) \neq \by]} \geq \min\set*{\frac{1}{2}, \frac{3\eta}{2}} - O\paren*{\frac{1}{d}}.
    \end{equation*}
\end{theorem}
Note that achieving an accuracy of $1/2$ is trivial: The learner that always outputs the coin-flip hypothesis achieves this. Hence, \Cref{thm:nasty-lb-body} shows that either that trivial algorithm or \Cref{thm:nasty-dist-free-upper-body} is optimal.

The proof of \Cref{thm:nasty-lb-body} is divided into three pieces. First, we describe the construction of the distributions $\bmcD_{\bc}, \bmcDc$. Then, in \Cref{claim:lb-tv-construction}, we verify this construction always chooses distributions satisfying $\dtv(\bmcD_{\bc}, \bmcDc) =\eta$. This, combined with the fact that our construction explicitly chooses a $\bc^{\star}$ that $\bmcD_{\bc}$ is consistent with, ensures that $\bmcD_{\bc}, \bmcDc$ meet the requirements of \Cref{def:TV-noise} with probability $1$. Finally, in \Cref{lem:lb-bad-error}, we show no algorithm can have better than $\approx \frac{3\eta}{2}$ error whenever $\eta \leq 1/3$.

\pparagraph{The construction.}
We begin by describing the construction of both the uncorrupted labeled distribution $\bmcD_{\bc^{\star}}$ and the corrupted labeled distribution, $\bmcDc$ . Let $S$ be the size-$d$ set that $\mcC$ shatters. We will draw the target concept $\bc^{\star}$ uniformly among a set of $2^d$ concepts that shatter this set (so $\bc^{\star}$ is equally likely to take on any labeling on $S$). The distribution $\bmcD$ will be chosen as follows. We fix arbitrary $x_{\rand} \in S$. Then, we choose a random $\bx_{\wrong} \sim \Unif(S \setminus \set{x_{\rand}})$. We use $\bS_{\correct} \coloneqq S \setminus (\set{x_{\rand}} \cup \set{\bx_{\wrong}})$ to refer to the remaining points.

Given these choices, to sample $(\bx, \by) \sim \bmcDc$, we
\begin{enumerate}
    \item With probability $\frac{1-2\eta}{d-3}$, we set $\bx = \bx_{\wrong}$ and $\by = -\bc^{\star}(\bx)$.
    \item For each $x \in \bS_{\correct}$, we also set $\bx = x$ with probability $\frac{1-2\eta}{d-3}$. In this case, we set $\by = \bc^{\star}(\bx)$.
    \item With the remaining\footnote{Observe this is nonnegative when $\eta \geq 1/(d-1)$. We are free to assume $\eta$ satisfies this bound in the statement of \Cref{thm:nasty-lb-body}, as otherwise, the desired result is vacuous because the $O(1/d)$ term dominates} $\frac{2(\eta d-\eta-1)}{d-3}$ probability, we set $\bx = x_{\rand}$ and draw $\by \sim \Unif(\bits)$.
\end{enumerate}
To sample $(\bx, \by) \sim \bmcD_{\bc^{\star}}$ we set,
\begin{enumerate}
    \item   With probability $\eta$, we set $\bx = \bx_{\wrong}$ and $\by = \bc^{\star}(\bx)$.
    \item For each $x \in \bS_{\correct}$, we set $\bx = x$ with probability $\frac{1-2\eta}{d-3}$ and $\by = \bc^{\star}(\bx)$.
    \item With the remaining $\frac{\eta d-\eta-1}{d-3}$ probability, we set $\bx = x_{\rand}$ and $\by = \bc^{\star}(\bx)$.
\end{enumerate}

\begin{figure}[htb]
    \centering
    \def\masssmall{0.0714}
    \def\masswrongclean{0.2500}
    \def\massrandclean{0.1786}
    \def\masswrongcorr{-0.0714}
    \def\massrandcorr{0.1786}
    \def\consblue{blue!55}
    \def\inconsred{red!55}
    \def\xwrongL{6.5}
    \def\xwrongR{7.5}
    \def\xrandL{9.5}
    \def\xrandR{10.5}

    \begin{minipage}[t]{0.47\textwidth}
        \centering
        \begin{tikzpicture}
            \begin{axis}[
                    width=7.6cm,
                    height=4.7cm,
                    xmin=0.5, xmax=10.5,
                    ymin=-0.30, ymax=0.33,
                    axis x line*=middle,
                    axis y line=left,
                    axis on top,
                    every axis x line/.append style={line width=0.9pt},
                    every axis y line/.append style={},
                    xtick=\empty,
                    xlabel={},
                    ylabel={},
                    ytick=\empty,
                    clip=false,
                    tick label style={font=\small},
                    label style={font=\small},
                    title={Uncorrupted distribution $\bmcD_{\bc^{\star}}$},
                    title style={font=\small}
                ]
                    \draw[fill=\consblue,draw=none] (axis cs:0.5,0) rectangle (axis cs:10.5,\masssmall);
                    \draw[fill=\consblue,draw=none] (axis cs:\xwrongL,0) rectangle (axis cs:\xwrongR,\masswrongclean);
                    \draw[fill=\consblue,draw=none] (axis cs:\xrandL,0) rectangle (axis cs:\xrandR,\massrandclean);
                    \node[anchor=south,font=\scriptsize] at (axis cs:7,0.257) {$\bx_{\wrong}$};
                    \node[anchor=south,font=\scriptsize] at (axis cs:10,0.186) {$x_{\rand}$};
                    \node[anchor=west,font=\scriptsize] at (axis cs:0.62,0.313) {Consistent with $\bc^{\star}$};
                    \node[anchor=west,font=\scriptsize] at (axis cs:0.62,-0.283) {Inconsistent with $\bc^{\star}$};
            \end{axis}
        \end{tikzpicture}
    \end{minipage}
    \hfill
    \begin{minipage}[t]{0.47\textwidth}
        \centering
        \begin{tikzpicture}
            \begin{axis}[
                    width=7.6cm,
                    height=4.7cm,
                    xmin=0.5, xmax=10.5,
                    ymin=-0.30, ymax=0.33,
                    axis x line*=middle,
                    axis y line=left,
                    axis on top,
                    every axis x line/.append style={line width=0.9pt},
                    every axis y line/.append style={},
                    xtick=\empty,
                    xlabel={},
                    ylabel={},
                    ytick=\empty,
                    clip=false,
                    tick label style={font=\small},
                    label style={font=\small},
                    title={Corrupted distribution $\bmcDc$},
                    title style={font=\small}
                ]
                    \draw[fill=\consblue,draw=none] (axis cs:0.5,0) rectangle (axis cs:\xwrongL,\masssmall);
                    \draw[fill=\inconsred,draw=none] (axis cs:\xwrongL,0) rectangle (axis cs:\xwrongR,\masswrongcorr);
                    \draw[fill=\consblue,draw=none] (axis cs:\xwrongR,0) rectangle (axis cs:10.5,\masssmall);
                    \draw[fill=\consblue,draw=none] (axis cs:\xrandL,0) rectangle (axis cs:\xrandR,\massrandcorr);
                    \draw[fill=\inconsred,draw=none] (axis cs:\xrandL,0) rectangle (axis cs:\xrandR,-\massrandcorr);
                    \node[anchor=north,font=\scriptsize] at (axis cs:7,-0.082) {$\bx_{\wrong}$};
                    \node[anchor=south,font=\scriptsize] at (axis cs:10,0.186) {$x_{\rand}$};
            \end{axis}
        \end{tikzpicture}
    \end{minipage}
    \caption{Visualization of our construction used in the proof of \Cref{thm:nasty-lb-body}.}
    \label{fig:nasty-lb-construction}
\end{figure}
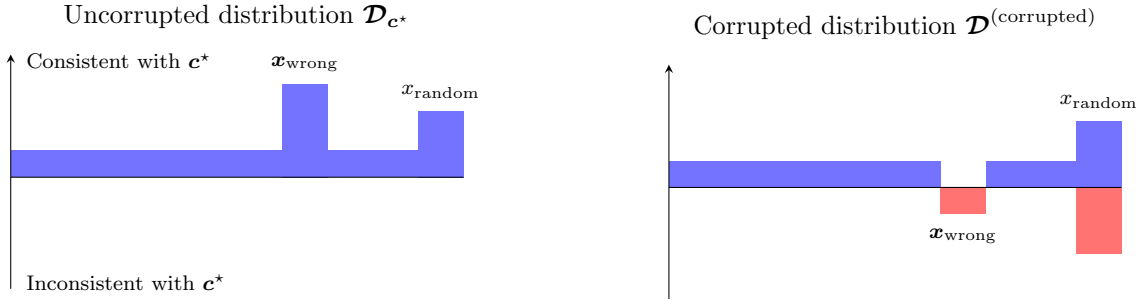

We first verify this construction has the desired total variation distance bound.
\begin{claim}
    \label{claim:lb-tv-construction}
    With probability $1$,
    \begin{equation*}
        \dtv(\bmcD_{\bc^{\star}}, \bmcDc) =\eta.
    \end{equation*}
\end{claim}
\begin{proof}
    We will use the fact that total variation distance can be written as
    \begin{equation*}
        \dtv(\bmcD_{\bc^{\star}}, \bmcDc) = \sum_{x \in S, y \in \bits}\max\set*{0, \bmcD_{\bc^{\star}}(x,y) - \bmcDc(x,y)}.
    \end{equation*}
    The only choice of $x,y$ for which $\bmcD_{\bc^{\star}}(x,y) \geq \bmcDc(x,y)$ is when $x = \bx_{\wrong}$ and $y = \bc^{\star}(x)$. In this case, the difference is exactly $\eta$.
\end{proof}
Next, we verify that no algorithm can learn to high accuracy in expectation.
\begin{lemma}
    \label{lem:lb-bad-error}
    For any (possibly randomized) algorithm $A$ mapping a distribution over labeled examples to a (possibly randomized) hypothesis, 
    \begin{equation*}
        \Ex_{ \bh \leftarrow A\paren*{\bmcDc}}\bracket*{\Prx_{\bx \sim \bmcD}\bracket*{\bh(\bx) \neq \bc^{\star}(\bx)} } \geq \frac{3\eta}{2} - O(1/d),
    \end{equation*}
    whenever $O(1/d) \leq \eta \leq \frac{1}{3} - O(1/d)$.
\end{lemma}
\begin{proof}
    Our construction has only two sources of randomness, the choice of $\bc^{\star}$, which is uniform over all labelings of $S$, and the choice of $\bx_{\wrong}$, which is uniform over $S\setminus \set{x_{\rand}}$. Given a corrupted dataset $\bmcDc$, we say that $c^{\star}, x_{\wrong}$ are ``plausible" if $\bmcDc$ would be the corrupted dataset given $\bc^{\star} = c^{\star}$ and $\bx_{\wrong} = x_{\wrong}$. We observe that $c^{\star}, x_{\wrong}$ are plausible if and only if,
    \begin{equation*}
        c^{\star}(x) = y(x) \cdot (-1)^{\Ind[x = x_{\wrong}]} \quad\quad\text{for all $x \neq x_{\rand}$,}
    \end{equation*}
    where $y(x)$ is the unambiguous value of $\by \mid \bx=x$ for a sample $\bx,\by \sim \bmcDc$. All plausible candidates are equally likely since the original choices in the construction were uniform. Therefore, conditioned on $\bmcDc$,
    \begin{enumerate}
        \item The value of $\bc^{\star}(x_{\rand})$ is equally likely to be $+1$ and $-1$.
        \item The choice of $\bx_{\wrong}$ is equally likely to be any point in $S\setminus \set{x_{\rand}}$.
    \end{enumerate}
    Using these observations, we now compute the expected error of the hypothesis. Let $\bar{h}_{\bmcDc}(x)$ be the average label of the hypothesis $A$ outputs on input $x$ given $\bmcDc$. We will drop the dependence on $\bmcDc$ for convenience. Then,
    \begin{align*}
        \Ex_{\bx \sim \bmcD}[\bar{h}(\bx)\bc^{\star}(\bx)] &=\Ex_{\bx \sim \bmcD}[\bar{h}(\bx)\bc^{\star}(\bx), \bx = x_{\rand}] + \Ex_{\bx \sim \bmcD}[\bar{h}(\bx)\bc^{\star}(\bx) , \bx \neq x_{\rand}] \\
        &=\Ex_{\bx \sim \bmcD}[\bar{h}(\bx)\bc^{\star}(\bx) , \bx \neq x_{\rand}] \tag{$\Ex[\bc^{\star}(x_{\rand}) \mid \bmcDc]= 0$}\\
        &=\Ex_{\bx \sim \bmcD}[\bar{h}(\bx)\bc^{\star}(\bx) , \bx = \bx_{\wrong}]  + \Ex_{\bx \sim \bmcD}[\bar{h}(\bx)\bc^{\star}(\bx) , \bx \in \bS_{\correct}]. 
    \end{align*}
    The above is easiest to analyze for each $x$ separately. The probability any individual $x$ is $\bx_{\wrong}$ is $\frac{1}{d-1}$. If it is, then $\bmcD$ puts weight $\eta$ on $x$, and $\bar{h}(x)\bc^{\star}(x) = -\bar{h}(x)\by(x)$. If $x$ is not $\bx_{\wrong}$, then $\bmcD$ puts weight $\frac{1-2\eta}{d-3}$ on $x$ and $\bar{h}(x)\bc^{\star}(x) = \bar{h}(x)\by(x)$. Hence,
    \begin{equation*}
         \Ex_{\bx \sim \bmcD}[\bar{h}(\bx)\bc^{\star}(\bx)]  = \sum_{x \neq x_{\rand}} \frac{\eta}{d-1} \cdot (-\bar{h}(x)\by(x)) + \frac{d-2}{d-1} \cdot \frac{1-2\eta}{d-3} \cdot \bar{h}(x)\by(x).
    \end{equation*}

    Under the constraint that $\eta \leq 1/3 - O(1/d)$, the above is maximized by setting $\bar{h}(x) = \by(x)$, in which case
    \begin{equation*}
         \Ex_{\bx \sim \bmcD}[\bar{h}(\bx)\bc^{\star}(\bx)] \leq (d-1)\cdot \paren*{\frac{-\eta}{d-1} + \frac{(d-2)(1-2\eta)}{(d-1)(d-3)}} = \frac{d-2}{d-3}\cdot (1-2\eta) - \eta.
    \end{equation*}
    Finally we use that $\Pr[\bh(x) \neq \bc^{\star}(x)] =\frac{1-\bar{h}(x)\bc^{\star}(x)}{2}$,
    \begin{equation*}
        \Pr[\bh(\bx) \neq \bc^{\star}(\bx)] \geq \frac{1}{2}\cdot \paren*{1 - \frac{d-2}{d-3}\cdot (1-2\eta) + \eta} \geq \frac{3\eta}{2} - O(1/d). 
        \qedhere
    \end{equation*}

\end{proof}

\section{Fixed-distribution learning with nasty noise, proof of \Cref{thm:nasty-fixed-dist-intro}}
\label{sec:nasty-new-proxy}
In this section, we prove the following.
\begin{theorem}[Fixed-distribution learning with nasty noise, formal version of \Cref{thm:nasty-fixed-dist-intro}]
    \label{thm:nasty-fixed-dist-body}
    For every concept class $\mcC$ and distribution over unlabeled examples $\mcD$, there exists an algorithm $A_{\mcC, \mcD}$ which $(\eta + \eps, \delta)$-learns $\mcC$ with $\eta$-nasty noise over the fixed distribution $\mcD$ using $n \coloneqq O\paren*{\frac{d + \log(1/\delta)}{\eps^2}}$ samples.
\end{theorem}

As discussed in \Cref{subsec:overview-fixed-nasty}, the first step is to design a way of certifying that if the true target function were $c$, then the adversary would have needed to corrupt at least $\eta(c,\bS, \mcD)$ fraction of the dataset.
\begin{claim}[Lower bounding corruptions]%
    \label{claim:lb-corruption-formal}
    For any functions $c:X \to \bits$ and $f:X \to [-1,1]$, define
    \begin{align*}
        \eta(c, f, S ,\mcD) &\coloneqq  \Prx_{(\bx, \by) \sim \Unif(S)}[c(\bx) \neq \by \text{ }\mathrm{or}\text{ } \boldf(\bx) = 1] - \Prx_{\bx \sim \mcD}[\boldf(\bx) = 1] \quad \text{where }\boldf(x) \sim \Rad(f(x)) \\
        &=\frac{1}{4} \cdot \Ex_{(\bx, \by) \sim \Unif(S)}\bracket*{1 + f(\bx) + c(\bx)\by(f(\bx) - 1)} - \frac{1}{2} \cdot \Ex_{\bx \sim \mcD}[f(\bx)].
    \end{align*}
    Then, for any function class $\mcF$ mapping to $[-1,1]$ and function $c^{\star}:X \to \bits$, with probability $1-\delta$ over $\bS^{\circ}$ consisting of $n$ i.i.d. samples of the form $(\bx, c^{\star}(\bx))_{\bx \sim \mcD}$, all datasets $S$ satisfy
    \begin{equation*}
        \dist(S, \bS^{\circ}) \geq \sup_{f \in \mcF}\set*{\eta(c^{\star},f,S, \mcD)} - O\paren*{\Rad_{n, \mcD}(\mcF) + \sqrt{\frac{\log(1/\delta)}{n}}},
    \end{equation*}
    where $\dist(S, S^{\circ})$ is the minimum fraction of points within $S^{\circ}$ that must be corrupted to form $S$.
\end{claim}
The second ingredient is a recipe for choosing the certifying function $f$ that has the following guarantee, on average over any distribution $\mu$.
\begin{lemma}[Lower bounding corruptions on average]
    \label{lem:nasty-fixed-dist-avg}
    For any distribution $\mu$ over functions mapping $X \to \bits$ let $\bar{\mu}$ be the function mapping $x$ to $\Ex_{\bc \sim \mu}[\bc(x)]$, define
    \begin{equation*}
        f_{c, \mu}(x) \coloneqq \begin{cases}
            \frac{c(x)\bar{\mu}(x) - 1}{c(x)\bar{\mu}(x) + 1} &\text{if }c(x)\bar{\mu}(x) \geq 0 \\
            -1 &\text{otherwise}.
        \end{cases}
    \end{equation*}
    Then, for any dataset $S$ and distribution $\mcD$,
    \begin{equation*}
        \Prx_{\bx \sim \mcD, \bc \sim \mu}[\sign(\bar{\mu}(\bx)) \neq \bc(\bx)] \leq \Ex_{\bc \sim \mu}[\eta(\bc, f_{\bc, \mu}, S, \mcD)].
    \end{equation*}
\end{lemma}
We'll also use the following.
\begin{restatable}[Approximating with $k$-wise majority]{claim}{MajClaim}
    \label{claim:approx-with-maj}
    Let $\mu$ be any distribution over functions mapping $X$ to $\pm 1$ with mean $\bar{\mu} \coloneqq \Ex_{\bc \sim \mu}[\bc]$ and $\mcD$ be any distribution over $X$. Then, for any $k = O(1/\eps^2)$
    \begin{equation*}
        \Ex_{\bc\sim \mu, \bx \sim \mcD}\bracket*{\Ex_{\bc_1, \ldots, \bc_k \iid \mu}\bracket*{\bh(\bx) \bc(\bx) : \bh \leftarrow \sign\paren*{\lfrac{\bc_1 + \cdots + \bc_k}{k}}}} \geq \Ex_{\bc \sim \mu, \bx \sim \mcD}\bracket*{\sign \circ \bar{\mu}(\bx) \bc(\bx)} - \eps.
    \end{equation*}
\end{restatable}
\pparagraph{Structure of this section.} We prove \Cref{claim:lb-corruption-formal} in \Cref{subsec:proof-of-lb-corrupt}, \Cref{lem:nasty-fixed-dist-avg} in \Cref{subsec:proof-of-nasty-fixed-avg}, and  \Cref{claim:approx-with-maj} in \Cref{subsec:proof-of-approx-with-maj}. For now, we show how these results imply our fixed-distribution nasty noise learner.
\begin{proof}[Proof of \Cref{thm:nasty-fixed-dist-body} assuming \Cref{claim:lb-corruption-formal}, \Cref{lem:nasty-fixed-dist-avg}, and \Cref{claim:approx-with-maj}]
    Define the $2$-Lipschitz function,
    \begin{equation*}
        \phi(t) = \begin{cases}
            \frac{t-1}{t + 1} &\text{if }t \geq 0 \\
            -1 &\text{otherwise}
        \end{cases}
    \end{equation*}
    and set
    \begin{equation*}
        \mcF \coloneqq \set{\phi \circ (c\cdot \bar{\mu}) \mid c \in \mcC, \bar{\mu}\in \Conv(\mcC)}.
    \end{equation*}
    Our learner will output any randomized hypothesis $\bh$ satisfying
    \begin{equation}
        \label{eq:nasty-error-by-proxy}
        \Prx_{\bx \sim \mcD}[\bh(\bx) \neq c(\bx)] \leq \sup_{f \in \mcF}\set*{\eta(c, f, \bS, \mcD)} + \eps/2 \quad\quad\text{for all }c \in \mcC.
    \end{equation}
    The remainder of this proof is divided into two pieces. First, we show such a hypothesis $\bh$ that is supported on $\Maj_k(\mcC)$ for $k = O(1/\eps^2)$ always exists. Since we are not bounding the runtime of the learner, this suffices to ensure that the algorithm can output such a hypothesis. Second, we show that such hypotheses indeed have the desired error bound.

    \pparagraph{Existence of a good hypothesis.} We will show that there is some $\bar{h} \in \Conv(\Maj_k(\mcC))$ s.t. the randomized hypothesis $\bh(x) = \Rad(\bar{h}(x))$ satisfies \Cref{eq:nasty-error-by-proxy}. Define,
    \begin{equation*}
        \ell(\bar{h}, c) \coloneqq \Prx_{\bx \sim \mcD}[\Rad(\bar{h}(\bx)) \neq c(\bx)] - \sup_{f \in \mcF}\set*{\eta(c, f, \bS, \mcD)}.
    \end{equation*}
    Then, by the minimax lemma (\Cref{cor:minimax}),
    \begin{equation*}
        \inf_{\bar{h} \in \Conv(\Maj_k(\mcC))} \set*{\sup_{c \in \mcC} \set*{\ell(\bar{h}, c)}} = \sup_{\text{distribution }\mu\text{ over }\mcC} \set*{ \inf_{\bar{h} \in \Conv(\Maj_k(\mcC))} \set*{\Ex_{\bc \sim \mu}[\ell(\bar{h}, \bc)]}}.
    \end{equation*}
    We will upper bound the right-hand side of the above equation by $\eps/2$ for any choice of $\mu$. Since the function $f_{c, \mu}$ used in \Cref{lem:nasty-fixed-dist-avg} is contained in $\mcF$, it gives that
    \begin{equation*}
        \Ex_{\bc \sim \mu}[\ell(\sign(\bar{\mu}), \bc)] \leq 0.
    \end{equation*}
    where $\bar{\mu} =\Ex_{\bc \sim \mu}[\bc]$. Then, \Cref{claim:approx-with-maj} gives that there is some $\bar{h} \in \Conv(\Maj_k(\mcC))$ with $\Ex_{\bc \sim \mu}[\ell(\bar h, \bc)] \leq \eps/2$ as desired.

    \pparagraph{Any good hypothesis has small error.} We will show that any $\bh$ satisfying \Cref{eq:nasty-error-by-proxy} satisfies, with probability at least $1-\delta$ that
    \begin{equation*}
        \Prx_{\bx \sim \mcD}[\bh(\bx) \neq c^{\star}(\bx)] \leq \eta + \eps.
    \end{equation*}
    Applying \Cref{eq:nasty-error-by-proxy} with $c = c^{\star}$, we get that
    \begin{align*}
        \Prx_{\bx \sim \mcD}[\bh(\bx) \neq c^{\star}(\bx)]  &\leq \sup_{f \in \mcF}\set*{\eta(c^{\star}, f, \bS, \mcD)} + \eps/2 \\
        &\leq \eta + \eps/2 + O\paren*{\Rad_{n, \mcD}(\mcF) + \sqrt{\frac{\log(1/\delta)}{n}}} \tag{\Cref{claim:lb-corruption-formal}}
    \end{align*}
    Finally, by combining \Cref{fact:sym-VC} with \Cref{fact:Rad-compose},\footnote{In the below, we are assuming without loss of generality that $\VC(\mcC) \geq 1$. If $\VC(\mcC) = 0$, it can only contain one function and learning is trivial.}
    \begin{equation*}
        \Rad_{n,\mcD}(\mcF) \leq O\paren*{\sqrt{\frac{\VC(\mcC)}{n}}}
    \end{equation*}
    For our choice of $n$, the term $O\paren*{\Rad_{n, \mcD}(\mcF) + \sqrt{\frac{\log(1/\delta)}{n}}}$ is at most $\eps/2$, giving the desired result.
\end{proof}

\begin{remark}[Sample access to $\mcD$]
    \label{remark:access-to-dist}
    The only information about $\mcD$ that $A$ needs is $O(\eps)$-accurate estimates of $\eta(c, f, S, \mcD)$ for all $c\in \mcC$ and $f \in \mcF$, as well as $O(\eps)$ accurate estimates of $\Prx_{\bx \sim \mcD}[\bh(x) \neq c(x)]$ for all $\bh$ supported on $\Maj_{O(1/\eps^2)}(\mcC)$ and $c \in \mcC$. All of these quantities can be estimated with access to an unlabeled i.i.d. sample from $\mcD$ of size $\tilde{O}((d + \log(1/\delta))/\eps^4)$.
\end{remark}

\subsection{Proof of \Cref{claim:lb-corruption-formal}}
\label{subsec:proof-of-lb-corrupt}
\begin{proof}
    Set 
    $$\eps \coloneqq O\paren*{\Rad_{n, \mcD}(\mcF) + \sqrt{\frac{\log(1/\delta)}{n}}}.$$
    We first show that with probability at least $1-\delta$ that $\sup_{f \in \mcF}\set*{\eta(c^{\star},f,\bS^{(\circ)}, \mcD)} \leq \eps$. This proves the special case where $S = \bS^{\circ}$. We'll then use a Lipschitzness argument to complete the proof.

    For this special case, we observe that since $\bS^{\circ}$ is i.i.d. examples from $\mcD$ consistent with $c^{\star}$,
    \begin{align*}
        \eta(c^{\star}, f, \bS^{\circ} ,\mcD) &=\frac{1}{4} \cdot \Ex_{(\bx, \by) \sim \Unif(\bS^{(\circ)})}\bracket*{1 + f(\bx) + c^{\star}(\bx)c^{\star}(\bx)(f(\bx) - 1)} - \frac{1}{2} \cdot \Ex_{\bx \sim \mcD}[f(\bx)] \\
        &= \frac{1}{2} \cdot \Ex_{(\bx, \by) \sim \Unif(\bS^{(\circ)})}\bracket*{f(\bx)} - \frac{1}{2} \cdot \Ex_{\bx \sim \mcD}[f(\bx)] \tag{$c^{\star}(x)^2 = 1$.}
    \end{align*}
    \Cref{fact:rad-gen} therefore gives that $\sup_{f \in \mcF} \set*{\eta(c^{\star}, f, \bS^{\circ} ,\mcD)} > \eps$ with probability at most $\delta$.

    We now consider the more general case of arbitrary $S$. We observe that
    \begin{equation*}
       \abs*{T_f(x,y)} \leq 1/2 \quad\text{where} \quad  T_f(x,y) \coloneqq \frac{1 + f(x) + c^{\star}(x)y(f(x) - 1)}{4}.
    \end{equation*}
    for any $f$ bounded on $[-1,1]$ and $c^{\star}(x), y \in \bits$. This is a straightforward calculation: If $c^{\star}(x) = y$, then $T_f(x,y) = f(x)/2$ which satisfies the desired bound. Otherwise, if $c^{\star}(x) \neq y$, then $T_f(x,y) = 1/2$ which also satisfies the desired bound. Then, for any $S$,
    \begin{equation*}
        \eta(c^{\star},f,S, \mcD) = \Ex_{(\bx, \by) \sim S}[T_f(x,y)] - \Ex_{(\bx, \by) \sim \bS^{\circ}}[T_f(x,y)] + \eta(c^{\star}, f, \bS^{\circ} ,\mcD).
    \end{equation*}
    The desired result follows by observing that $ \Ex_{(\bx, \by) \sim S}[T_f(\bx,\by)] - \Ex_{(\bx, \by) \sim \bS^{\circ}}[T_f(\bx,\by)]  \leq \dist(S, \bS^{\circ})$, since the maximum and minimum of $T$ differ by $1$, and our earlier bound that $\sup_{f \in \mcF} \set*{\eta(c^{\star}, f, \bS^{\circ} ,\mcD)} > \eps$ with probability at most $\delta$.
\end{proof}

\subsection{Proof of \Cref{lem:nasty-fixed-dist-avg}}
\label{subsec:proof-of-nasty-fixed-avg}
Rather than directly prove the result of \Cref{lem:nasty-fixed-dist-avg}, we will describe how to derive the lemma statement (in particular, our choice of $f_{c, \mu}$). Our ultimate goal is to ensure that
\begin{equation*}
    \Prx_{\bx \sim \mcD, \bc \sim \mu}[h(\bx) \neq \bc(x)] \leq \Ex_{\bc \sim \mu}[\eta(\bc, f_{\bc, \mu}, S, \mcD)].
\end{equation*}
Setting $h(x) = \sign(\bar{\mu}(x))$ is the obvious choice as that minimizes the left-hand side and does not affect the right-hand side. We will set $f_{c, \mu}(x) = f(c(x), \bar{\mu}(x))$ for some function $f:\bits \times [-1,1] \to [-1,1]$ we aim to derive.
With this choice, we then expand the definition of $\eta$ and, after some straightforward algebraic manipulations, arrive at a goal of
\begin{align*}
    &\Ex_{\bx \sim \mcD, \bc \sim \mu} \bracket*{1 - \sign(\bar{\mu}(\bx)) \bc(\bx) + f(\bc(\bx), \bar{\mu}(\bx))} \\
    &\quad\quad\quad\quad\leq \frac{1}{2} \cdot \Ex_{(\bx,\by) \sim \Unif(S), \bc \sim \mu}\bracket*{1 + f(\bc(\bx), \bar{\mu}(\bx)) + \bc(\bx)\by(f(\bc(\bx), \bar{\mu}(\bx)) - 1)}.
\end{align*}
A challenge with satisfying the above equation is that we are viewing $S$ as a worst-case sample, and there is no guarantee that it is representative of $\mcD$. Therefore, the behavior of $\mu$ on $\mcD$ may be completely different than its behavior on $S$. To handle this, for some constant value $v$, we will ensure the left-hand side of the above equation is \textsl{always} at most $v$, and the right-hand side is \textsl{always} at least $v$. To accomplish this, it is sufficient (and indeed necessary) for
\begin{align*}
    &1 - \abs*{\bar{\mu}} + \Ex_{\bc \sim \Rad(\bar{\mu})}[f(\bc, \bar{\mu})] \leq v \quad \text{for all }\bar{\mu} \in [-1,1]\\
    &\abs*{\bar{\mu} - \Ex_{\bc \sim \Rad(\bar{\mu})}[\bc \cdot f(\bc, \bar{\mu})]} \leq 1 - 2v +  \Ex_{\bc \sim \Rad(\bar{\mu})}[f(\bc, \bar{\mu})] \quad \text{for all }\bar{\mu} \in [-1,1].
\end{align*}
    It turns out, the above constraints have exactly one solution\footnote{Not counting additional solutions from setting $f(-1,1)$ and $f(1,-1)$ arbitrarily because they don't affect the constraints} among $[-1,1]$-valued $f$, which is satisfied by setting $v = 0$ and
    \begin{equation}
        \label{eq:def-f}
        f(c, \bar{\mu}) \coloneqq \begin{cases}
            \frac{c\bar{\mu}- 1}{c\bar{\mu} + 1} &\text{if }c\bar{\mu} \geq 0 \\
            -1 &\text{otherwise}.
        \end{cases}
\end{equation}

We now proceed to give a formal proof of \Cref{lem:nasty-fixed-dist-avg}. We first observe the below expected values.
\begin{proposition}
    \label{prop:compute-exp}
    For $f$ defined in \Cref{eq:def-f} and any $\bar{\mu} \in [-1,1]$,
    \begin{equation*}
        \Ex_{\bc \sim \Rad(\bar{\mu})}[f(\bc, \bar{\mu})] = \abs*{\bar{\mu}} - 1 \quad\quad\text{and}\quad\quad \Ex_{\bc \sim \Rad(\bar{\mu})}[\bc \cdot f(\bc, \bar{\mu})] =0.
    \end{equation*}
\end{proposition}
\begin{proof}
    Both are straightforward calculations. For the first:
    \begin{equation*}
        \Ex_{\bc \sim \Rad(\bar{\mu})}[f(\bc, \bar{\mu})]  = \frac{1 + \abs*{\bar{\mu}}}{2} \cdot \frac{\abs*{\bar{\mu}} - 1}{\abs*{\bar{\mu}} + 1} + \frac{1 - \abs*{\bar{\mu}}}{2} \cdot (-1) = \abs*{\bar{\mu}} - 1
    \end{equation*}
    The second:
    \begin{equation*}
        \Ex_{\bc \sim \Rad(\bar{\mu})}[\bc \cdot f(\bc, \bar{\mu})]  = \frac{1 + \abs*{\bar{\mu}}}{2} \cdot \frac{\abs*{\bar{\mu}} - 1}{\abs*{\bar{\mu}} + 1} \cdot \sign(\bar{\mu}) + \frac{1 - \abs*{\bar{\mu}}}{2} \cdot (-1) \cdot (-\sign(\bar{\mu}))= 0.
    \end{equation*}
\end{proof}

We are now ready to prove the main result of this subsection.
\begin{proof}[Proof of \Cref{lem:nasty-fixed-dist-avg}]
    Our goal is to prove that
    \begin{equation*}
         \Prx_{\bx \sim \mcD, \bc \sim \mu}[\sign(\bar{\mu}(\bx)) \neq \bc(\bx)] \leq \Ex_{\bc \sim \mu}[\eta(\bc, f_{\bc, \mu}, S, \mcD)].
    \end{equation*}
    Expanding both sides, setting $f_{c,\mu}(x) = f(c(x), \bar{\mu}(x))$, and using that $\bc(x) \sim \Rad(\bar{\mu}(x))$ for any $x$,
    \begin{align*}
        \Ex_{\bx \sim \mcD}\bracket*{\frac{1 - \abs*{\bar{\mu}(\bx)}}{2}} &\leq \frac{1}{4} \cdot \Ex_{(\bx, \by) \sim \Unif(S), \bc \sim \Rad(\bar{\mu}(\bx))}\bracket*{1 + f(\bc, \bar{\mu}(\bx)) + \bc\by(f(\bc, \bar{\mu}(\bx)) - 1)} \\
        &\quad- \frac{1}{2} \cdot \Ex_{\bx \sim \mcD, \bc \sim \Rad(\bar{\mu}(\bx))}[f(\bc, \bar{\mu}(\bx))].
    \end{align*}
    Moving the last term to the left-hand side and using both expectations computed in \Cref{prop:compute-exp}, this is equivalent to
    \begin{equation*}
         \Ex_{\bx \sim \mcD}\bracket*{\frac{1 - \abs*{\bar{\mu}(\bx)}}{2} + \frac{\abs*{ \bar{\mu}(\bx)} - 1}{2}} \leq \frac{1}{4} \cdot \Ex_{(\bx, \by) \sim \Unif(S)} \bracket*{1 + \paren*{\abs*{ \bar{\mu}(\bx)} - 1} + 0 - \by \bar{\mu}(\bx)}.
    \end{equation*}
    The left-hand side is always exactly $0$. Using that $-y \mu \geq -\abs*{\mu}$, the right-hand side is at least $0$, which gives the desired bound.
\end{proof}

\subsection{Proof of \Cref{claim:approx-with-maj}}
\label{subsec:proof-of-approx-with-maj}
    It suffices to show the desired inequality holds individually for each $x$. Let $\tilde{\bmu}(x)$ be the random variable taking on value $\lfrac{\bc_1(x) + \cdots + \bc_k(x)}{k}$. Then, we wish to show for any $\bar{\mu}(x) \in [-1,1]$ that
    \begin{equation*}
        \Ex\bracket*{\sign(\tilde{\bmu}(x))\bar{\mu}(x)} \geq \sign(\bar{\mu}(x)) \bar{\mu}(x) - \frac{2}{\sqrt{k}}.
    \end{equation*}
    We will use that $\tilde{\bmu}(x)$ has mean $\bar{\mu}(x)$ and variance at most $1/k$ to show that
    \begin{equation*}
        \Ex\bracket*{\abs*{\paren*{\sign(\tilde{\bmu}(x)) - \sign(\bar{\mu}(x))} \bar{\mu}(x)}} \leq \frac{2}{\sqrt{k}},
    \end{equation*}
    which easily implies the desired result. For $\sign(a)$ and $\sign(b)$ to differ, we must have that $|a - b| \geq |b|$, and if they differ, it is by a value of magnitude $2$. Therefore,
    \begin{align*}
         \Ex\bracket*{\abs*{\paren*{\sign(\tilde{\bmu}(x)) - \sign(\bar{\mu}(x))} \bar{\mu}(x)}} &\leq \Ex\bracket*{2 \cdot \Ind\bracket[\big]{\abs*{\tilde{\bmu}(x) - \bar{\mu}(x)} \geq \abs*{\bar{\mu}(x)}} \cdot \abs*{\bar{\mu}(x)}}\\
         &= \Ex\bracket*{2 |\bar{\mu}(x)|\cdot \Ind\bracket[\big]{\paren*{\tilde{\bmu}(x) - \bar{\mu}(x)}^2 \geq \bar{\mu}(x)^2}} \\
         &\leq \Ex\bracket*{\frac{2 |\bar{\mu}(x)|}{\bar{\mu}(x)^2} \cdot \paren*{\tilde{\bmu}(x) - \bar{\mu}(x)}^2 } \\
         &= \frac{2 \Var[\tilde{\bmu}(x)]}{|\bar{\mu}(x)|} \leq \frac{2}{k\cdot|\bar{\mu}(x)|}
    \end{align*}
    If $\abs*{\bar{\mu}(x)} \geq 1/\sqrt{k}$, the above gives the desired bound. Otherwise, if $\abs*{\bar{\mu}(x)} \leq 1/\sqrt{k}$, we can use the simpler bound of $2 \abs*{\bar{\mu}(x)}$ which follows from the first line of the above equation block, which also implies the desired result. \qed

\section{Learning with nasty classification noise: Proof of \Cref{thm:agnostic-intro}}
\label{sec:agnostic}

In this section, we prove the following.
\begin{theorem}[Optimal learning with nasty classification noise, formal version of \Cref{thm:agnostic-intro}]
    \label{thm:agnostic-body}
    For any concept class $\mcC$ with VC dimension $d$ and $\eps, \delta \in (0,1), \alpha \in [0,1)$, there is an algorithm that $((1 + \alpha)\cdot \eta + \eps,\delta)$ learns $\mcC$ with $\eta$-nasty classification noise using
    \begin{equation*}
        n\coloneqq O\paren*{\frac{d \cdot \log(2 + \alpha/\eps)\log(1/(\eps + \alpha))+ \log(1/\delta)}{ \eps (\alpha + \eps)^3}} = \tilde{O}\paren*{\frac{d + \log(1/\delta)}{\eps(\alpha + \eps)^3}}
    \end{equation*}
    samples. Furthermore, this algorithm is efficiently implementable using $O\paren*{\frac{1}{\eps(\eps+\alpha)}}$ calls to an ERM oracle for $\mcC$.
\end{theorem}
We encourage the reader to consider two regimes. \Cref{thm:agnostic-body} shows that error $1.01 \eta + \eps$ can be accomplished with $\tilde{O}\paren*{\frac{d + \log(1/\delta)}{\eps}}$ samples and $O(1/\eps)$ ERM calls, and error $\eta + \eps$ with $\tilde{O}\paren*{\frac{d + \log(1/\delta)}{\eps^4}}$ samples and $O(1/\eps^2)$ ERM calls.

As with our proof of \Cref{thm:malicious-intro} we begin by defining a linearized loss function.
\begin{claim}[Linear loss for agnostic noise]
    \label{claim:agnostic-loss}
    For any dataset $S \in (X \times \bits)^n$, let $\bar{h}:X \to [-1,1]$ be any function satisfying, for all $c \in \mcC$,
    \begin{equation}
        \label{eq:agnostic-loss}
        \ell_S(c, \bar{h}) \leq \eps \quad\quad\text{where } \ell_S(c, \bar{h}) \coloneqq\Ex_{\bx,\by \sim \Unif(S)}[c(\bx)((1 + \alpha)\by - \bar{h}(\bx))] - \alpha.
    \end{equation}
    Then, the randomized function $\bh = \Rad(\bar{h})$ satisfies
    \begin{equation*}
        \Prx_{(\bx, \by) \sim \Unif(S), \bh}[\bh(\bx) \neq c(\bx)] \leq (1 + \alpha) \cdot \Prx_{(\bx, \by) \sim \Unif(S)}[c(\bx) \neq \by] + \eps/2
    \end{equation*}
    simultaneously for all $c \in \mcC$.
\end{claim}
Second, we need to choose a function that satisfies our loss on average for any distribution $\mu$ over concepts.
\begin{claim}[Majority has low loss]
    \label{claim:agnostic-maj-low-loss}
    Let $k = O(1/(\alpha + \eps)^2)$. Then, $g \coloneqq M_k$ as defined in \Cref{claim:maj-suffices}, satisfies, for any distribution $\mu$ over functions and dataset $S$,
    \begin{equation*}
         \Ex_{\bc \sim \mu}\bracket*{\ell_S(\bc, \bar{h})} \leq \eps \quad\quad\text{where}\quad \bar{h}(x) = g\paren*{\Ex_{\bc \sim \mu}[\bc(x)]}.
    \end{equation*}
\end{claim}

We will also need to upper bound the Lipschitz constant of this choice of $g$.
\begin{proposition}[Lipschitzness of majority]
    \label{prop:lipschitz-constant-maj}
    For any odd $k$, $M_k$ is $O(\sqrt{k})$-Lipschitz.
\end{proposition}

\textbf{Structure of this section.} We prove \Cref{claim:agnostic-loss} in \Cref{subsec:proof-of-agnostic-loss}, \Cref{claim:agnostic-maj-low-loss} in \Cref{subsec:proof-of-agnostic-maj}, and \Cref{prop:lipschitz-constant-maj} in \Cref{subsec:lipschitz-maj}. For now, we show how to use them to prove \Cref{thm:agnostic-body}.
\begin{proof}
    Set $k = O(1/(\alpha + \eps)^2)$, and define the slightly shifted loss $\tilde{\ell}_S(c, \bar{h}) = \ell_S(c,\bar{h}) - \eps$ (this is so that \Cref{claim:agnostic-maj-low-loss} gives a loss of $0$ rather than $\eps$). We will use \Cref{thm:general-optimization} to find a $\mu \in \Conv(\mcC)$ so that $\overline{h} \coloneqq M_k \circ \mu$ satisfies $\tilde{\ell}_S(c,\bar{h}) \leq \eps$. We verify the assumptions of \Cref{thm:general-optimization} hold: The loss $\tilde{\ell}_S$ corresponds to function
    \begin{equation*}
        f(c,h,y) = c((1 + \alpha)y - h) - \alpha - \eps,
    \end{equation*}
    which has nonpositive agreement coefficients and has $\linf{f} \leq O(1)$. Furthermore it is $(g \coloneqq M_k)$-feasible by \Cref{claim:agnostic-maj-low-loss} and this $g$ is $O(1/(\alpha + \eps))$-Lipschitz by \Cref{prop:lipschitz-constant-maj}. Taken together, the algorithm of \Cref{thm:general-optimization} will converge in $O(1/(\eps \cdot (\alpha + \eps)))$ iterations to a $\mu \in \Conv(\mcC)$ for which $\overline{h} \coloneqq M_k \circ \mu$ satisfies $\tilde{\ell}_S(c,\bar{h}) \leq \eps$ and therefore $\ell_S(c,\bar{h}) \leq 2\eps$. Each such iteration can be done with one ERM call by \Cref{fact:weighted-ERM}.

    We output $\bh = \Rad(\bar{h})$. Observe that this function can be written as a mixture of functions within $\Maj_k(\mcC)$: Taking $\mcD_c$ to the distribution over $\mcC$ corresponding to $\mu$, we can sample $\bc_1, \ldots,\bc_k \iid \mcD_c$ and set $\bh = \Maj(\bc_1, \ldots, \bc_k)$. The class $\Maj_k(\mcC)$ has VC dimension $O(dk\log k) = O(\frac{d\log(1/(\eps + \alpha))}{(\eps + \alpha)^2})$ by \Cref{fact:maj-VC}. Therefore, with probability at least $1-\delta$ by \Cref{cor:VC-random}, using $\bS^{\circ}$ to denote the uncorrupted sample and the fact that the unlabeled portions of $\bS$ and $\bS^{\circ}$ agree,
    \begin{equation*}
        \error_{\mcD}(\bh, c^{\star}) \leq (1 + \alpha) \cdot \Prx_{(\bx,\by) \sim \Unif(\bS^{\circ}),\bh}[\bh(\bx) \neq c^{\star}(\bx)] + \eps = (1 + \alpha) \cdot \Prx_{(\bx,\by) \sim \Unif(\bS),\bh}[\bh(\bx) \neq c^{\star}(\bx)] + \eps.
    \end{equation*}
    Recall that the number of points the adversary corrupts is distributed according to $\Bin(n,\eta)$ (\Cref{def:nasty-classification-noise,def:nasty-noise-basic}) and therefore by \Cref{cor:Chernoff}, with probability at least $1-\delta$, the fraction of points the adversary could corrupt, $\boldeta$ is at most $(1+ \alpha)\eta + \eps$. Combining this with the prior bound, with probability at least $1 - 2\delta$,
    \begin{align*}
	         \error_{\mcD}(\bh, c^{\star}) &\leq (1 + \alpha) \cdot \Prx_{(\bx,\by) \sim \Unif(\bS),\bh}[\bh(\bx) \neq c^{\star}(\bx)] + \eps \\
	         &\leq (1 + \alpha) \cdot \paren*{(1 + \alpha) \cdot \Prx_{(\bx, \by) \sim \Unif(\bS)}[c^{\star}(\bx) \neq \by] + \ell_S(c,\overline{h})/2} + \eps \tag{\Cref{claim:agnostic-loss}} \\
	         &\leq (1 + \alpha) \cdot \paren*{(1 + \alpha) \cdot \boldeta + \eps} + \eps \tag{$c^{\star}$ only wrong on corrupted points, $\ell_S \leq 2\eps$} \\
	         &\leq  (1 + \alpha) \cdot \paren*{(1 + \alpha) \cdot ((1 + \alpha)\eta + \eps) + \eps} + \eps \tag{$\boldeta \leq (1+\alpha)\eta + \eps$}\\
	         &=  (1 + O(\alpha))\eta + O(\eps).
    \end{align*}
    By adjusting $\alpha$, $\eps$ and $\delta$ by constants, we can conclude that $\error_{\mcD}(\bh,c^{\star})$ is at most $(1 + \alpha)\eta + \eps$ with probability at least $1-\delta$.
\end{proof}

\subsection{Proof of \Cref{claim:agnostic-loss}}
\label{subsec:proof-of-agnostic-loss}
Fix any $c \in \mcC$. Then, a rearrangement of \Cref{eq:agnostic-loss} gives that
\begin{equation*}
    -\Ex[c(\bx)\bar{h}(\bx)] = \ell_S(c,\bar{h}) + \alpha - (1+\alpha) \cdot \Ex[c(\bx)\by]
\end{equation*}
Since $c(x)\bar{h}(x) = 1 - 2\Pr[c(x)\neq \bh(x)]$ and likewise for $c(x)y$, this is equivalent to
\begin{equation*}
    2\Pr[c(\bx) \neq \bh(\bx)] - 1 = \ell_S(c,\bar{h}) + \alpha  + (1 + \alpha)\cdot (2 \Pr[c(\bx) \neq \by] -1),
\end{equation*}
which can simply be written as
\begin{equation*}
    2\Pr[c(\bx) \neq \bh(\bx)] = (1 + \alpha) \cdot 2 \Pr[c(\bx) \neq \by] + \ell_S(c,\bar{h}),
\end{equation*}
which gives the desired result under the assumption $\ell_S(c, \bar{h}) \leq \eps$ \qed

\subsection{Proof of \Cref{claim:agnostic-maj-low-loss}}
\label{subsec:proof-of-agnostic-maj}
This proof is quite similar to \Cref{claim:approx-with-maj}. That simpler proof roughly corresponds to the $\alpha=0$ case of this proof.

\begin{proof}
We show that for any fixed choice of $x,y$ that
\begin{equation*}
    \Ex_{\bc \sim \mu}[\bc(x)((1 + \alpha)y - \bar{h}(x))] - \alpha \leq \eps,
\end{equation*}
which implies the desired result by taking an expectation over $(\bx,\by) \sim \Unif(S)$. Using $\bar{\mu}$ to denote, $\Ex_{\bc \sim \mu}[\bc(x)]$, we wish to show
\begin{equation*}
    \bar{\mu}((1 + \alpha)y - M_k(\bar{\mu})) - \alpha \leq \eps.
\end{equation*}
This equation is hardest to satisfy when $y = \sign(\bar{\mu})$, in which case it can be rewritten,
\begin{equation*}
    \bar{\mu} \cdot (\sign(\bar{\mu}) - M_k(\bar{\mu}) ) \leq  \eps + \alpha\cdot (1-|\bar{\mu}|).
\end{equation*}
By definition, $M_k(\bar{\mu}) = \Ex[\sign(\tilde{\bmu})]$  where $\bz_1, \ldots, \bz_k \iid \Rad(\bar{\mu})$ and $\tilde{\bmu} = \frac{\bz_1 + \cdots \bz_k}{k}$. Using this notation, we wish to show,
\begin{equation}
    \label{eq:agnostic-loss-main-bound}
     \Ex[\bar{\mu} \cdot (\sign(\tilde{\bmu}) - \sign(\bar{\mu}))] \leq  \eps + \alpha\cdot (1-|\bar{\mu}|).
\end{equation}
We will show the above holds by considering three regimes on how large $|\bar{\mu}|$ is.

\pparagraph{Small regime:} If $\abs{\bar{\mu}} \leq \min(\eps + \alpha/2, 1/2)$, then the left-hand side of \Cref{eq:agnostic-loss-main-bound} is upper bounded by $\abs{\bar{\mu}} \leq \eps + \alpha/2$ and the right-hand side must be larger.

\pparagraph{Large regime:} If $|\bar{\mu}| \geq 3/4$, we will use that 
\begin{equation*}
    \Pr[\sign(\bar{\mu}) \neq \sign(\tilde{\bmu})] = \Pr[\Bin(k, (1 - \abs*{\bar{\mu}})/2) \geq k/2] \leq 2^k \cdot \paren*{\frac{1 - \abs*{\bar{\mu}}}{2}}^{k/2},
\end{equation*}
where the last inequality uses that in order for $\Bin(k,p) \geq t$, there must exist a choice of $\binom{k}{t}$ indices that are all sampled positively (each of which occurs with probability $p^t$) along with the crude bound that $\binom{k}{t} \leq 2^k$. Therefore,
\begin{align*}
	     \Ex[\bar{\mu} \cdot (\sign(\tilde{\bmu}) - \sign(\bar\mu))] &\leq 2 \cdot\paren*{2(1-\abs*{\bar{\mu}})}^{k/2} \\
	     &= 4(1-\abs*{\bar{\mu}}) \cdot \paren*{2(1-\abs*{\bar{\mu}})}^{k/2-1}\\
	     &\leq 4(1-\abs*{\bar{\mu}})\cdot (1/2)^{k/2 - 1}.
\end{align*}
Hence, for \Cref{eq:agnostic-loss-main-bound} to hold in this regime, it suffices $k = \Omega(\log(1/(\eps + \alpha)))$, whereas we are picking the much larger and sufficient choice of $k = \Omega(1/(\eps + \alpha)^2)$.

\pparagraph{Middle regime:} If $\min(\eps + \alpha/2, 1/2) \leq \abs{\bar{\mu}} \leq \frac{3}{4}$, we bound each term within the expectation (using that $\abs*{\sign(a) - \sign(b)} \leq 2\cdot  \Ind[|a-b| \geq |a|]$)
\begin{align*}
	     \Ex[\bar{\mu} \cdot (\sign(\tilde{\bmu}) - \sign(\bar{\mu}))] &\leq  2\cdot\Ex[|\bar{\mu}| \cdot\Ind[|\tilde{\bmu} - \bar{\mu}| \geq |\bar{\mu}|]\\
	     &\leq  2\cdot\Ex\bracket*{\frac{\paren*{\tilde{\bmu} - \bar{\mu}}^2}{|\bar{\mu}|} }.
\end{align*}
Then, since $\tilde{\bmu}$ has mean $\bar{\mu}$ and variance at most $1/k$, we have
\begin{equation*}
     \Ex[\bar{\mu} \cdot (\sign(\tilde{\bmu}) - \sign(\bar{\mu}))] \leq \frac{2}{k |\bar{\mu}|} \leq \frac{2}{k\cdot\min(\eps + \alpha/2, 1/2)}.
\end{equation*}
For any $k \geq \Omega(1/(\eps + \alpha)^2)$, this quantity is at most $\eps + \alpha/4$, which means \Cref{eq:agnostic-loss-main-bound} holds in this regime.
\end{proof}

\subsection{Proof of \Cref{prop:lipschitz-constant-maj}}
\label{subsec:lipschitz-maj}
    We will show that the derivative satisfies $0 \leq M'_k(t) \leq O(\sqrt{k})$. 
    
    It will be more convenient to work with $B_k(p) \coloneqq \Pr[\Bin(k,p) \geq k/2]$. Note this is related to $M_k$ via the relation $M_k(t) = 2\cdot B_k((t+1)/2) -1$ which means that $M_k'(t) = B'_k((t+1)/2)$ so the two functions have the same bounds on their derivatives.
   
    The Russo–Margulis formula (see e.g. \cite{ODBook}) gives that the derivative of $B_k$ can be expressed as
    \begin{equation*}
        B_k'(p) = \sum_{i \in [k]}  \Prx_{\bz \sim \Ber(p)^k}[\Maj(\bz) \neq \Maj(\bz^{\oplus i})]
    \end{equation*}
    where $z^{\oplus i}$ indicates $z$ with the $i^{\text{th}}$ bit flipped, and $\Maj(z_1,\ldots, z_k) = \Ind[z_1 + \cdots + z_k \geq k/2]$. In order for $\Maj(z)$ and $\Maj(z^{\oplus i})$ to differ, the remaining $k-1$ bits must have exactly $(k-1)/2$ many $1$s. It is well known that $\Pr[\Bin(k-1,p) = (k-1)/2]$ is upper bounded by $O(1/\sqrt{k})$ for any choice of $p$. Summing over the $k$ indices in the above equation therefore gives the desired bound. \qed

\section{Connections to multiaccuracy and calibration}
\label{sec:fairness}
\subsection{Multiaccuracy is equivalent to optimal error agnostic learning}
\Cref{thm:agnostic-body} implies a learner for agnostic noise because nasty classification noise is strictly harder than agnostic noise. Here, we give an alternative approach via an \textsl{equivalence} to multiaccuracy.
\begin{definition}[Multiaccuracy, \cite{HKRR18,KGZ19}]
    Let $\mcD_{x,y}$ be a distribution over $X \times \bits$. A hypothesis $\bar{h}:X \to [-1,1]$ is $\tau$-multiaccurate with respect to a concept class $\mcC$ if, for all $c \in \mcC$
    \begin{equation*}
        \Ex_{(\bx,\by) \sim \mcD_{x,y}}[c(\bx) (\by - \bar{h}(\bx))] \leq \tau.
    \end{equation*}
\end{definition}
\cite{KGZ19} justify multiaccuracy as a fairness condition. For example, suppose we wish to ensure that $\overline{h}$ is unbiased on various protected groups $S_1,\ldots, S_{m} \subseteq X$. Then, we can include within $\mcC$ the functions $x \mapsto \Ind[x \in S_i]$ and $x \mapsto -\Ind[x \in S_i]$ for all $i \in [m]$. Then, multiaccuracy requires, for all $i \in [m]$, that\footnote{The guarantee becomes weaker for smaller sets. This is necessary if one wants a sample efficient method to construct $\bar{h}$ because small sets will have fewer samples fall within them.}
\begin{equation*}
    \abs[\big]{\Ex[\bar{h}(\bx) \mid \bx \in S_i] - \Ex[\by \mid \bx \in S_i]} \leq \frac{\tau}{\Pr[\bx \in S_i]}.
\end{equation*}
We show that multiaccuracy is also desirable for its robustness properties.
\begin{claim}[Multiaccuracy is equivalent to agnostic learning with randomized hypothesis]
    \label{claim:multiaccuracy-equiv}
    For any concept class $\mcC$ of $\bits$-valued functions and distribution $\mcD_{x,y}$ over $X \times \bits$, the following properties of a hypothesis $\bar{h}:X \to [-1,1]$ are equivalent.
    \begin{enumerate}
        \item The randomized hypothesis $\bh \coloneqq \Rad(\bar{h})$ satisfies
        \begin{equation*}
            \error_{\mcD_x}(\bh, c) \leq \eta(c) + \eps \quad\quad\text{for all }c \in \mcC,
        \end{equation*}
        where $\mcD_x$ is the marginal distribution of $\mcD_{x,y}$ over $X$ and $\eta(c)$ is the smallest value for which the distribution $\mcD_{x,y}$ can be formed by an $\eta(c)$-agnostic adversary when the target function is $c^{\star} = c$. 
        \item The hypothesis $\bar{h}$ is $(\tau =2\eps)$-multiaccurate w.r.t. $\mcC$.
    \end{enumerate}
\end{claim}
\begin{proof}
    We first analyze $\eta(c)$. For the distribution $\mcD_{x,y}$ to be formed by the agnostic adversary, it must pick the function $\boldf(x)$ whose output has the distribution of $(\by \mid \bx=x)$ when $(\bx,\by) \sim \mcD_{x,y}$. Then,
    \begin{equation*}
        \eta(c) = \error_{\mcD_x}(\boldf, c) = \Prx_{\bx,\by \sim \mcD_{x,y}}[\by \neq c(\bx)].
    \end{equation*}
    Therefore, the first requirement can be written as
    \begin{equation*}
        \Prx_{\bx \sim \mcD_x}[\bh(\bx) \neq c(\bx)] \leq \Prx_{\bx,\by \sim \mcD_{x,y}}[\by \neq c(\bx)] + \eps \quad\quad\text{for all $c \in \mcC$.}
    \end{equation*}
    Since $c(x) \in \bits$ and $\Ex[\bh(x)] = \bar{h}(x)$, the above is equivalent to 
    \begin{equation*}
        \Ex_{\bx \sim \mcD_x}\bracket*{\frac{1 - c(\bx) \bar{h}(\bx)}{2}} \leq \Ex_{\bx,\by \sim \mcD_{x,y}}\bracket*{\frac{1 - c(\bx) \by}{2}} + \eps \quad\quad\text{for all $c \in \mcC$.}
    \end{equation*}
    The above rearranges to exactly the multiaccuracy requirement,
    \begin{equation*}
        \Ex_{(\bx,\by) \sim \mcD_{x,y}}[c(\bx) (\by - \bar{h}(\bx))] \leq 2\eps \quad\quad\text{for all $c \in \mcC$.}\qedhere
    \end{equation*}
\end{proof}

There are various algorithms for constructing multiaccurate learners (see e.g. \cite{KGZ19}). Due to \Cref{claim:multiaccuracy-equiv}, these can be used to derive an optimal error agnostic learner. Such learners can then be converted to optimal error nasty classification noise learners using the equivalence of \cite{BV25}, although with sample complexity overheads. This gives an alternative method to derive a less efficient learner with the accuracy guarantee of \Cref{thm:agnostic-body}. We remark that even for agnostic noise, \Cref{thm:agnostic-body} can still give a better sample complexity than approaches using multiaccuracy in the relative error setting.

\subsection{Applications of calibrated multiaccuracy}
A second popular property of hypotheses is \textsl{calibration} \cite{Daw82}. This property is desirable because, roughly speaking, many types of decision making algorithms can treat the output of a calibrated predictor as ``true" probabilities without sacrificing much utility (see e.g. \cite{HW24}). 
\begin{definition}[Calibrated hypothesis]
    We say a hypothesis $\bar{h}:X \to [-1,1]$ is $\tau$-calibrated for a distribution $\mcD_{x,y}$ if for any output value $v \in [-1,1]$ that $\bar{h}$ takes on with nonzero probability,
    \begin{equation*}
        \abs*{\Ex_{\bx,\by}[\by \mid \bar{h}(\bx) = v] - v} \leq \tau.
    \end{equation*}
\end{definition}

We show below, in \Cref{claim:cal-ma}, that a calibrated and multiaccurate hypothesis can be post-processed into a hypothesis meeting the criteria of \Cref{lem:mal-error-loss-overview}. Since there are a variety of techniques for constructing such a hypothesis with an ERM oracle (see e.g. \cite{HKRR18}), this can be used in place of \Cref{thm:general-optimization} in our proofs of \Cref{thm:malicious-intro,thm:nasty-dist-free-upper-intro}. 

That said, this alternative approach is in a sense ``overkill" and requires a larger sample/oracle complexity than ours (see e.g. the sample complexity lower bound of \cite{GT25}). It is known that calibrated and multiaccurate predictors can be post-processed to have low loss for a variety of losses \cite{GHKRW23,CGKR25}, though \Cref{claim:cal-ma} does not appear to easily follow black-box from that approach.\footnote{Briefly, calibrated and multiaccurate predictors are \textsl{omnipredictors} \cite{GKRVW21}, which allow for minimizing any loss that can be decomposed as $\ell_S(c,h) = \Ex_{(\bx,\by) \sim \Unif(S)}[L(h(\bx),\by) - L(c(\bx),\by)]$. Our loss has a term measuring the correlation between $h(\bx)$ and $c(\bx)$ and so is not of this form.} Our method of generating a low loss hypothesis in the proof of \Cref{thm:malicious-intro} can be viewed as a demonstration that if we only care about one loss, then a bespoke approach is more efficient. \Cref{thm:general-optimization} may be useful for designing such a bespoke approach for other losses as well.

\begin{claim}[Using calibration and multiaccuracy to find a low-loss hypothesis]
    \label{claim:cal-ma}
    For any sample $S$, define $\mcD_{x,y} \coloneqq \Unif(S)$ and define $\mu:X \to [-1,1]$ to be any $\tau$-calibrated and multiaccurate hypothesis on $\mcD_{x,y}$ with respect to a concept class $\mcC$. Then, for any bijective $g$ satisfying the requirements of \Cref{lem:intro-g-for-dist} the loss of $\overline{h} \coloneqq g \circ \mu$ as defined in \Cref{eq:mal-loss-intro} is at most $O(\tau)$.
\end{claim}
\begin{proof}
    Our goal is to show that for any $c \in \mcC$,
    \begin{equation*}
        \ell_S(c, \bar{h}) = \Ex_{\bx,\by \sim \Unif(S)}[2c(\bx)\by - \bar{h}(\bx)c(\bx) -\bar{h}(\bx)\by] \leq O(\tau).
    \end{equation*}
    The guarantee that $g$ satisfies \Cref{lem:intro-g-for-dist} means, for any $z \in \bits$ and $\bar{\mu} \in [-1,1]$ that $$2\bar{\mu} z - g(\bar{\mu})(\bar{\mu} + z)\leq 0.$$
    Since $c(x) \in \bits$, this means that
    \begin{equation}
        \label{eq:g-implication}
        \Ex[2 \mu(\bx) c(\bx) - \bar{h}(\bx) (\mu(\bx) + c(\bx))] \leq 0.
    \end{equation}
    We will use this in our analysis of the loss and so rewrite it as
    \begin{equation*}
        \ell_S(c, \bar{h}) = \Ex[2 \mu(\bx) c(\bx) - \bar{h}(\bx) (\mu(\bx) + c(\bx))] + 2 \Ex[c(\bx)(\by - \mu(\bx))] + \Ex[\bar{h}(\bx)(\mu(\bx) - \by)].
    \end{equation*}
    The first term is $\leq 0$ from \Cref{eq:g-implication}. The second term is $\leq 2 \tau$ from the multiaccuracy constraint. We'll bound the third term using calibration. Let $V \coloneqq \set{\bar{h}(x) : (x,y) \in S}$ be the set containing all values that $\bar{h}$ takes on. Then,
    \begin{align*}
        \Ex[\bar{h}(\bx)(\mu(\bx) - \by)] &= \sum_{v \in V}    \Ex[\bar{h}(\bx)(\mu(\bx) - \by) \mid \bar{h}(\bx) = v] \cdot \Pr[\bar{h}(\bx) = v] \\
        &\leq \max_{v \in V}\Ex[\bar{h}(\bx)(\mu(\bx) - \by) \mid \bar{h}(\bx) = v].
    \end{align*}
    Since $g$ is bijective and $\bar{h}(x) = g(\mu(x))$, the condition $\bar{h}(x) = v$ is equivalent to $\mu(x) = g^{-1}(v)$. We can then write,
    \begin{equation*}
        \Ex[\bar{h}(\bx)(\mu(\bx) - \by) \mid \bar{h}(\bx) = v] = v \cdot \Ex[\mu(\bx) - \by \mid \mu(\bx) = g^{-1}(v)] \leq \tau |v|,
    \end{equation*}
    where the last inequality uses the calibrating constraint. The desired bound then follows from $|v| \leq 1$.
\end{proof}

\section{Impossibility results}

\subsection{The necessity of improper learners, proof of \Cref{claim:improper}}
We prove the following, restated for convenience.
\improper*
Both of our constructions in this section will use the concept class over domain $[k]$ that consists of exactly-$1$ sparse functions
\begin{equation*}
    \mcS_k \coloneqq \set{c: c(x) = 1\text{ for exactly one choice of }x \in [k]}.
\end{equation*}
This concept class has VC dimension $1$. We begin by showing that achieving malicious noise error better than $\eta$ (such as the guarantee of \Cref{thm:malicious-intro}) is not doable with a proper learner. All of our lower bounds will be for expected error.
\begin{claim}[Improper learners are necessary for malicious noise]
    \label{claim:improper-mal}
    Let $\eta = 1/k$ for any integer $k$. There is no algorithm that learns $\mcS_k$ with $\eta$-malicious noise and always outputs a hypothesis supported on $\mcS_k$ with expected error better than $\eta$.
\end{claim}
\begin{proof}
    Let $\mcD_{x,y}$ be the distribution of $(\bx,\by)$ where $\bx \sim \Unif([k])$ and $\by = -1$ always. For any algorithm $A$, we consider the average output of $A$ given an i.i.d. sample from $\mcD_{x,y}$,
    \begin{equation*}
        \bh = \Ex_{\bS \iid \mcD_{x,y}}[A(\bS)].
    \end{equation*}
    Since $A$ outputs hypotheses supported on $\mcS_k$, that is also true of $\bh$. We can therefore write $\bh = w_1 c_1 + \cdots + w_k c_k$ where $c_i$ is the hypothesis that outputs $1$ only on input $i$. These weights satisfy that $w_1 + \cdots + w_k = 1$, so there is some $i^{\star}$ where $w_{i^{\star}} \leq 1/k$.

    We now construct a malicious adversary that makes samples from $c_{i^\star}$ look like $\mcD_{x,y}$. Let the distribution over unlabeled examples, $\mcD_x$ be $\Unif([k] \setminus \set{i^\star})$. Then, the malicious adversary can create i.i.d. examples from $\mcD_{x,y}$ given base distribution $\mcD_x$ and target $c_{i^\star}$ if, whenever it corrupts a point, it sets it to $(x' = i^{\star}, y' = -1)$. In this setting, $A$ on average outputs $\bh$.

    Finally, we compute the accuracy of $\bh$. For any $j \neq i^{\star}$, we have that
    \begin{equation*}
        \error_{\mcD_x}(c_j, c_i^{\star}) = \frac{1}{k-1},
    \end{equation*}
    because the two functions differ on exactly one point on $[k] \setminus \set{i^\star}$ (the point $x = j$), and $\mcD_x$ puts weight $1/(k-1)$ on that point. For $j = i^{\star}$, we simply have $\error_{\mcD_x}(c_j, c_i^{\star})  = 0$. Hence,
    \begin{equation*}
        \error_{\mcD_x}(\bh, c_{i}^{\star}) = \sum_{j \neq i^{\star}} w_j \cdot \frac{1}{k-1} = \frac{1-w_{i^{\star}}}{k-1} \geq \frac{1 - 1/k}{k-1} = 1/k = \eta. \qedhere
    \end{equation*}
\end{proof}

    Our second construction proves a lower bound for fixed-distribution agnostic noise. Since agnostic noise is easier than nasty classification noise, this implies that the guarantee of \Cref{thm:agnostic-intro} is not possible. Since nasty classification noise is easier than nasty noise, and our lower bound is in the fixed-distribution setting, this similarly implies the error guarantees of \Cref{thm:nasty-dist-free-upper-intro,thm:nasty-fixed-dist-intro} are not possible with proper learners.

\begin{claim}[Improper learners are necessary for agnostic noise even over the uniform distribution]
    \label{claim:improper-agnostic}
    Let $\eta = 1/k$ for any integer $k$. There is no algorithm that learns $\mcS_k$ over the distribution $\mcD_x = \Unif([k])$ with $\eta$-agnostic noise and always outputs a hypothesis supported on $\mcS_k$ with expected error better than $2\eta(1-\eta)$.
\end{claim}
\begin{proof}
    We consider the same distribution $\mcD_{x,y}$ as in the proof of \Cref{claim:improper-mal} and averaged hypothesis $\bh$. Once again, we can write $\bh = w_1 c_1 + \cdots + w_k c_k$ and there is some $i^{\star}$ for which $w_{i^{\star}} \leq 1/k$.

    Consider the agnostic adversary that selects the function $f$ that always outputs $-1$. This function satisfies that $\error_{\mcD_x}(f, c_i) = \eta$ for all $i \in [k]$, so is a legal choice when $c_{i^{\star}}$ is the target function. With this choice, the learner receives i.i.d. samples from $\mcD_{x,y}$ and so outputs $\bh$.

    To compute the error of $\bh$, we first observe that, for any $j \neq i^{\star}$, 
    \begin{equation*}
        \error_{\mcD_x}(c_j, c_i^{\star}) = \frac{2}{k} = 2\eta,
    \end{equation*}
    since the two concepts differ on exactly two points. Therefore,
     \begin{equation*}
        \error_{\mcD_x}(\bh, c_{i}^{\star}) = \sum_{j \neq i^{\star}} w_j \cdot 2\eta = 2\eta (1 - w_{i^{\star}}) \geq 2\eta(1 - 1/k)= 2\eta(1-\eta). \qedhere
    \end{equation*}
\end{proof}

\subsection{The necessity of distinct learners, proof of \Cref{claim:distinct}}
We prove the following, restated for convenience.
\distinct*
\begin{proof}
    Let the domain be $X \coloneqq \set{0,1}$ and $\mcC = \set{c_{-1}, c_1}$ consist of two functions,
    \begin{equation*}
        c_{-1}(x) \coloneqq\begin{cases}
            -1&\text{if }x=0,\\
            -1&\text{if }x=1.
        \end{cases} \quad\quad\text{and}\quad\quad 
        c_1(x) \coloneqq\begin{cases}
            -1&\text{if }x=0,\\
            1&\text{if }x=1.
        \end{cases} 
    \end{equation*}

    Fix any $p \leq 1$ and learning algorithm $A$. Let 
    \begin{equation*}
        q \coloneqq \Prx_{\bS \iid \mcD_{x,y},\bh \leftarrow A(\bS)}[\bh(1)=1].
    \end{equation*} 
    where $\mcD_{x,y}$ is defined as
    \begin{equation*}
        \mcD_{x,y}(x,y) = \begin{cases}
            1-p &\text{if }x=0,y=-1\\
            0&\text{if }x=0,y=1\\
            p/3 &\text{if }x=1,y=-1\\
            2p/3&\text{if }x=1,y=1
        \end{cases}
    \end{equation*}
    We will show that regardless of what $q$ is, $A$ will have too much error on one of learning with malicious noise or agnostic noise (recall that agnostic noise is easier than nasty classification noise, so the error guarantee of \Cref{thm:agnostic-intro} covers it as well).
    
    First, we consider $(\eta_{\mathrm{mal}} = p/3)$-malicious noise. Let the unlabeled distribution $\mcD_x$ be the distribution that outputs $0$ with probability $(1-p)/(1-p/3)$ and output $1$ with probability $(2p/3)/(1-p/3)$. We observe that if the target function $c^{\star} = c_1$, then the malicious adversary can create a $\bS$ that is drawn i.i.d. from $\mcD_{x,y}$ by always choosing the point $(x'=1, y'=-1)$ whenever it makes a corruption. Then, the average error of the algorithm $A$ is at least $(1-q) \cdot \frac{2p/3}{1-p/3}$. In order to satisfy the error guarantee of \Cref{thm:malicious-intro}, this must satisfy
    \begin{equation*}
        (1-q) \cdot \frac{2p/3}{1-p/3} \leq \frac{1}{2} \cdot \frac{p/3}{1-p/3} + \eps.
    \end{equation*}
    This simplifies to the requirement that
    \begin{equation}
        \label{eq:q-large}
        q \geq \frac{3}{4} - \eps \cdot \frac{1 - p/3}{2p/3}.
    \end{equation}

    Next consider $(\eta_{\mathrm{agn}} = 2p/3)$-agnostic noise. Let the unlabeled distribution $\mcD_x$ be the distribution that outputs $0$ with probability $1-p$ and $1$ with probability $p$. We observe that if the target function is $c_{-1}$, the agnostic adversary can create a dataset that is drawn i.i.d. from $\mcD_{x,y}$ by choosing the corrupted function $\boldf$ that on input $0$ always outputs $-1$ and on input $1$ outputs $1$ with probability $2/3$. In this case, the error of $A$ is $pq$. In order to satisfy the error guarantee of \Cref{thm:agnostic-intro}, this must satisfy
    \begin{equation*}
        pq \leq \frac{2p}{3} + \eps,
    \end{equation*}
    which simplifies to
    \begin{equation}
        \label{eq:q-small}
        q \leq \frac{2}{3} + \eps/p.
    \end{equation}
    This construction works for any choice of $p$. Setting $p = 1$, it is impossible to satisfy both \Cref{eq:q-large,eq:q-small} with any $\eps < 1/24$.
\end{proof}

\section{Acknowledgements}
I thank Li-Yang Tan, S\'ilvia Casacuberta, Aaron Roth, and Michael Kim for their very helpful discussions and feedback. This research was supported by NSF awards 1942123, 2211237, and 2224246 as well as Omer Reingold's Simons Foundation investigators award.

\subsection{AI disclosure}
We used ChatGPT-5 primarily for minor editing, literature search, and the creation of figures. Two more specific uses: It noticed a connection to the Frank-Wolfe algorithm (see \Cref{remark:AI-FW} in \Cref{sec:general-opt} for details). Second, we used it to write much of a straightforward but tedious proof in \Cref{appendix:g-choice} (see the discussion in that section for details).

\bibliographystyle{alpha}
\bibliography{ref}

@STRING{acm = {ACM Press}}

@string{colt = "IEEE Conference on Computational Learning Theory"}

@string{isaac = aisaac}

@STRING{jacm = {Journal of the ACM}}

@STRING{ml = {Machine Learning}}

@string{sigact = "SIGACT News"}

@string{springer = "Springer Verlag"}

@string{stoc = astoc}

@STRING{stoc1984 = {Proc.\ 16th Annual ACM Symposium on Theory of Computing (STOC)}}

@STRING{stoc1993 = {Proc.\ 25th Annual ACM Symposium on Theory of Computing (STOC)}}

@STRING{stoc1998 = {Proc.\ 30th Annual ACM Symposium on Theory of Computing (STOC)}}

@string{wiley = "John Wiley and Sons"}

@article{BEK02,
  title={{PAC} learning with nasty noise},
  author={Bshouty, Nader H and Eiron, Nadav and Kushilevitz, Eyal},
  journal={Theoretical Computer Science},
  volume={288},
  number={2},
  pages={255--275},
  year={2002},
  publisher={Elsevier}
}

@article{DKKLMS19,
  title={Robust estimators in high-dimensions without the computational intractability},
  author={Diakonikolas, Ilias and Kamath, Gautam and Kane, Daniel and Li, Jerry and Moitra, Ankur and Stewart, Alistair},
  journal={SIAM Journal on Computing},
  volume={48},
  number={2},
  pages={742--864},
  year={2019},
  publisher={SIAM}
}

@article{Hau92,
title = {Decision theoretic generalizations of the PAC model for neural net and other learning applications},
journal = {Information and Computation},
volume = {100},
number = {1},
pages = {78-150},
year = {1992},
author = {David Haussler},
}

@article{KL93,
  title={Learning in the presence of malicious errors},
  author={Kearns, Michael and Li, Ming},
  journal={SIAM Journal on Computing},
  volume={22},
  number={4},
  pages={807--837},
  year={1993},
}

@article{KSS94,
author={Michael Kearns and Robert Schapire and Linda Sellie},
title={{Toward efficient agnostic learning}},
journal={Machine Learning},
year={1994},
volume={17},
number={2/3},
pages={115-141}
}

@inproceedings{Val85,
  author    = {Leslie G. Valiant},
  title     = {Learning Disjunction of Conjunctions},
  booktitle = {Proceedings of the 9th International Joint Conference on Artificial Intelligence (IJCAI)},
  pages     = {560--566},
  year      = {1985},
}

@INPROCEEDINGS{Val84,
  author = {L. G. Valiant},
  title = {A theory of the learnable},
  booktitle = stoc1984,
  year = {1984},
  pages = {436--445},
  publisher = acm,
  printoutnum = {271},
}

@inproceedings{DKPP22,
  title={Streaming algorithms for high-dimensional robust statistics},
  author={Diakonikolas, Ilias and Kane, Daniel M and Pensia, Ankit and Pittas, Thanasis},
  booktitle={International Conference on Machine Learning},
  pages={5061--5117},
  year={2022},
  organization={PMLR}
}

@article{LMN93,
  author={N. Linial and Y. Mansour and N. Nisan},
  title={Constant depth circuits, {F}ourier transform and learnability},
  journal={Journal of the ACM},
  year={1993},
  volume={40},
  number={3},
  pages={607-620}
}

@book{DK23book,
  title={Algorithmic high-dimensional robust statistics},
  author={Diakonikolas, Ilias and Kane, Daniel M},
  year={2023},
  publisher={Cambridge university press}
}

@inproceedings{BLMT22,
  title={On the power of adaptivity in statistical adversaries},
  author={Blanc, Guy and Lange, Jane and Malik, Ali and Tan, Li-Yang},
  booktitle={Conference on Learning Theory},
  pages={5030--5061},
  year={2022},
  organization={PMLR}
}

@inproceedings{DKKLMS18,
  title={Robustly learning a gaussian: Getting optimal error, efficiently},
  author={Diakonikolas, Ilias and Kamath, Gautam and Kane, Daniel M and Li, Jerry and Moitra, Ankur and Stewart, Alistair},
  booktitle={Proceedings of the Twenty-Ninth Annual ACM-SIAM Symposium on Discrete Algorithms},
  pages={2683--2702},
  year={2018},
  organization={SIAM}
}

@article{KKMS08,
  title={Agnostically learning halfspaces},
  author={Kalai, Adam Tauman and Klivans, Adam R and Mansour, Yishay and Servedio, Rocco A},
  journal={SIAM Journal on Computing},
  volume={37},
  number={6},
  pages={1777--1805},
  year={2008},
  publisher={SIAM}
}

@article{KLS09,
author={A. Klivans and P. Long and R. Servedio},
  title={{Learning Halfspaces with Malicious Noise}},
  journal={Journal of Machine Learning Research},
  volume={10},
  year={2009},
  pages={2715-2740}
}

@article{KK09,
  title={Potential-based agnostic boosting},
  author={Kanade, Varun and Kalai, Adam},
  journal={Advances in neural information processing systems},
  volume={22},
  year={2009}
}

@article{F10,
  title={Distribution-specific agnostic boosting},
  author={Feldman, Vitaly},
  journal={Innovations in Computer Science},
  year={2010}
}

@article{LS11,
title="Learning large-margin halfspaces with more malicious noise",
author="P. Long and R. Servedio",
journal="NIPS",
year=2011
}

@article{V28,
  title={Zur theorie der gesellschaftsspiele},
  author={von Neumann, John},
  journal={Mathematische annalen},
  volume={100},
  number={1},
  pages={295--320},
  year={1928},
  publisher={Springer}
}

@article{CBNDFSS99,
  title={Sample-efficient strategies for learning in the presence of noise},
  author={Cesa-Bianchi, Nicolo and Dichterman, Eli and Fischer, Paul and Shamir, Eli and Simon, Hans Ulrich},
  journal={Journal of the ACM (JACM)},
  volume={46},
  number={5},
  pages={684--719},
  year={1999},
  publisher={ACM New York, NY, USA}
}

@article{KS94,
  title={Efficient distribution-free learning of probabilistic concepts},
  author={Kearns, Michael J and Schapire, Robert E},
  journal={Journal of Computer and System Sciences},
  volume={48},
  number={3},
  pages={464--497},
  year={1994},
  publisher={Elsevier}
}

@inproceedings{BHMS26,
  title={Is nasty noise actually harder than malicious noise?},
  author={Blanc, Guy and Huang, Yizhi and Malkin, Tal and Servedio, Rocco A},
  booktitle={Proceedings of the 2026 Annual ACM-SIAM Symposium on Discrete Algorithms (SODA)},
  pages={6767--6787},
  year={2026},
  organization={SIAM}
}

@inproceedings{BV25,
  title={Adaptive and oblivious statistical adversaries are equivalent},
  author={Blanc, Guy and Valiant, Gregory},
  booktitle={Proceedings of the 57th Annual ACM Symposium on Theory of Computing},
  pages={2031--2042},
  year={2025}
}

@inproceedings{HKRR18,
  title={Multicalibration: Calibration for the (computationally-identifiable) masses},
  author={H{\'e}bert-Johnson, Ursula and Kim, Michael and Reingold, Omer and Rothblum, Guy},
  booktitle={International Conference on Machine Learning},
  pages={1939--1948},
  year={2018},
  organization={PMLR}
}

@inproceedings{KGZ19,
  title={Multiaccuracy: Black-box post-processing for fairness in classification},
  author={Kim, Michael P and Ghorbani, Amirata and Zou, James},
  booktitle={Proceedings of the 2019 AAAI/ACM Conference on AI, Ethics, and Society},
  pages={247--254},
  year={2019}
}

@article{GT25,
  title={Sample-Efficient Omniprediction for Proper Losses},
  author={Gibbs, Isaac and Tibshirani, Ryan J},
  journal={arXiv preprint arXiv:2510.12769},
  year={2025}
}

@article{FS97,
  title={A decision-theoretic generalization of on-line learning and an application to boosting},
  author={Freund, Yoav and Schapire, Robert E},
  journal={Journal of computer and system sciences},
  volume={55},
  number={1},
  pages={119--139},
  year={1997},
  publisher={Elsevier}
}

@inproceedings{BDLP08,
  title={Does Unlabeled Data Provably Help? Worst-case Analysis of the Sample Complexity of Semi-Supervised Learning.},
  author={Ben-David, Shai and Lu, Tyler and P{\'a}l, D{\'a}vid},
  booktitle={COLT},
  pages={33--44},
  year={2008}
}

@article{BB10,
  title={A discriminative model for semi-supervised learning},
  author={Balcan, Maria-Florina and Blum, Avrim},
  journal={Journal of the ACM (JACM)},
  volume={57},
  number={3},
  pages={1--46},
  year={2010},
  publisher={ACM New York, NY, USA}
}

@inproceedings{J13,
  title={Revisiting Frank-Wolfe: Projection-free sparse convex optimization},
  author={Jaggi, Martin},
  booktitle={International conference on machine learning},
  pages={427--435},
  year={2013},
  organization={PMLR}
}

@article{FW56,
  title={An algorithm for quadratic programming},
  author={Frank, Marguerite and Wolfe, Philip},
  journal={Naval research logistics quarterly},
  volume={3},
  number={1-2},
  pages={95--110},
  year={1956},
  publisher={Wiley Online Library}
}

@book{ODBook,
  title={Analysis of boolean functions},
  author={O'Donnell, Ryan},
  year={2014},
  publisher={Cambridge University Press}
}

@article{BM02,
  title={Rademacher and gaussian complexities: Risk bounds and structural results},
  author={Bartlett, Peter L and Mendelson, Shahar},
  journal={Journal of machine learning research},
  volume={3},
  number={Nov},
  pages={463--482},
  year={2002}
}

@article{LLS01,
  title={Improved bounds on the sample complexity of learning},
  author={Li, Yi and Long, Philip M and Srinivasan, Aravind},
  journal={Journal of Computer and System Sciences},
  volume={62},
  number={3},
  pages={516--527},
  year={2001},
  publisher={Elsevier}
}

@book{Vid13,
  title={Learning and generalisation: with applications to neural networks},
  author={Vidyasagar, Mathukumalli},
  year={2013},
  publisher={Springer Science \& Business Media}
}

@article{MP96,
  author={Y. Mansour and M. Parnas},
  title={Learning conjunctions with noise under product distributions},
  journal={Information Processing Letters},
  year={1998},
  volume={68},
  number={4},
  pages={189--196}
}

@article{Aue97,
  author={P. Auer},
  title={Learning nested differences in the presence of malicious noise},
  journal={Theor. Comp. Sci.},
  year={1997},
  volume={185},
  number={1},
  pages={159-175}
}

@INPROCEEDINGS{Bsh98,
  author = {Nader H. Bshouty},
  title = {A New Composition Theorem for Learning Algorithms},
  booktitle = stoc1998,
  year = {1998},
  pages = {583--589},
}

@article{Ser03,
  author={R. Servedio},
  title={Smooth boosting and learning with malicious noise},
  journal={Journal of Machine Learning Research},
  pages={633-648},
  volume={4},
  year={2003},
}

@ARTICLE{AW98,
  author = {Peter Auer and Manfred K. Warmuth},
  title = {Tracking the best disjunction},
  journal = ml,
  year = {1998},
  volume = {32},
  pages = {127--150},
  number = {2},
  printoutnum = {023},
}

@inproceedings{BS12,
  author       = {Aharon Birnbaum and
                  Shai Shalev{-}Shwartz},
  title        = {Learning Halfspaces with the Zero-One Loss: Time-Accuracy Tradeoffs},
  booktitle    = {Advances in Neural Information Processing Systems 25: 26th Annual
                  Conference on Neural Information Processing Systems},
  pages        = {935--943},
  year         = {2012}
  }

@article{ABL17,
  author       = {Pranjal Awasthi and
                  Maria{-}Florina Balcan and
                  Philip M. Long},
  title        = {The Power of Localization for Efficiently Learning Linear Separators
                  with Noise},
  journal      = {J. {ACM}},
  volume       = {63},
  number       = {6},
  pages        = {50:1--50:27},
  year         = {2017}
  }

@inproceedings{DKS18,
  author       = {Ilias Diakonikolas and
                  Daniel M. Kane and
                  Alistair Stewart},
  title        = {Learning geometric concepts with nasty noise},
  booktitle    = {Proceedings of the 50th Annual {ACM} {SIGACT} Symposium on Theory of Computing (STOC)},
  pages        = {1061--1073},
  year         = {2018}
  }

@inproceedings{SZ21,
  author       = {Jie Shen and
                  Chicheng Zhang},
  title        = {Attribute-Efficient Learning of Halfspaces with Malicious Noise: Near-Optimal
                  Label Complexity and Noise Tolerance},
  booktitle    = {Algorithmic Learning Theory, 16-19 March 2021, Virtual Conference,
                  Worldwide},
  volume       = {132},
  pages        = {1072--1113},
  year={2021}
  }

@InProceedings{Shen23,
  title = 	 {{PAC} Learning of Halfspaces with Malicious Noise in Nearly Linear Time},
  author =       {Shen, Jie},
  booktitle = 	 {Proceedings of The 26th International Conference on Artificial Intelligence and Statistics},
  pages = 	 {30--46},
  year = 	 {2023}}

@article{HKLM24,
  author       = {Max Hopkins and
                  Daniel M. Kane and
                  Shachar Lovett and
                  Gaurav Mahajan},
  title        = {Realizable Learning is All You Need},
  journal      = {TheoretiCS},
  volume       = {3},
  year         = {2024}
  }

@inproceedings{KSTV25,
  title={The Power of Iterative Filtering for Supervised Learning with (Heavy) Contamination},
  author={Klivans, Adam and Stavropoulos, Konstantinos and Tian, Kevin and Vasilyan, Arsen},
  booktitle={The Thirty-ninth Annual Conference on Neural Information Processing Systems},
  year={2025}
}

@inproceedings{GHKRW23,
  title={Loss Minimization Through the Lens Of Outcome Indistinguishability},
  author={Gopalan, Parikshit and Hu, Lunjia and Kim, Michael P and Reingold, Omer and Wieder, Udi},
  booktitle={14th Innovations in Theoretical Computer Science Conference (ITCS 2023)},
  pages={60--1},
  year={2023},
  organization={Schloss Dagstuhl--Leibniz-Zentrum f{\"u}r Informatik}
}

@article{CGKR25,
  title={How global calibration strengthens multiaccuracy},
  author={Casacuberta, S{\'\i}lvia and Gopalan, Parikshit and Kanade, Varun and Reingold, Omer},
  journal={arXiv preprint arXiv:2504.15206},
  year={2025}
}

@article{Daw82,
  title={The well-calibrated Bayesian},
  author={Dawid, A Philip},
  journal={Journal of the American statistical Association},
  volume={77},
  number={379},
  pages={605--610},
  year={1982},
  publisher={Taylor \& Francis}
}

@article{HW24,
  title={Calibration error for decision making},
  author={Hu, Lunjia and Wu, Yifan},
  journal={arXiv preprint arXiv:2404.13503},
  year={2024}
}

@INPROCEEDINGS{Kha93,
  author = {Michael Kharitonov},
  title = {Cryptographic hardness of distribution-specific learning},
  booktitle = stoc1993,
  year = {1993},
  pages = {372--381},
  publisher = acm,
  printoutnum = {203},
}

@article{SSS95,
  title={Chernoff--Hoeffding bounds for applications with limited independence},
  author={Schmidt, Jeanette P and Siegel, Alan and Srinivasan, Aravind},
  journal={SIAM Journal on Discrete Mathematics},
  volume={8},
  number={2},
  pages={223--250},
  year={1995},
  publisher={SIAM}
}

@article{Vad12,
  title={Pseudorandomness},
  author={Vadhan, Salil P},
  journal={Foundations and Trends{\textregistered} in Theoretical Computer Science},
  volume={7},
  number={1-3},
  pages={1--336},
  year={2012},
  publisher={Emerald Publishing Limited}
}

@article{Jof74,
author = {Anatole Joffe},
title = {{On a Set of Almost Deterministic $k$-Independent Random Variables}},
volume = {2},
journal = {The Annals of Probability},
number = {1},
publisher = {Institute of Mathematical Statistics},
pages = {161 -- 162},
year = {1974},
doi = {10.1214/aop/1176996762},
URL = {https://doi.org/10.1214/aop/1176996762}
}

@inproceedings{GKRVW21,
  title={Omnipredictors},
  author={Gopalan, Parikshit and Kalai, Adam Tauman and Reingold, Omer and Sharan, Vatsal and Wieder, Udi},
  booktitle={13th Innovations in Theoretical Computer Science Conference (ITCS 2022)},
  pages={79--1},
  year={2022},
  organization={Schloss Dagstuhl--Leibniz-Zentrum f{\"u}r Informatik}
}

\appendix
\crefalias{section}{appendix}
\crefname{appendix}{appendix}{appendices}
\Crefname{appendix}{Appendix}{Appendices}

\section{Deferred proofs}

\subsection{Randomized rounding when there are no heavy points}

We show that any randomized hypothesis can be efficiently converted to a deterministic hypothesis with comparable error when the distribution does not have too much weight on any one point. This gives a post-processing step that can be applied to any of our algorithms outputting a randomized hypothesis to produce one outputting a deterministic hypothesis.
\begin{claim}
\label{claim:randomized-rounding}
    Let $\bh:X \to \bits$ be a randomized hypothesis, $c^{\star}:X \to \bits$ be any target function, and $\mcD$ be a distribution over $X$ satisfying that $\mcD(x) \le p_{\max}$ for all $x \in X$. There is an algorithm $A$ that takes as input $\delta > 0$ as well as $\bh$ represented as a size-$n$ circuit, runs in time $\poly(n, \log(1/\delta))$ and outputs a deterministic hypothesis $\hat{h}$ that with probability at least $1-\delta$ over the randomness of $A$, satisfies
    \begin{equation*}
          \error_{\mcD}(\hat{h}, c^{\star}) \leq \error_{\mcD}(\bh, c^{\star}) + O\paren*{\sqrt{p_{\max} \ln(1/\delta)}}.
    \end{equation*}
\end{claim}
To ensure the algorithm $A$ in \Cref{claim:randomized-rounding} is time-efficient, we will use $k$-wise independent rounding rather than fully independent rounding.
\begin{theorem}[Chernoff bounds under limited independence, \cite{SSS95}]
    \label{thm:limited-independence}
    Let $\bz_1, \ldots, \bz_n$ be $k$-wise independent random variables each bounded on $[0,1]$. For $\bZ \coloneqq \bz_1 + \cdots + \bz_n$ and $\mu \coloneqq \Ex[\bZ]$,
    \begin{equation*}
        \Pr[\bZ \geq (1 + \Delta)\mu] \leq  \max\set*{e^{-\Omega(k)}, e^{-\Omega(\Delta^2 \mu)}, e^{-\Omega(\Delta \mu)} }.
    \end{equation*}
\end{theorem}
Note that the last two terms are what one would get from a standard fully independent Chernoff bound (e.g. \Cref{fact:Chernoff}), whereas the first term is required given only the assumption of $k$-wise independence.

We'll use a classic efficient construction of $k$-wise independent random variables.
\begin{fact}[\cite{Jof74}, Corollary 3.34 of \cite{Vad12}]
    \label{fact:k-wise-construction}
    For any parameters $d,r, k \in \N$ and seed $s \in \zo^{k \cdot \max(d,r)}$ there exists a family of functions $F_s:\zo^d \to \zo^r$ so that, if we draw $\bs \sim \Unif(\zo^{k \cdot \max(d,r)})$ the random variables $\set{F_{\bs}(x)}_{x \in \zo^d}$ are $k$-wise independent and uniform on $\zo^r$. Furthermore, evaluation of $F_s(x)$ given $s$ and $x$ can be done in time $\poly(d,r,k)$.
\end{fact}

\begin{proof}[Proof of \Cref{claim:randomized-rounding}]
    Let $C$ be the circuit computing $\bh$. Then, $C$ takes $d$ bits representing its input $x$ and $r$ bits representing its random seed, where $d + r \leq n$. Viewing $C$ as a function $C:\zo^d \times \zo^r \to \bits$, we have that $\bh(x)$ is equivalent in distribution to\footnote{Implicitly, $C$ views the domain $X$ represented as $\zo^d$.} $C(x, \br)$ where $\br \sim \Unif(\zo^r)$. 
    
    The algorithm $A$ sets $k = O(\log(1/\delta))$ and draws a seed $\bs \sim \Unif(\zo^{k \cdot \max(d,r)})$ and, for the function $F_{\bs}$ defined in \Cref{fact:k-wise-construction}, outputs the function
    \begin{equation*}
        \hat{h}_{\bs}(x) \coloneqq C(x, F_{\bs}(x)).
    \end{equation*}
    Note that $\hat{h}_{\bs}$ is a deterministic function once we hardwire any one choice for $\bs$. What remains is to analyze the error of $h$. For each $x \in X$, define
    \begin{equation*}
        \bz_x \coloneqq \frac{\mcD(x)}{p_{\max}} \cdot \Ind[\hat{h}_{\bs}(x) \neq c^{\star}(x)].
    \end{equation*}
    Then, $\set{\bz_x}_{x \in X}$ are $k$-wise independent random variables each bounded on $[0,1]$ satisfying that
    \begin{equation*}
        \frac{\error_{\mcD}(\hat{h}_{\bs}, c^{\star}) }{p_{\max}} =\bZ = \sum_{x\in X}\bz_x.
    \end{equation*}
    Applying \Cref{thm:limited-independence} with our choice of $k = O(\log(1/\delta))$ along with $\Ex_{\bs}[\error_{\mcD}(\hat{h}_{\bs}, c^{\star})] =\error_{\mcD}(\bh, c^{\star})$ gives that
    \begin{equation*}
        \Prx_{\bs}\bracket*{\frac{\error_{\mcD}(\hat{h}_{\bs}, c^{\star}) }{p_{\max}} \geq \frac{\error_{\mcD}(\bh, c^{\star})(1 + \Delta) }{p_{\max}}} \leq \max\paren*{\delta,e^{-\Omega(\Delta^2 \mu)}, e^{-\Omega(\Delta \mu)}},
    \end{equation*}
    where $\mu \coloneqq \frac{\error_{\mcD}(\bh, c^{\star})}{p_{\max}}$. We can instead write this as
     \begin{equation*}
        \Prx_{\bs}\bracket*{\frac{\error_{\mcD}(\hat{h}_{\bs}, c^{\star}) }{p_{\max}} \geq \frac{\error_{\mcD}(\bh, c^{\star})}{p_{\max}} + \eps} \leq \max\paren*{\delta,e^{-\Omega(\eps^2/\mu)}, e^{-\Omega(\eps)}}.
    \end{equation*}
    This bound improves as $\mu$ gets smaller. We are therefore free to pessimistically set $\mu = 1/p_{\max}$, which is the largest value it can take, to obtain,
    \begin{equation*}
        \Prx_{\bs}\bracket*{\frac{\error_{\mcD}(\hat{h}_{\bs}, c^{\star}) }{p_{\max}} \geq \frac{\error_{\mcD}(\bh, c^{\star})}{p_{\max}} + \eps} \leq \max\paren*{\delta,e^{-\Omega(\eps^2\cdot p_{\max})}, e^{-\Omega(\eps)}}.
    \end{equation*}
    We solve for $\eps$ so this failure probability is $\delta$ giving,
    \begin{equation*}
        \Prx_{\bs}\bracket*{\frac{\error_{\mcD}(\hat{h}_{\bs}, c^{\star}) }{p_{\max}} \geq \frac{\error_{\mcD}(\bh, c^{\star})}{p_{\max}} +O\paren*{\max\set*{\sqrt{\frac{\log(1/\delta)}{p_{\max}}}, \log(1/\delta)}}} \leq \delta.
    \end{equation*}
    Multiplying by $p_{\max}$ we have shown that with probability at least $1-\delta$,
    \begin{equation*}
        \error_{\mcD}(\hat{h}_{\bs}, c^{\star}) \leq \error_{\mcD}(\bh, c^{\star}) +O\paren*{\max\set*{\sqrt{p_{\max}\log(1/\delta)}, p_{\max}\log(1/\delta)}}.
    \end{equation*}
    Finally, we observe that since error is bounded on $[0,1]$, the second term is dominated by the first term in the regime that matters.
\end{proof}

\subsection{Proof of \Cref{claim:maj-suffices}}
\label{appendix:g-choice}

This section contains a tedious but straightforward proof (see also \Cref{fig:g-constraints-mal-intro} for a ``proof by picture"). Much of it was generated by ChatGPT-5.3-Codex with author verification, scaffolding, and editing. Before diving into that AI-generated proof, the author wishes to give some intuition for the particular choice of $g$. In order for \Cref{eq:g-constraints-mal-body} to be satisfied, it must be the case that $g'(0) = 2$. For odd $k$ (see \Cref{subsec:lipschitz-maj} for how to derive this formula),
\begin{equation*}
    M_k'(0) = k \cdot \Pr[\Bin(k-1,1/2) = (k-1)/2]=  \frac{k}{2^{k-1}}\binom{k-1}{(k-1)/2},
\end{equation*}
    The smallest value for which this is at least $2$ is $k=7$. For $k = 7$, we have $M_7'(0) = \frac{35}{16}$, which is too large. In order to make the derivative exactly $2$, we mix in an appropriate amount of the identity function, giving $g(t) = \frac{16}{19}M_7(t) + \frac{3}{19} t$.

While \Cref{fig:g-constraints-mal-intro} may make the desired inequalities appear close at the origin, this is alleviated by a second plot of $g(t)/t$ (see \Cref{fig:g-constraints-divided}), which makes it more clear.
\begin{figure}[htb]
    \centering
    \begin{tikzpicture}
        \begin{axis}[
            width=0.6\textwidth,
            height=0.4\textwidth,
            xmin=-1, xmax=1,
            ymin=0, ymax=2.5,
            xtick={-1,1},
            ytick={0,2.5},
            xlabel={$t$},
            axis lines=middle,
            domain=-0.999:0.999,
            samples=401,
            clip mode=individual,
            legend cell align=left,
            legend style={at={(0.02,0.08)},anchor=south west,fill=white,draw=none,font=\small},
        ]
            \addplot[name path=lower,draw=none] {max(0,min(2/(1+x),2/(1-x)))};
            \addplot[name path=upper,draw=none] {min(2.5,max(2/(1+x),2/(1-x)))};
            \addplot[blue!20,draw=none] fill between[of=lower and upper];

            \addplot[blue!80!black,thick,dotted] {max(0,min(2/(1+x),2/(1-x)))};
            \addplot[blue!80!black,thick,dotted] {min(2.5,max(2/(1+x),2/(1-x)))};

            \addplot[black,very thick] {(38 - 35*x^2 + 21*x^4 - 5*x^6)/19};
            \addlegendentry{Our choice of $g(t)/t$}
        \end{axis}
    \end{tikzpicture}
    \caption{The shaded region depicts the constraints of \Cref{eq:g-constraints-mal-body} on the function $t \mapsto g(t)/t$.}
    \label{fig:g-constraints-divided}
\end{figure}
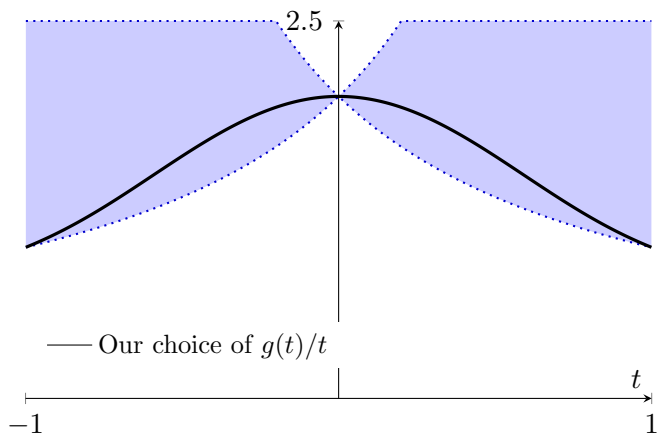

\begin{proof}[Proof of \Cref{claim:maj-suffices}, partially AI generated.]
    For any $k$, the function $M_k$ is odd and can be written as
    \begin{equation*}
        M_k(t) \coloneqq \sum_{j=0}^k \sign(2j-k)\binom{k}{j}\paren[\bigg]{\frac{1+t}{2}}^j\paren[\bigg]{\frac{1-t}{2}}^{k-j},
    \end{equation*}
    since the number of $+1$ values among $\bz_1,\ldots,\bz_k$ is distributed as $\Bin(k,\frac{1+t}{2})$.
    For $k=7$, expanding gives
    \begin{equation*}
        M_7(t)=\frac{35t-35t^3+21t^5-5t^7}{16}.
    \end{equation*}
    As a result, our function $g$ is odd and is equal to the polynomial,
    \begin{equation*}
        g(t)=\frac{16}{19}M_7(t)+\frac{3}{19}t
        =\frac{38t-35t^3+21t^5-5t^7}{19}.
    \end{equation*}
    Since $g$ is odd, it suffices to verify
    \begin{equation*}
        g(t)\ge \frac{2t}{1+t}\quad\text{for all }t\in[-1,1].
    \end{equation*}
    Indeed, if this lower bound holds for every $t$, then applying it to $-t$ gives
    \begin{equation*}
        g(-t)\ge \frac{-2t}{1-t},
    \end{equation*}
    which is equivalent (by oddness) to
    \begin{equation*}
        g(t)\le \frac{2t}{1-t}.
    \end{equation*}
    Hence \Cref{eq:g-constraints-mal-body} follows.

    It remains to prove the lower bound on $[-1,1]$. At $t=-1$ the right-hand side is $-\infty$, so the bound is trivial; thus fix $t\in(-1,1]$. A direct calculation yields
    \begin{align*}
        g(t)-\frac{2t}{1+t}
        &=\frac{t^2(1-t)}{19(1+t)}
        \paren[\big]{5t^5+10t^4-11t^3-32t^2+3t+38}.
    \end{align*}
    We lower bound the last factor by rewriting
    \begin{align*}
        5t^5+10t^4-11t^3-32t^2+3t+38
        &=13+(1-t)\paren[\big]{25+28t-4t^2-15t^3-5t^4}\\
        &=13+(1-t)\paren[\big]{3+(1+t)\paren[\big]{22+6t-10t^2-5t^3}}\\
        &\ge 13,
    \end{align*}
    where we used $t\in[-1,1]$, so
    \begin{equation*}
        22+6t-10t^2-5t^3\ge 22-6-10-5=1.
    \end{equation*}
    Hence $g(t)-\frac{2t}{1+t}\ge 0$ for all $t\in(-1,1]$, and therefore for all $t\in[-1,1]$.

    Finally, $g$ is $2$-Lipschitz. Differentiating the polynomial gives
    \begin{align*}
        g'(t)
        &=\frac{38-105t^2+105t^4-35t^6}{19}
        =\frac{3+35(1-t^2)^3}{19}.
    \end{align*}
    Therefore, for all $t\in[-1,1]$,
    \begin{equation*}
        0\le g'(t)\le \frac{3+35}{19}=2,
    \end{equation*}
   which gives that $g$ is $2$-Lipschitz.
\end{proof}

\end{document}